\documentclass[10pt]{article}

\usepackage[round]{natbib}					% need for \cite
\usepackage{amsmath}    													% need for subequations
\usepackage{graphicx}   													% need for figures
\usepackage{subfigure}														% need for subfigures
\usepackage{amssymb}    													% gives you \mathbb{} font
\usepackage[mathscr]{eucal}												% gives you \mathscr font
\usepackage{lscape}
\usepackage{caption}
\usepackage{cancel}																% gives you the ability to visibly cross out terms in equations}
\usepackage[normalem]{ulem} 																% gives you \sout
\usepackage{pstricks}
\usepackage{rotating}
\usepackage{mdwlist}															% need for \suspend{enumerate} text \resume{enumerate}
\usepackage[paperwidth=8.5in,paperheight=11in,top=1.25in, bottom=1.25in, left=1.00in, right=1.00in]{geometry}
\usepackage{mathtools}														% need for `show only references'
\mathtoolsset{showonlyrefs=true}									% only equations which are labeled AND referenced will be numbered.
\usepackage{amsthm}																% need for theorem-proof environment
\usepackage{hyperref}
\theoremstyle{plain}
\newtheorem{theorem}{Theorem}[section]
\theoremstyle{definition}

\newtheorem{remark}[theorem]{Remark}

\numberwithin{equation}{section}

\newcommand{\<}{\langle} 
\renewcommand{\>}{\rangle}
\renewcommand{\(}{\left(}				
\renewcommand{\)}{\right)}
\renewcommand{\[}{\left[}
\renewcommand{\]}{\right]}			

\newcommand\Eb{\mathbb{E}}
\newcommand\Fb{\mathbb{F}}			
\newcommand\Pb{\mathbb{P}}	
							
\newcommand\Rb{\mathbb{R}}										
\newcommand\Ib{\mathbb{I}}

\newcommand\Ac{\mathscr{A}}
\newcommand\Bc{\mathscr{B}}

\newcommand\Fc{\mathscr{F}}

\newcommand\Mc{\mathscr{M}}		
\newcommand\Nc{\mathscr{N}}
\newcommand\Oc{\mathscr{O}}

\newcommand\eps{\varepsilon}

\newcommand\Om{\Omega}
\newcommand\sig{\sigma}
\newcommand\Lam{\Lambda}
\newcommand\gam{\gamma}

\newcommand\lam{\lambda}
\newcommand\del{\delta}

\newcommand\fh{\widehat{f}}
\newcommand\hh{\widehat{h}}
\newcommand\uh{\widehat{u}}
\newcommand\vh{\widehat{v}}

\newcommand\Et{\widetilde{\Eb}}
\newcommand\Pt{\widetilde{\Pb}}

\newcommand\Wt{\widetilde{W}}
\newcommand\Bt{\widetilde{B}}
\newcommand\Nt{\widetilde{N}}

\renewcommand\d{\partial}

\begin{document}

\title{Exponential L\'evy-type models with stochastic volatility and jump intensity}

\author{
Matthew Lorig
\thanks{
ORFE Department, Sherrerd Hall, Princeton University, Princeton, NJ  08540, United States.
\textbf{E-mail}: mlorig@princeton.edu.
\textbf{Phone}: (609)-258-6758.
\textbf{Fax}: (609)-258-4363.
\textbf{Web}: www.princeton.edu/$\sim$mlorig.
Work partially supported by NSF grant DMS-0739195.
}
\and
Oriol Lozano-Carbass\'e
\thanks{
Bendheim Center of Finance, 26 Prospect Avenue, Princeton University, Princeton, NJ  08540, United States.
}
}

\date{This version: \today}

\maketitle

\begin{abstract}
We consider the problem of valuing a European option written on an asset whose dynamics are described by an exponential L\'evy-type model.  In our framework, both the volatility and jump-intensity are allowed to vary stochastically in time through common driving factors -- one fast-varying and one slow-varying.  Using Fourier analysis we derive an explicit formula for the approximate price of any European-style derivative whose payoff has a generalized Fourier transform; in particular, this includes European calls and puts.  From a theoretical perspective, our results extend the class of multiscale stochastic volatility models of \citet*{fpss} to models of the exponential L\'evy type.  From a financial perspective, the inclusion of jumps and stochastic volatility allow us to capture the term-structure of implied volatility.  To illustrate the flexibility of our modeling framework we extend five exponential L\'evy processes to include stochastic volatility and jump-intensity.  For each of the extended models, using a single fast-varying factor of volatility and jump-intensity, we perform a calibration to the S\&P500 implied volatility surface.  Our results show decisively that the extended framework provides a significantly better fit to implied volatility than both the traditional exponential L\'evy models and the fast mean-reverting stochastic volatility models of \citet{fpss}.
%We consider the problem of valuing a European option written on an asset whose dynamics are described by an exponential L\'evy-type model.  Both the volatility and jump-intensity are allowed to vary stochastically in time through common driving factors -- one fast-varying and one slow-varying.  Our results extend the class of multiscale stochastic volatility models of \citet*{fpss} to models of the exponential L\'evy type.  Using generalized Fourier transform techniques we derive an explicit formula for the approximate price of any European-style derivative whose payoff has a generalized Fourier transform.  We also perform Monte Carlo tests in order to confirm the accuracy of our pricing approximation for European call options. As an example of our framework, we extend the jump-diffusion models of \citet*{merton1976option} to include stochastic volatility and stochastic jump-intensity.  We perform a calibration to S\&P500 options and find that the extended framework with a range of different jump measures and a single fast-varying factor of volatility provides a better fit to implied volatility than both the traditional jump models and the fast mean-reverting stochastic volatility models of \citet{fpss}.
\end{abstract}

\noindent 
{\bf Key words}: multiscale, L\'evy-type process, stochastic volatility, asymptotics, Fourier.

%%%%%%%%%%%%%%%%%%%%%%%%%%%%%%%%%%%%%%%%%%%%%%%%%%%%%%%%%
%
%								Introduction
%
%%%%%%%%%%%%%%%%%%%%%%%%%%%%%%%%%%%%%%%%%%%%%%%%%%%%%%%%%

\section{Introduction}
\label{sec:intro}
An \emph{exponential L\'evy model} is an equity model in which an underlying $S=e^X$ is described by the exponential of a L\'evy process $X$.  Such models extend the geometric Brownian motion description of \citet*{black1973pricing} by allowing the underlying $S$ to experience jumps, the need for which is well-documented in literature (see, \cite{eraker} and references therein).  
In particular, it is known that jumps are required in order to fit the strong skew and smile of implied volatility for short-maturity options (see \citet*{contbook}, Chapter 15).
In addition to allowing the underlying $S$ to jump, exponential L\'evy models are important because they capture many of the stylized features of asset prices, such as heavy tails, high-kurtosis and asymmetry of $\log$ returns.
\par
Several well-known models fit within the exponential L\'evy class: the jump-diffusion model of \citet*{merton1976option}, the pure jump models of \citet*{mandelbrot1963variation}, the variance gamma model of \citet*{madan1998variance}, the extended Koponen family of \citet*{levendorskii} and the double exponential model of \citet*{kou2002jump}.  
The popularity of the above models is at least in part due to their analytic tractability.  Indeed, in \citet*{lewis2001simple,lipton2002}, it is demonstrated that European option prices in all of the above-mentioned models can be computed quickly and easily via (generalized) one-dimensional Fourier transforms.  A comprehensive reference on the subject of option-pricing in an exponential L\'evy setting can be found in  \citet*{levendorskiibook}, as well as Chapter 11 of \citet*{contbook}.
\par
Despite their success, exponential L\'evy models have some shortcomings.  For example, because the $\log$ returns of all exponential L\'evy process are independent and identically distributed, these models cannot exhibit volatility clustering (the tendency for volatility to rise sharply for short periods of time) or the leverage effect (the tendency for volatility to rise when asset prices decline); both of these phenomena are well-documented in time-series literature.  There is also evidence from options markets that exponential L\'evy processes are inadequate.  Indeed, L\'evy-based models cannot fit the term structure of implied volatility; as the maturity date increases the implied volatility surface induced by exponential L\'evy models (unrealistically) flattens.  To capture the implied volatility smile of long-maturity options one requires stochastic volatility.  Another shortcoming of L\'evy processes is that they exhibit constant jump intensities.  However, a recent study of S\&P500 index returns indicates that jump-intensities -- like volatility -- are stochastic (see \citet*{christoffersen2009}).  To address these shortcomings, \citet*{carr2004time} add stochastic volatility (with correlation to the underlying) by stochastically time-changing a L\'evy process.  Notably, the models described in \cite{carr2004time} maintain the analytic tractability that makes the class of exponential L\'evy processes attractive.
\par
In this paper, we address the need for volatility clustering, the leverage effect and stochastic jump intensity by modeling the returns process $X$ by a L\'evy-type process whose local characteristics $(\gam_t,\sig_t,\nu_t)$ are stochastic.  We then use generalized Fourier transform techniques, as well as singular and regular perturbation methods to derive an explicit formula for the approximate price of any European-style derivative whose payoff has a generalized Fourier transform; this includes calls and puts.
\par
From a mathematical perspective, our results are powerful because we extend the multiscale stochastic volatility models of \cite{fpss} to exponential L\'evy models.  Indeed, much like geometric Brownian motion arises as special case of an exponential L\'evy process, the class of fast mean-reverting and multiscale stochastic volatility models considered in \cite{fouque} and \cite{fpss} arise as a special subset of the class of models we consider.  In fact, by removing jumps from our framework, one recovers the Fourier representation of the European option pricing formulas derived in \cite{fouque} and \cite{fpss}.
\par
From a financial perspective, the use of the L\'evy-type models we consider is strongly supported by data.  To be specific, in what follows, we extend five different exponential L\'evy models to include stochastic volatility and jump intensity.  For each of these models, we demonstrate that the extended framework provides significantly better fit to implied volatility than both the traditional exponential L\'evy models and the fast mean-reverting stochastic volatility models of \citet{fpss}.
\par
The rest of this paper proceeds as follows.  In Section \ref{sec:model} we introduce a class of exponential L\'evy-type models in which the volatility and jump-intensity are stochastically driven by a common fast-varying factor.  In Section \ref{sec:pricing} we derive an expression for the approximate price of a European option (Theorem \ref{thm:u0u1}) when the underlying is described by the class of models introduced in Section \ref{sec:model}.  In Section \ref{sec:example}, as an example of our framework, we extend the jump-diffusion model of \citet*{merton1976option} to include stochastic volatility and stochastic jump-intensity.  We also compute (numerically) the implied volatility surface generated by this example.  In Section \ref{sec:calibration}, using a variety of L\'evy measures, we calibrate the extended class of models to the implied volatility surface of S\&P500 options and we compare to the calibration obtained for the corresponding L\'evy models as well as for the fast mean-reverting models of \cite*{fouque}.  In Section \ref{sec:multiscale} we briefly describe how the class of models introduced in Section \ref{sec:model} can be extended to allow for multiple driving factors of volatility and jump-intensity -- one fast-varying factor and one slow-varying factor.  Proofs are provided in an appendix.

%%%%%%%%%%%%%%%%%%%%%%%%%%%%%%%%%%%%%%%%%%%%%%%%%%%%%%%%%
%
%								Model
%
%%%%%%%%%%%%%%%%%%%%%%%%%%%%%%%%%%%%%%%%%%%%%%%%%%%%%%%%%

\section{Stochastic volatility and jump intensity L\'evy-type processes}
\label{sec:model}
Let $(\Om,\Fc,\Pt)$ be a probability space endowed with a filtration $\Fb=\{\Fc_t,t\geq 0\}$, which satisfies the usual conditions.  Here, $\Pt$ is the risk-neutral pricing measure, which we assume is chosen by the market.  The filtration $\Fb$ represents the history of the market.  For simplicity, we assume that the risk-free rate of interest is zero so that all non-dividend paying assets are $(\Pt,\Fb)$-martingales.  All of our results can easily be extended to include constant or deterministic interest rates.
\par
We consider a non-dividend paying asset $S$ whose dynamics under $\Pt$ are described by the following It\^o-L\'evy stochastic differential equation (SDE)
\begin{align}
\left.
\begin{aligned}
dS_t
		&=		\sig(Y_t) S_t \, d\Wt_t + S_{t-} \int_{\Rb} \(e^z - 1\) \, d\Nt_t(Y_t,dz) , &
S_0
		&=		e^x , \\
dY_t
		&=		\( \frac{1}{\eps^2} \alpha(Y_t) - \frac{1}{\eps} \Lam(Y_t) \, \beta(Y_t) \) dt + \frac{1}{\eps} \beta(Y_t) d\Bt_t , &
Y_0
		&=		y , \\
d\<\Wt,\Bt\>
		&=		\rho \, dt , &
|\rho|
		&\leq	1 .
\end{aligned} \right\} \qquad (\text{under $\Pt$}) \label{eq:dS}
\end{align}
Here $\Wt$ and $\Bt$ are correlated Brownian motions and $\Nt(Y,dz)$ is a compensated Poisson random measure
\begin{align}
d\Nt_t(Y_t,dz)
		&=		dN_t(Y_t,dz) - \zeta(Y_t)\nu(dz) dt , &
\Et[dN_t(Y_t,dz)|Y_t]
		&=		\zeta(Y_t) \nu(dz) dt ,
\end{align}
We require that the measure $\nu$ satisfy 
\begin{align}
\int_\Rb \min(1,z^2)\nu(dz)
		&<		\infty , &
\int_{|z| \geq 1} e^z \nu(dz)
		&<		\infty , &
		&\text{and} &
\int_{|z| \geq 1} |z| \nu(dz)	
		&<		\infty . \label{eq:conditions}
\end{align}
The first integrability condition must be satisfied by all L\'evy measures.  The second integrability condition is needed to ensure $\Et [S_t]<\infty$ for all $t \in \Rb^{+}$.  The last integrability condition allows us to replace the indicator function that usually appears in the L\'evy-Kintchine formula $\Ib_{\{|z|<1\}}$ with the constant $1$.  Although we do not require it, a correlation of $\rho<0$ between $\Wt$ and $\Bt$ would be consistent with the leverage effect (i.e. a drop in the value of $S$ will usually be accompanied by an increase in volatility).
\par
Note that both the volatility of $S$, given by $\sig(Y)$, and the state-dependent L\'evy measure $\zeta(Y)\nu(dz)$, which controls the jumps of $S$, are driven by a common stochastic process $Y$.  The driving process $Y$ is \emph{fast-varying} in the following sense: under the physical measure $\Pb$, the dynamics of $Y$ are described by  
\begin{align}
\left.
\begin{aligned}
dY_t
		&=		\frac{1}{\eps^2} \alpha(Y_t)  dt + \frac{1}{\eps} \beta(Y_t) dB_t
\end{aligned} \right\} \qquad (\text{under $\Pb$})
\end{align}
where $B_t = \Bt_t - \int_0^t \Lam(Y_s) ds$ is a $\Pb$-Brownian motion. The generator of $Y$ under $\Pb$ is scaled by a factor of $1/\eps^2$
\begin{align}
\Ac_Y^\eps
		&=		\frac{1}{\eps^2} \( \frac{1}{2}\beta^2(y) \d_{yy}^2 + \alpha(y) \d_y \) .
\end{align}
Thus, $Y$ operates with an intrinsic time-scale $\eps^2$.   We assume $\eps^2 << 1$ so that the intrinsic time-scale of $Y$ is small.  Thus, $Y$ is fast-varying.  Throughout this text, we assume that, under $\Pb$, the process $Y$ is ergodic, has a unique invariant distribution $F_Y$, and that the smallest non-zero eigenvalue of $-\Ac_Y^\eps$ is strictly positive.  We also assume that the functions $\alpha$ and $\beta$, $\sig$, $\zeta$ are $\Lam$ is smooth and bounded and that there exists of a unique strong solution to SDE \eqref{eq:dS}.
\par
As mentioned in the introduction, the class of models described by \eqref{eq:dS} is a natural extension of the models considered in \cite{fouque}.  The key difference between the class of models we consider and those considered in \cite{fouque} is that we allow for the underlying $S$ to jump.  Moreover, we allow the jump intensity to be stochastic.

%%%%%%%%%%%%%%%%%%%%%%%%%%%%%%%%%%%%%%%%%%%%%%%%%%%%%%%%%
%
%								Pricing
%
%%%%%%%%%%%%%%%%%%%%%%%%%%%%%%%%%%%%%%%%%%%%%%%%%%%%%%%%%

\section{Option pricing}
\label{sec:pricing}
We wish to price a European-style option, which pays $H(S_t)$ at the maturity date $t>0$.  It will be convenient to introduce the returns process $X = \log S$.  Using It\^o's formula for It\^o-L\'evy processes (see \cite{oksendal2}, Theorem 1.14) one derives
\begin{align}
dX_t
		&=		\gam(Y_t) \, dt + \sig(Y_t) \, d\Wt_t + \int_\Rb z \, d\Nt_t(Y_t,dz) , &
X_0
		&= 		x ,  
\end{align}
where the drift $\gamma(Y_t)$ is given by
\begin{align}
\gam(Y_t)
		&=		-\frac{1}{2}\sig^2(Y_t) - \zeta(Y_t) \int_\Rb (e^z-1-z)\nu(dz) . 
\end{align}
Using risk-neutral pricing, the value $u^\eps(t,x,y)$ of the European option under consideration is 
\begin{align}
u^\eps(t,x,y)
		&=		\Et_{x,y} \[ h(X_t) \] , &
h(x)
		&:=		H(e^x) .
\end{align}
From the Kolmogorov backward equation we find that $u^\eps(t,x,y)$ satisfies the following partial integro-differential equation (PIDE) and boundary condition (BC)
\begin{align}
\(-\d_t + \Ac^\eps \) u^\eps
		&=		0 , &
u^\eps(0,x,y)
		&=		h(x) . \label{eq:PIDE}
\end{align}
where $\Ac^\eps$ is the generator of $(X,Y)$; it is a partial integro-differential operator given explicitly by
\begin{align}
\Ac^\eps
		&=		\frac{1}{\eps^2} \Ac_0 + \frac{1}{\eps} \Ac_1 + \Ac_2 , \label{eq:Aeps.1} \\
\Ac_0
		&=		\Ac_Y^1 = \frac{1}{2}\beta^2(y) \d_{yy}^2 + \alpha(y)\d_y , \label{eq:A0} \\
\Ac_1
		&=		\rho \beta(y) \sig(y) \d_{xy}^2 - \Lam(y) \beta(y) \d_y, \\
\Ac_2
		&=		\gam(y) \d_x + \frac{1}{2} \sig^2(y) \d_{xx}^2 + \zeta(y)\int_\Rb \(\theta_z - 1 - z \d_x \)\nu(dz) . \label{eq:A2first}
\end{align}
Here, $\theta_z$ is the \emph{shift operator} which acts on the $x$ variable: $\theta_z f(x,y) := f(x+z,y)$.  We shall assume that Cauchy problem \eqref{eq:PIDE} admits a unique classical solution.

%%%%%%%%%%%%%%%%%%%%%%%%%%%%%%%%%%%%%%%%%%%%%%%%%%%%%%%%%
%
%								Asymptotics
%
%%%%%%%%%%%%%%%%%%%%%%%%%%%%%%%%%%%%%%%%%%%%%%%%%%%%%%%%%

\subsection{Formal asymptotic analysis}
\label{sec:asymptotics}
For general ($\sig$, $\zeta$, $\alpha$, $\beta$, $\Lam$) there is no analytic solution to \eqref{eq:PIDE}.  We notice, however, that terms containing $\eps$ in \eqref{eq:PIDE} are diverging in the small-$\eps$ limit, giving rise to a \emph{singular} perturbation about the $\Oc(1)$ operator $(-\d_t + \Ac_2)$.  This special form suggests that we seek an asymptotic solution to PIDE \eqref{eq:PIDE}.  Thus, we expand $u^\eps$ in powers of the small parameter $\eps$
\begin{align}
u^\eps
		&=		\sum_{n=0}^\infty \eps^n u_n . \label{eq:uexpand}
\end{align}
Our goal will be to find an approximation $u^\eps = u_0 + \eps \, u_1 + \Oc(\eps^2)$ for the price of an option.  The choice of expanding in integer powers of $\eps$ is natural given the form of $\Ac^\eps$.
\par
In the formal asymptotic analysis that follows, we insert expansion \eqref{eq:uexpand} into PIDE \eqref{eq:PIDE} and collect terms of like powers of $\eps$, starting at the lowest order.  The $\Oc(1/\eps^2)$ and $\Oc(1/\eps)$ terms are
\begin{align}
\Oc(1/\eps^2):&&
0
		&=		\Ac_0 u_0 , \\
\Oc(1/\eps):&&
0
		&=		\Ac_1 u_0  + \Ac_0 u_1 .
\end{align}
Noting that all terms in $\Ac_0$ and $\Ac_1$ take derivatives with respect to $y$, we choose $u_0=u_0(t,x)$ and $u_1=u_1(t,x)$.  Continuing the asymptotic analysis, the $\Oc(1)$ and $\Oc(\eps)$ terms are
\begin{align}
\Oc(1):&&
0
		&=		(-\d_t + \Ac_2)u_0 + \Ac_0 u_2 , \label{eq:O1} \\
\Oc(\eps):&&
0
		&=		(-\d_t + \Ac_2)u_1 + \Ac_1 u_2 + \Ac_0 u_3 , \label{eq:Oeps}
\end{align}
where we have used the fact that $\Ac_1 u_1 = 0$ in the $\Oc(1)$ equation.  Equations \eqref{eq:O1} and \eqref{eq:Oeps} are equations of the form
\begin{align}
\Ac_0 u
		&=		\chi . \label{eq:fredholm}
\end{align}
Noting that $\int (\Ac_0 u) dF_Y = 0$ we observe that a solution $u$ to \eqref{eq:fredholm} exists if and only if $\chi$ satisfies the \emph{centering condition}
\begin{align}
\< \chi \>
		&:=		\int \chi \, dF_Y = 0 . \label{eq:center}
\end{align}  
Applying the centering condition to \eqref{eq:O1} and \eqref{eq:Oeps} yields
\begin{align}
\Oc(1):&&
0
		&=		(-\d_t + \< \Ac_2 \> ) u_0 , \label{eq:PIDEu0} \\
\Oc(\eps):&&
0
		&=		(-\d_t + \< \Ac_2\> ) u_1 + \< \Ac_1 u_2 \> . \label{eq:PIDEu1.a}
\end{align}
Note, from and \eqref{eq:O1} and \eqref{eq:PIDEu0} we have
\begin{align}
\Ac_0 u_2
		&=		- (-\d_t + \Ac_2)u_0 + (-\d_t + \<\Ac_2\>)u_0 
		=			- \( \Ac_2 - \<\Ac_2\> \) u_0 \\
		&=			- \frac{1}{2}\( \sig^2 - \< \sig^2\> \) \( \d_{xx}^2 - \d_x \) u_0 \\ & \qquad
						- \( \zeta - \< \zeta \> \) \( - \int_\Rb \Big( e^z - 1 - z \Big) \nu(dz) \d_x
							+ \int_\Rb \Big( \theta_z - 1 - z \d_x \Big) \nu(dz) \) u_0 \\
		&=			- \Ac_0 \( \frac{1}{2} \eta \( \d_{xx}^2 - \d_x \)
						- \xi \int_\Rb \Big( e^z - 1 - z \Big) \nu(dz) \d_x
						+ \xi \int_\Rb \Big( \theta_z - 1 - z \d_x \Big) \nu(dz) \) u_0 , \label{eq:A0u2=A0u0}
\end{align}
where we have introduced $\eta(y)$ and $\xi(y)$ as solutions to 
\begin{align}
\Ac_0 \eta
		&=		\sig^2 - \<\sig^2\> , &
\Ac_0 \xi
		&=		\zeta - \< \zeta \> . 		\label{eq:eta.xi}
\end{align}
Thus, from \eqref{eq:PIDEu1.a} and \eqref{eq:A0u2=A0u0} we find
\begin{align}
\Oc(\eps):&&
(-\d_t + \< \Ac_2\> )u_1
		&=		- \Bc u_0 , \label{eq:PIDEu1}
\end{align}
where the operator $\Bc$ is given by
\begin{align}
\Bc
		&=		\left\langle - \Ac_1 \( \frac{1}{2} \eta(y) \( \d_{xx}^2 - \d_x \)
					- \xi \int_\Rb \Big( e^z - 1 - z \Big) \nu(dz) \d_x
					+ \xi \int_\Rb \Big( \theta_z - 1 - z \d_x \Big) \nu(dz) \) \right\rangle \\
		&=		V_3 \( \d_{xxx}^3 - \d_{xx}^2 \) 
					+ U_3 \( - \int_\Rb \Big( e^z - 1 - z \Big) \nu(dz) \d_{xx}^2
						+ \int_\Rb \Big( \theta_z - 1 - z \d_x \Big) \d_x \nu(dz) \) \label{eq:B} \\ &\qquad
					+ V_2 \( \d_{xx}^2 - \d_x \)
					+ U_2 \( - \int_\Rb \Big( e^z - 1 - z \Big) \nu(dz) \d_x
						+ \int_\Rb \Big( \theta_z - 1 - z \d_x \Big) \nu(dz) \) ,  
\end{align}
and the constants ($V_3$, $U_3$, $V_2$, $U_2$) are defined as
\begin{align}
V_3
		&= 		- \frac{\rho}{2}\< \beta \sigma \d_y \eta \> , &
U_3
		&=		- \rho \< \beta \sigma \d_y \xi \> , &
V_2
		&=		\frac{1}{2} \< \beta \Lam \d_y \eta \> , &
U_2
		&=		\< \beta \Lam \d_y \xi \> .
\end{align}
This is as far as we will take the asymptotic analysis.  To review, we have found that $u_0(t,x)$ and $u_1(t,x)$ satisfy PIDEs \eqref{eq:PIDEu0} and \eqref{eq:PIDEu1} respectively.  We also impose the following BCs
\begin{align}
\Oc(1):&&
u_0(0,x)
	&=		h(x) , \label{eq:BCu0} \\
\Oc(\eps):&&
u_1(0,x)
	&=		0 . \label{eq:BCu1}
\end{align}

%%%%%%%%%%%%%%%%%%%%%%%%%%%%%%%%%%%%%%%%%%%%%%%%%%%%%%%%%
%
%								Solution for u_0 and u_1
%
%%%%%%%%%%%%%%%%%%%%%%%%%%%%%%%%%%%%%%%%%%%%%%%%%%%%%%%%%

\subsection{Explicit solution for $u_0(t,x)$ and $u_1(t,x)$}
\label{sec:explicit}
In order to find explicit formulas for $u_0(t,x)$ and $u_1(t,x)$, we note that the operator
\begin{align}
\< \Ac_2 \>
		&=		\< \gam\> \d_x + \frac{1}{2} \< \sig^2 \> \d_{xx}^2 + \<\zeta\> \int_\Rb \(\theta_z - 1 - z \d_x \) \nu(dz) ,
					\label{eq:<A2>} \\
\<\gam\>
		&=		-\frac{1}{2}\<\sig^2\> - \<\zeta\> \int_\Rb (e^z-1-z)\nu(dz) , 
\end{align}
is the generator of a L\'evy process with L\'evy triplet $(\<\gam\>,\<\sig^2\>,\<\zeta\>\nu)$.  Thus, we may apply standard results from the classical theory of Fourier transforms to obtain solutions to PIDEs \eqref{eq:PIDEu0} and \eqref{eq:PIDEu1}.
\begin{theorem}
\label{thm:u0u1}
Assume $h$ has a generalized Fourier transform
\begin{align}
\hh(\lam)
	&:=	\frac{1}{\sqrt{2\pi}} \int_\Rb dx \, e^{-i \lam x} h(x) < \infty &
	&\textrm{for some} &
\lam
	&:= \lam_r + i \lam_i & 
	&\textrm{where} &
\lam_r, \lam_i
	&\in 	\Rb .
\end{align}
Define
\begin{align}
\phi_\lam
		&=		i \< \gam\> \lam - \frac{1}{2} \< \sig^2 \> \lam^2 
					+ \<\zeta\> \int_\Rb \(e^{i \lam z} - 1 - i \lam z \) \nu(dz) , \label{eq:phi} \\
B_\lam
		&=		V_3 \( -i\lam^3 + \lam^2 \) 
					+ U_3 \( \lam^2 \int_\Rb \Big( e^{z} - 1 - z \Big) \nu(dz) 
						+ i \lam \int_\Rb \Big( e^{i \lam z} - 1 - i \lam z \Big) \nu(dz) \) \\ &\qquad
					+ V_2 \( -\lam^2 - i \lam \) 
					+ U_2 \( -i \lam \int_\Rb \Big( e^{z} - 1 - z \Big) \nu(dz) 
						+ \int_\Rb \Big( e^{i \lam z} - 1 - i \lam z \Big) \nu(dz) \) . \label{eq:B.lambda}
\end{align}
Assume that $\phi_\lam$ is analytic in an infinite strip parallel to the real axis which contains $i \lam_i$.  Then the solution $u_0(t,x)$ to PIDE \eqref{eq:PIDEu0} with BC \eqref{eq:BCu0} is
\begin{align}
u_0(t,x)
		&=		\frac{1}{\sqrt{2\pi}} \int_\Rb d\lam_r \, e^{t \phi_\lam} \hh(\lam) e^{i \lam x} , \label{eq:u0.FT}
\end{align}
and the solution $u_1(t,x)$ to PIDE \eqref{eq:PIDEu1} with BC \eqref{eq:BCu1} is
\begin{align}
u_1(t,x)
		&=		\frac{1}{\sqrt{2\pi}} \int_\Rb d\lam_r \, t \, e^{t \phi_\lam} \hh(\lam) B_\lam e^{i \lam x}  . \label{eq:u1.FT}
\end{align}
\end{theorem}
\begin{proof}
See appendix \ref{sec:u0u1}.
\end{proof}
\begin{remark}
Those who are familiar with L\'evy processes will recognize $\phi_\lam$ as the characteristic L\'evy exponent corresponding to L\'evy triplet $(\<\gam\>,\<\sig^2\>,\<\zeta\>\nu)$.
\end{remark}
\begin{remark}[On calls and puts]
Note that a European call option with payoff function $h(x)=(e^x-e^k)^+$ has a generalized Fourier transform
\begin{align}
\hh(\lam)
		&=		\frac{1}{\sqrt{2\pi}} \int_\Rb dx \, e^{-i \lam x} (e^x - e^k)^+ 
		=			\frac{-e^{k-i k \lam}}{\sqrt{2 \pi } \(i \lam + \lam^2\)} ,  & 
\lam
	&=	\lam_r + i \lam_i , &
\lam_i
	&\in (-\infty, -1) . \label{eq:hhat}
\end{align}
Likewise, a European put option with payoff function $h(x)=(e^k-e^x)^+$ has a generalized Fourier transform
\begin{align}
\hh(\lam)
		&=		\frac{1}{\sqrt{2\pi}} \int_\Rb dx \, e^{-i \lam x} (e^x - e^k)^+ 
		=			\frac{-e^{k-i k \lam}}{\sqrt{2 \pi } \(i \lam + \lam^2\)} ,  & 
\lam
	&=	\lam_r + i \lam_i , &
\lam_i
	&\in (0, \infty) . 
\end{align}
For a review of the generalized Fourier transforms as they relate to L\'evy processes, we refer the reader to any of the following: \cite{levendorskiibook,lewis2001simple,lipton2002}.
\end{remark}
From the arguments in \cite{fpss}, Chapter 4, it follows that for fixed $(t,x,y)$ there exists a constant $C$ such that $|u^\eps - (u_0 + \eps \, u_1)| < C \eps^2$ when $h$ is smooth.  We verify this result numerically in Section \ref{sec:example} by comparing the approximate price $u_0 + \eps \, u_1$ of a derivative-asset, calculated using the formulas in Theorem \ref{thm:u0u1}, to the full price $u^\eps$, calculated via Monte Carlo simulation.

%%%%%%%%%%%%%%%%%%%%%%%%%%%%%%%%%%%%%%%%%%%%%%%%%%%%%%%%%
%
%								Example
%
%%%%%%%%%%%%%%%%%%%%%%%%%%%%%%%%%%%%%%%%%%%%%%%%%%%%%%%%%

\section{Example: Merton jump-diffusion with stochastic volatility and stochastic jump-intensity}
\label{sec:example}
In this section we provide one specific example within the class of models described in Section \ref{sec:model}.  Specifically, we extend the jump-diffusion model of \cite{merton1976option} to include stochastic volatility and jump-intensity.  We refer to this class of models as the \emph{Extended Merton} class or simply \emph{ExtMerton}.  In the Merton jump-diffusion model, jumps are $\log$-normally distributed.  Thus, we let the measure $\nu$ be given by
\begin{align}
\nu(dz)
		&=		\frac{1}{\sqrt{2 \pi s^2}}\exp \( \frac{-(z-m)^2}{2s^2} \) dz . \label{eq:normal.nu}
\end{align}
Under this specification, we have
\begin{align}
\< \gamma \>
		&=		-\frac{1}{2}\<\sig^2\> - \<\zeta\> \(e^{m + \frac{s^2}{2}} - 1 - m \) , \\
\phi_\lam
		&=		i \< \gam\> \lam - \frac{1}{2} \< \sig^2 \> \lam^2 
					+ \<\zeta\> \( e^{i \lam m - \frac{1}{2}s^2 \lam^2} -  1 - i \lam m \) , \\
B_\lam
		&=		V_3 \( -i\lam^3 + \lam^2 \) 
					+ U_3 \( \lam^2 \( e^{m+s^2/2} - 1 - m \) + i \lam \( e^{i \lam m - s^2 \lam^2/2} - 1 - i \lam m \) \) \\ &\qquad
					+ V_2 \( -\lam^2 - i \lam \) 
					+ U_2 \( - i \lam \( e^{m+s^2/2} - 1 - m \) + \(e^{i \lam m - s^2 \lam^2/2} - 1 - i \lam m \) \) .
\end{align}
For a European call option with payoff $h(X_t)=(e^{X_t} - e^k)^+$, the generalized Fourier transform of $h(x)$ is given by \eqref{eq:hhat}.  The values of ($\<\sig^2\>$, $\<\zeta\>$, $V_3$, $U_3$, $V_2$, $U_2$), which are needed to compute $u_1$, depend on the particular choice of $\sig(y)$ and $\zeta(y)$ as well as a specific choice for the $Y$ process.  In the numerical examples below we let $\alpha(y)=-y$, $\beta(y)=\beta$, and $\Lam(y)=\Lam$ so that
\begin{align}
dY_t
		&=		\( - \frac{1}{\eps^2} Y_t - \frac{1}{\eps} \Lam \, \beta \) dt + \frac{1}{\eps} \beta \, d\Bt_t ,
\end{align}
and we choose $\sig(y)=a e^y$ and $\zeta(y)= b e^y$.  With these choices the invariant distribution of $Y$ under the physical measure $\Pb$ is normal $F_Y \sim \Nc (0,\tfrac{\beta^2}{2})$ and we can compute explicitly
\begin{align}
\<\sig^2\> 
		&= a^2 e^{\beta^2}, &
\<\zeta\> 
		&= b e^{\frac{\beta^2}{4}}, \\
V_3
		&=	\frac{\rho}{\beta}a^3 e^{\frac{5 \beta^2}{4}} \(e^{\beta^2}-1\) , &
U_3
		&=	\frac{\rho }{\beta }2 a b \left(e^{\beta ^2}-e^{\frac{\beta ^2}{2}}\right) , \\
V_2
		&=	- \beta \Lam a^2 e^{\beta^2} , &
U_2
		&=	- \beta  \Lam  b e^{\frac{\beta^2}{4}} .
\end{align}
The implied volatility $I$ corresponding to a European call option with price $u$ is defined implicitly though
\begin{align}
u^{BS}(I)
	&=	u	 , 
\end{align}
where $u^{BS}(I)$ is the price of the call option (with the same strike and maturity) as computed in the Black-Scholes framework assuming a volatility of $I$.  
In figure \ref{fig:ImpVol} we fix the time to maturity at $t=1/10$ and we plot the implied volatility smile induced by the approximate price of European calls $u_0 + \eps \, u_1$ for $\eps=\{0.1, .033, 0.01\}$.  For comparison, we also plot the implied volatility smile induced by the full price $u^\eps$ (computed using Monte Carlo simulation).
As expected, as $\eps$ goes to zero, the implied volatility induced by the approximate price $u_0 + \eps \, u_1$ converges to the implied volatility induced by the full price $u^\eps$.  
\par
Note that, within our framework, there is nothing unique about the Merton model.  
By using the methods outlined in this paper, \emph{any} exponential L\'evy model for which one can explicitly compute
\begin{align}
\int_\Rb \( e^{i \lam z} - 1 - i \lam z \) \nu(dz) ,
\end{align}
can be extended to include stochastic volatility and stochastic jump intensity.
Likewise, there is nothing unique about our particular choice of functions $\sig(y)$, $\zeta(y)$ or our choice of driving process $Y$.  One may choose any combination of $\sig(y)$, $\zeta(y)$ and $Y$ that allow one to compute (analytically or numerically) the values of ($\<\sig^2\>$, $\<\zeta\>$, $V_3$, $U_3$, $V_2$, $U_2$).  Thus, the framework described in this paper provides considerable modeling flexibility.

%%%%%%%%%%%%%%%%%%%%%%%%%%%%%%%%%%%%%%%%%%%%%%%%%%%%%%%%%
%									Monte Carlo Figure
%%%%%%%%%%%%%%%%%%%%%%%%%%%%%%%%%%%%%%%%%%%%%%%%%%%%%%%%

\begin{figure}
\centering
\begin{tabular}{ | c | }
\hline
\includegraphics[width=.5\textwidth,height=.26\textheight]{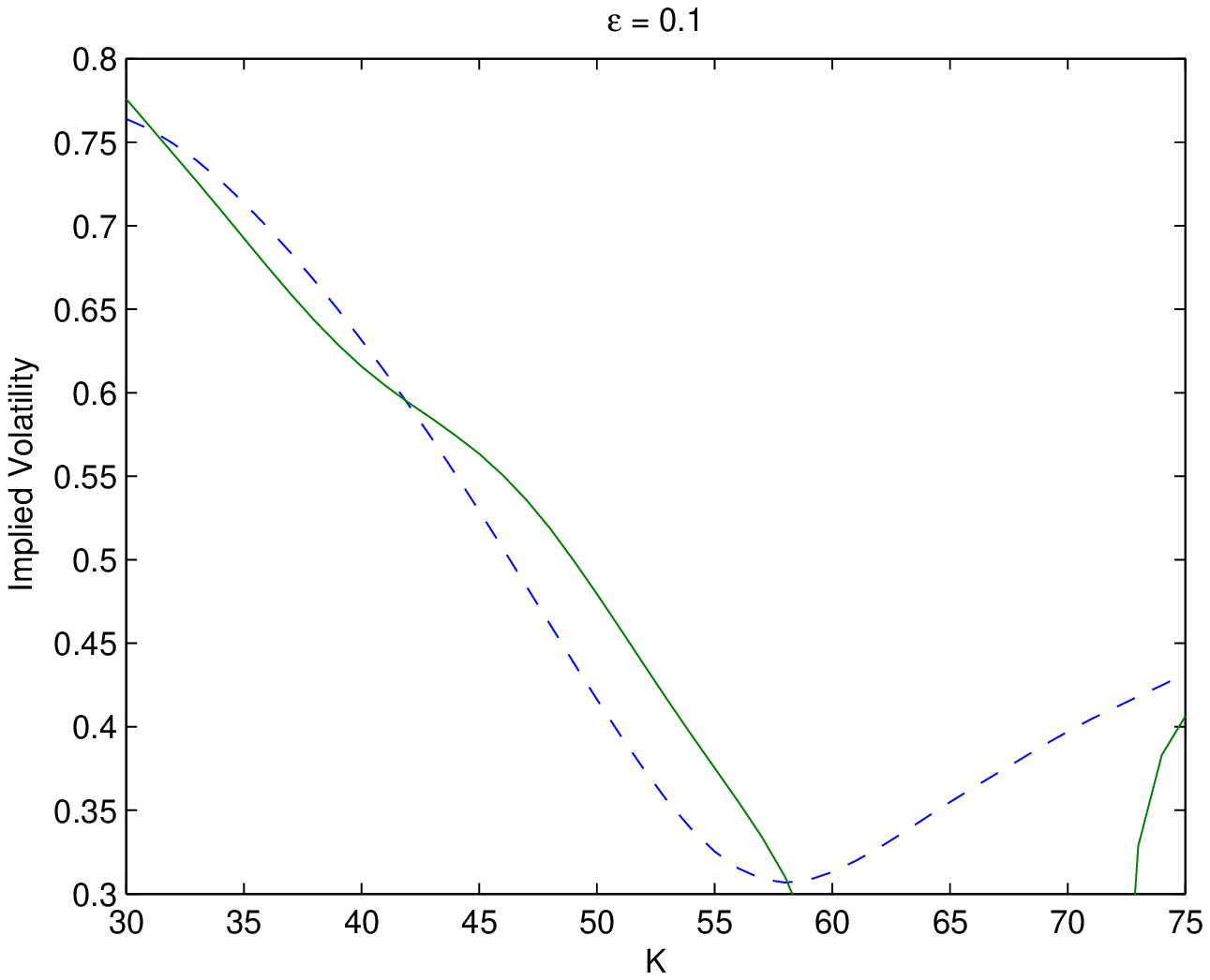} \\ \hline
\includegraphics[width=.5\textwidth,height=.26\textheight]{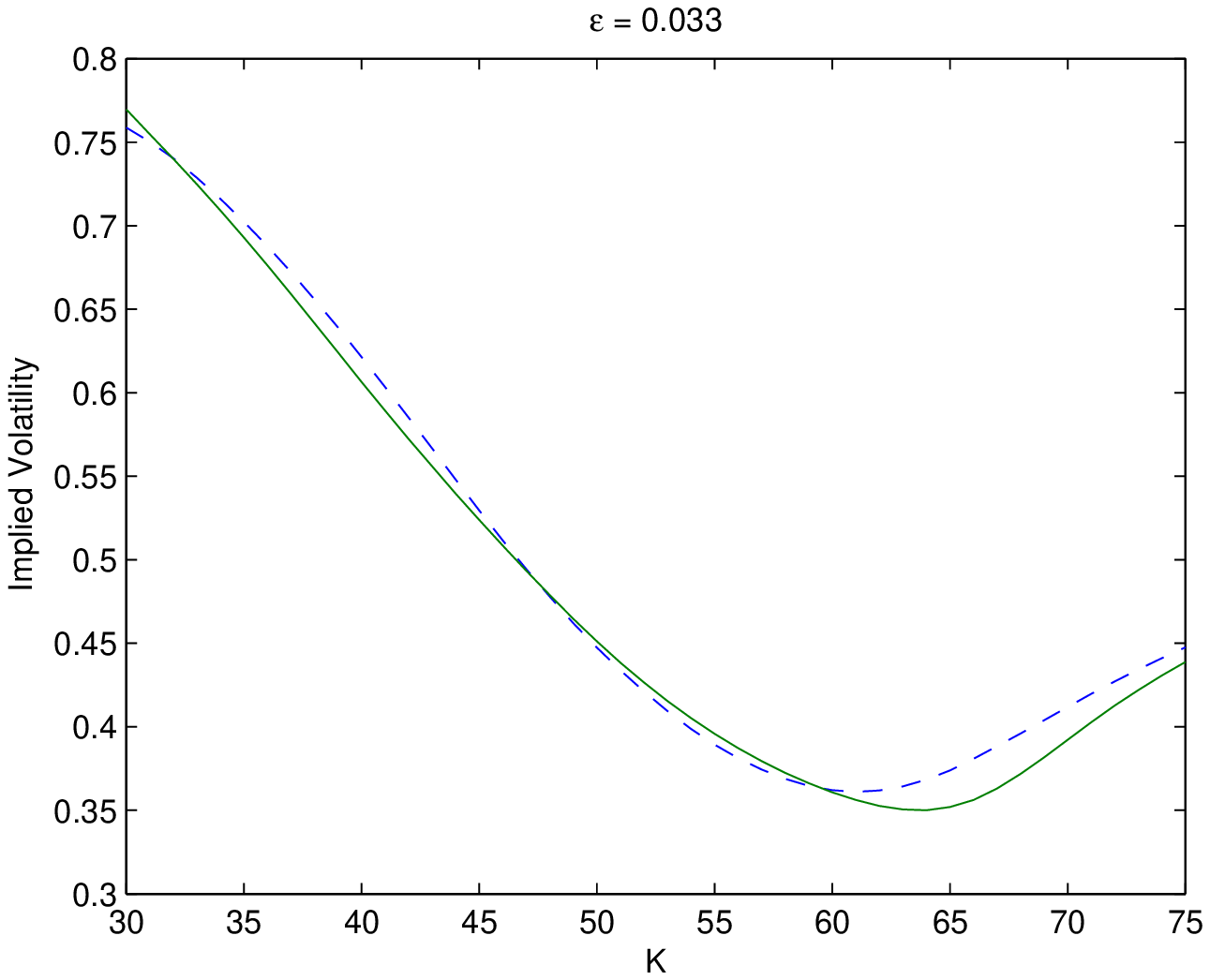} \\ \hline
\includegraphics[width=.5\textwidth,height=.26\textheight]{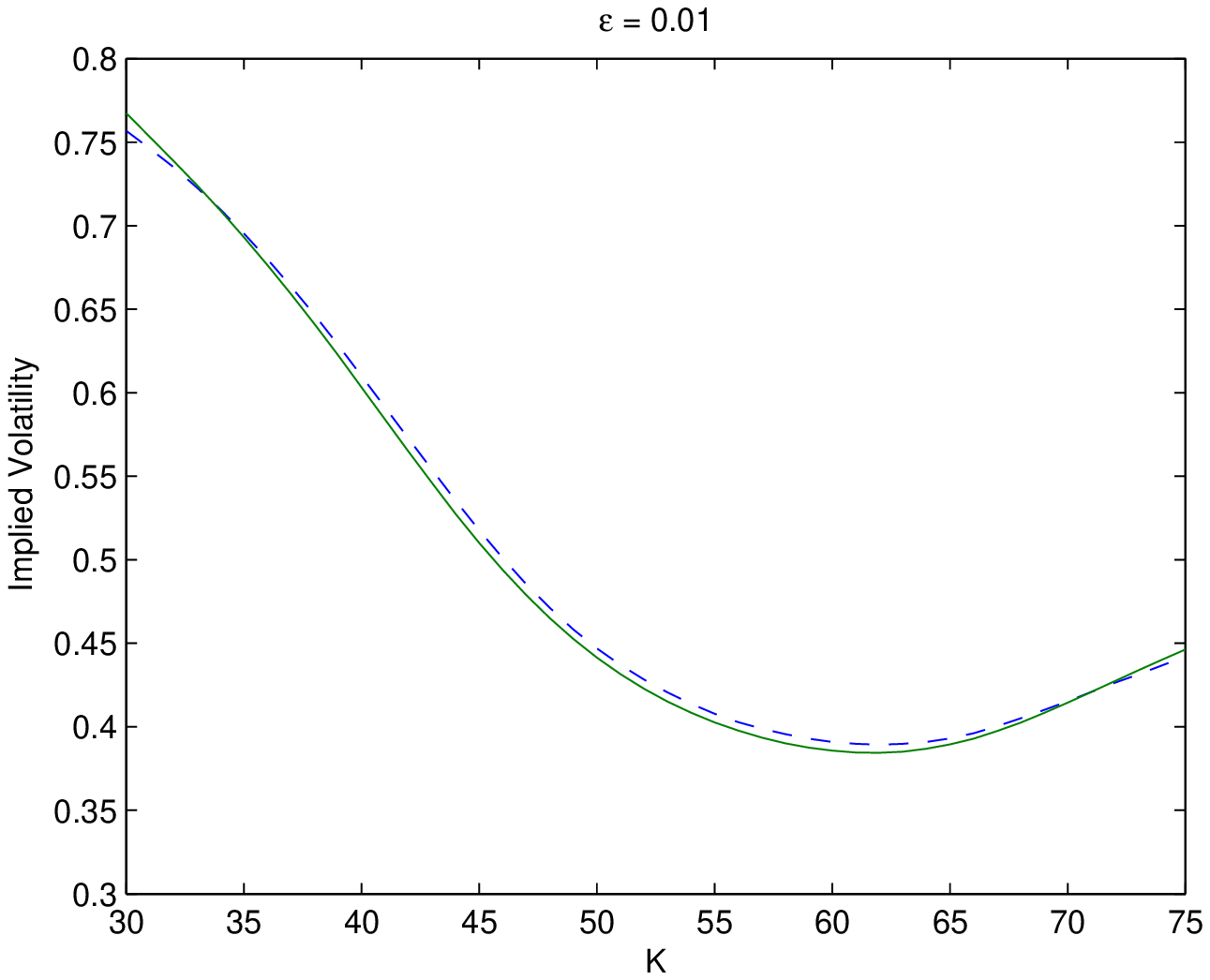} \\ \hline
\end{tabular}
\caption{Using the model described in section \ref{sec:example}, we plot the implied volatility induced by the price of European call option as a function of the strike price $K$.  In each plot, the dashed blue line corresponds to the implied volatility induced by the full price $u^\eps$ (computed via Monte Carlo simulation) and the solid green line corresponds to the implied volatility induced by our approximation $u_0 + \eps u_1$.  For all plots we use the following parameter values: $t=1/10$, $e^x=50$, $m=-0.2$, $s=0.2$, $\rho=-0.7$, $a=0.2$, $b=1.5$, $\beta=1.0$ and $\Lam=0.25$.}
\label{fig:ImpVol}
\end{figure}

%%%%%%%%%%%%%%%%%%%%%%%%%%%%%%%%%%%%%%%%%%%%%%%%%%%%%%%%%
%
%  						Calibration
%
%%%%%%%%%%%%%%%%%%%%%%%%%%%%%%%%%%%%%%%%%%%%%%%%%%%%%%%%

\section{Calibration to S\&P500 index options}
\label{sec:calibration}
In this section we calibrate ExtMerton jump-diffusion class discussed in Section \ref{sec:example} to the implied volatility surface of S\&P500 options.  For comparison, we also calibrate the classical Merton model and the fast mean-reverting stochastic volatility (FMR-SV) class of models of \citet*{fouque} to the same set of data.
\par
In order to formulate the calibration procedure, we introduce the following notation
\begin{align}
\theta
	&:=	( \<\sig^2\>,\<\zeta\>,m,s,V_2^\eps,V_3^\eps,U_2^\eps,U_3^\eps ) , \\
\Theta
	&:=	\{ \theta : \<\sig^2\> > 0, \<\zeta\> \geq 0, m \in \Rb, s \geq 0 , (V_2^\eps,V_3^\eps,U_2^\eps ,U_3^\eps) \in \Rb^4 \} ,
\end{align}
where we have defined $V_i^\eps:= \eps \, V_i$ and $U_i^\eps := \eps \, U_i$.  Note that the components of $\theta$ are the unobservable parameters needed to compute the approximate price of an option $u_0 + \eps \, u_1$ in the ExtMerton framework, and $\Theta$ is the feasible state space of these parameters.  Note also that we do not assume a specific value for $\eps$, a specific volatility process $Y$, or specific functions: $\sig(y)$ or $\zeta(y)$.  In fact, this is one of the \emph{main features} of the class of models considered in this paper.  By assuming that the driving factor $Y$ is fast-varying and ergodic, specific choices for ($\eps$, $Y$, $\sig(y)$, $\zeta(y)$) are \emph{not needed} to compute the approximate price $u_0 + \eps \, u_1$ of an option (or the corresponding implied volatility).  For the purposes of calibration and pricing, the relevant information about ($\eps$, $Y$, $\sig(y)$, $\zeta(y)$) is neatly contained in $\<\sig^2\>$, $\<\zeta\>$ and the four group parameters $\{V_i^\eps,U_i^\eps,i=1,2\}$.
\par
Let $I^{\text{obs}}(t,k)$ be the observed implied volatility of a European call option with time to maturity $t$ and $\log$ strike $k = \log K$.  Let $I^\eps(t,k;\theta)$ be the implied volatility of a European call option with the same maturity and strike as computed in the ExtMerton framework using parameters $\theta \in \Theta$.  We formulate the calibration problem for the ExtMerton class as a least squares optimization.  That is, we seek $\theta^{*}$ such that
\begin{align}
\min_{\theta \in \Theta} \sum_{i} \( I^{\text{obs}}(t_i,k_i) - I^\eps(t_i,k_i;\theta) \)^2
	&=	\sum_{i} \( I(t_i,k_i) - I^\eps(t_i,k_j;\theta^*) \)^2 . \label{eq:lsq}
\end{align}
Here, the sum runs over all pairs $(t_i,k_i)$ in the data set.  Note: we do \emph{not} calibrate maturity-by-maturity.  The calibration procedures for the Merton model and the FMR-SV class are performed in a similar fashion by solving \eqref{eq:lsq} for $\theta \in \Theta^{\text{Mer}}$ and $\theta \in \Theta^{\text{FMR}}$ respectively, where
\begin{align}
\Theta^{\textrm{Mer}}
	&:=	\{ \theta : \<\sig^2\> > 0, \<\zeta\> \geq 0, m \in \Rb, s \geq 0, (V_2^\eps,V_3^\eps,U_2^\eps ,U_3^\eps) = 0 \} , \\
\Theta^{\textrm{FMR}}
	&:=	\{ \theta : \<\sig^2\> > 0, (\<\zeta\>, m,s) = 0, (V_2^\eps,V_3^\eps) \in \Rb^2, (U_2^\eps ,U_3^\eps) = 0 \} .
\end{align}
Note that by requiring $(V_2^\eps,V_3^\eps,U_2^\eps ,U_3^\eps) = 0$ in $\Theta^{\text{Mer}}$ the effects of stochastic volatility and stochastic jump intensity disappear, and the approximate option price in the ExtMerton class $u_0 + \eps \, u_1$ reduces to the Merton price $u_0$.  Similarly, by requiring that $(\<\zeta\>,m,s,U_2^\eps,U_3^\eps) = 0$ in $\Theta^{\text{FMR}}$, the effect of the jumps disappears (the effects of stochastic volatility remain), and the approximate option price in the ExtMerton class $u_0 + \eps \, u_1$ reduces to the price as computed in the FMR-SV class.
\par
We perform the calibration procedure for all three frameworks (ExtMerton class, classical Merton model, and FMR-SV class) on S\&P500 index options on four separate dates:
\begin{itemize}
\item January 4, 2010 encompassing maturities of 47, 75 and 103 days,
\item October 1, 2010 encompassing maturities of 50, 78 and 113 days,
\item December 19, 2011 encompassing maturities of 59, 88, 122, 177, 273 and 363 days, and
\item January 11, 2012 encompassing maturities of 66, 100, 155, 251 and 341 days.
\end{itemize}
To perform the calibration we use 
Matlab's built-in non-linear least squares optimizer: \texttt{fmincon}.
The obtained fits for all three models are plotted in Figures
\ref{fig:ImpVol2010Jan}, \ref{fig:ImpVol2010Oct}, \ref{fig:merton1} and \ref{fig:merton2}.
For each plot, the units of the horizontal axis are \emph{log-moneyness}: $\text{LM} := k-x$.  The vertical axis represents implied volatility.
Summarizing statistics can be found in Table \ref{tab:stats}.
\par
A visual inspection of figures \ref{fig:ImpVol2010Jan}, \ref{fig:ImpVol2010Oct}, \ref{fig:merton1} and \ref{fig:merton2} clearly supports the use of the ExtMerton class over both the Merton model and the FMR-SV class.  The visual evidence is confirmed by the obtained root mean-square error (RMSE), which for the ExtMerton class is of the same order as the implied volatility bid-ask spread.  Furthermore, the ExtMerton RMSE is less than half the RMSE of the classical Merton model and roughly one fourth the RMSE of the FMR-SV class.  Intuitively, the reason for the improved fit in ExtMerton class is that to obtain a tight fit at longer maturities one requires a model with stochastic volatility, whereas short maturities require a model with jumps in order to reproduce the strong smile.
\par
To further support the addition of stochastic volatility and jump intensity we consider other L\'evy measures: Gumbel, Variance Gamma, uniform and Dirac.  For each measure we calibrate the corresponding L\'evy model and its extended counterpart to S\&P500 implied volatilities from December 19, 2011.  As in the Merton model, we observe that the RMSE for the extended models is of the order of the implied volatility bid-ask spread and roughly one half the RMSE of the classical L\'evy counterparts.  Summarizing statistics can be found in table \ref{tab:stats}.

%%%%%%%%%%%%%%%%%%%%%%%%%%%%%%%%%%%%%%%%%%%%%%%%%%%%%%%%%%
%									Calibration Figures
%%%%%%%%%%%%%%%%%%%%%%%%%%%%%%%%%%%%%%%%%%%%%%%%%%%%%%%%%%

%\clearpage
\begin{figure}
\centering
\begin{tabular}{ | c | c | c | }
\hline
Extended Merton & Merton & FMR-SV \\ \hline
47 days-to-maturity & 47 days-to-maturity & 47 days-to-maturity  \\
\includegraphics[width=.33\textwidth,height=.2\textheight]{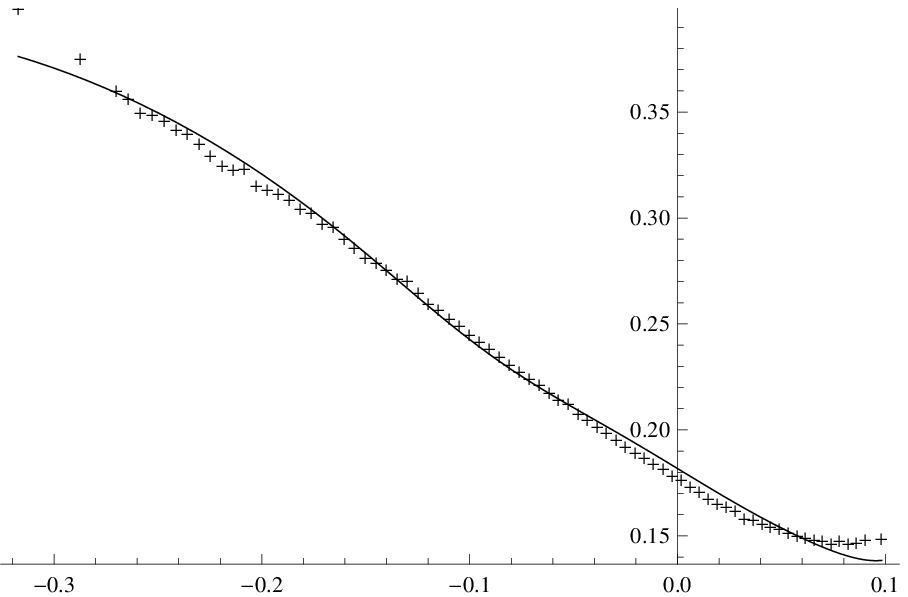} &
\includegraphics[width=.33\textwidth,height=.2\textheight]{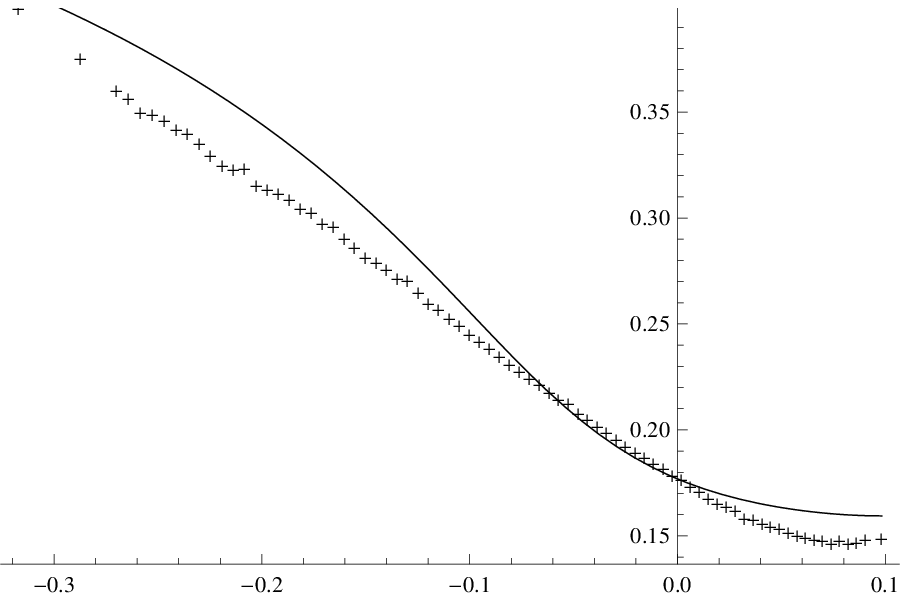} &
\includegraphics[width=.33\textwidth,height=.2\textheight]{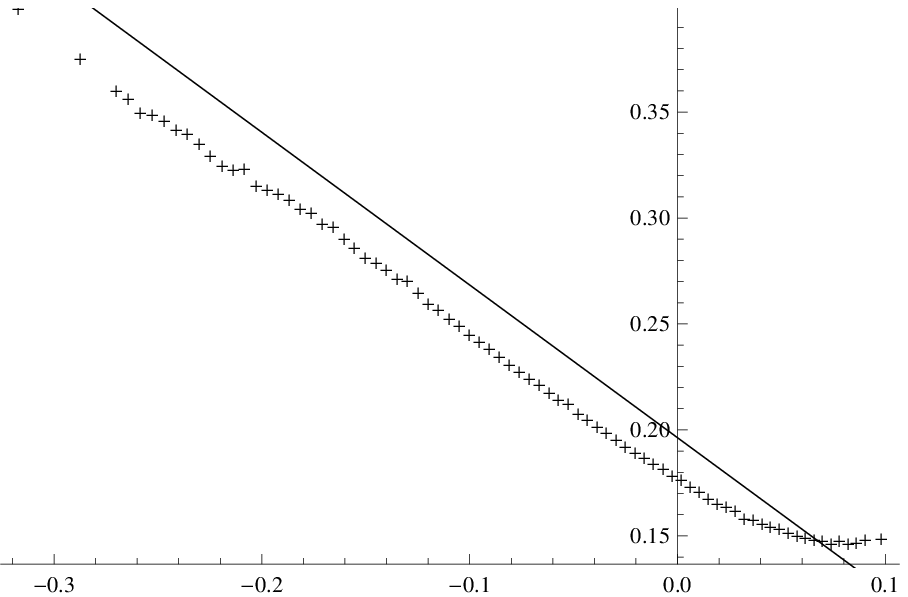} \\ \hline
75 days-to-maturity& 75 days-to-maturity & 75 days-to-maturity  \\
\includegraphics[width=.33\textwidth,height=.2\textheight]{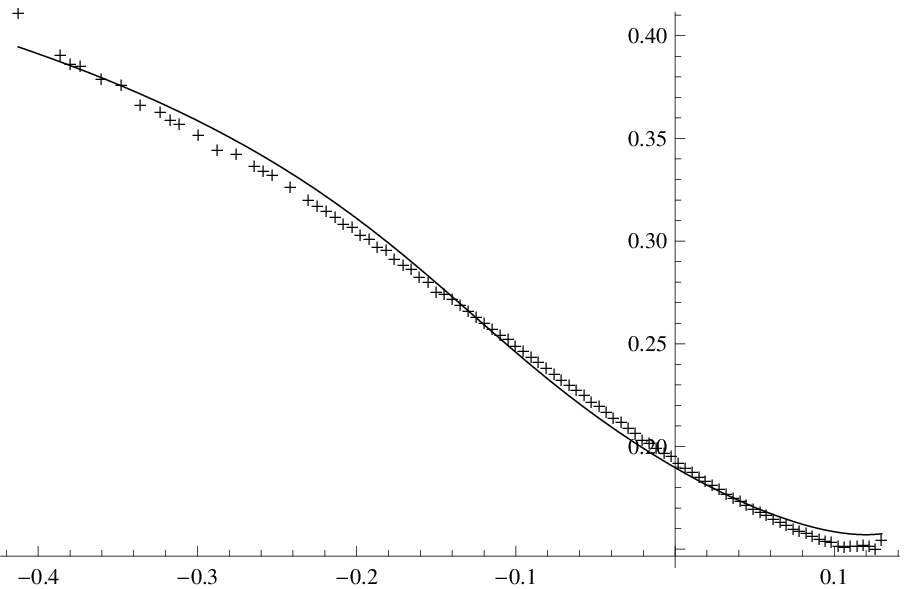} &
\includegraphics[width=.33\textwidth,height=.2\textheight]{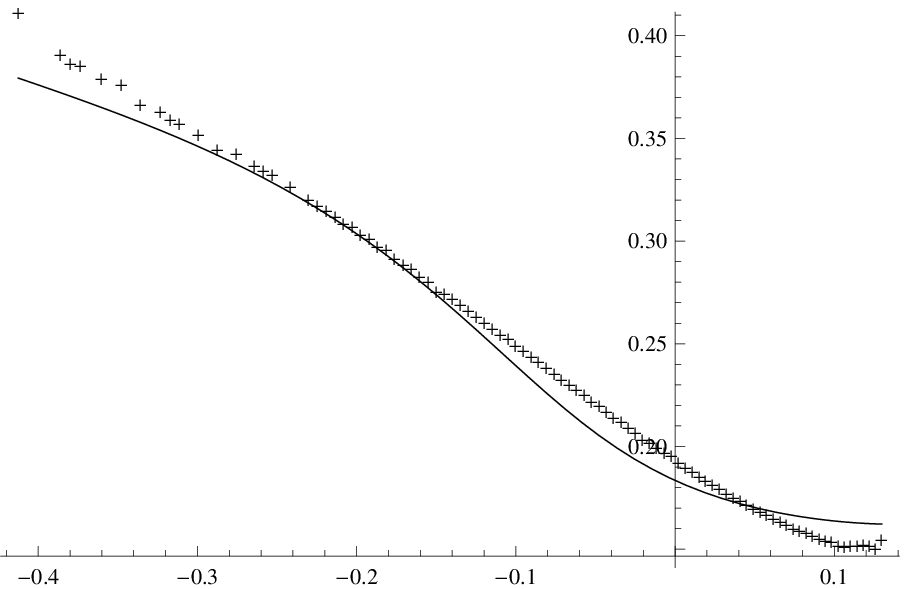} &
\includegraphics[width=.33\textwidth,height=.2\textheight]{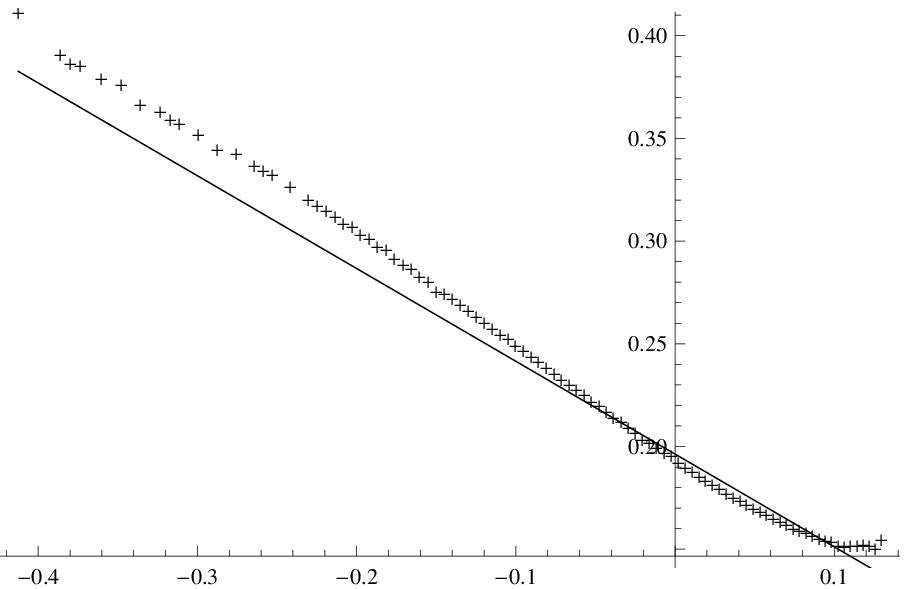} \\ \hline
103 days-to-maturity & 103 days-to-maturity & 103 days-to-maturity  \\
\includegraphics[width=.33\textwidth,height=.2\textheight]{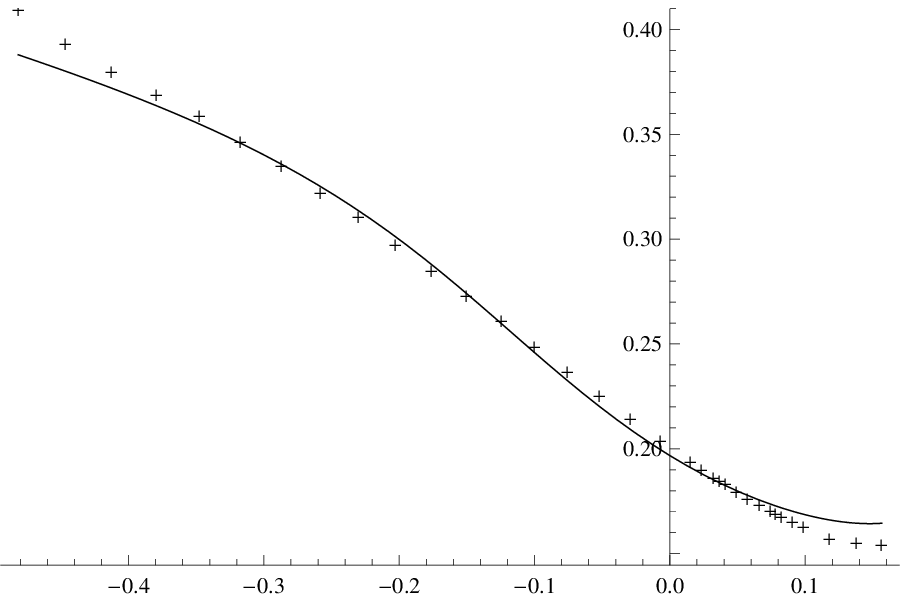} &
\includegraphics[width=.33\textwidth,height=.2\textheight]{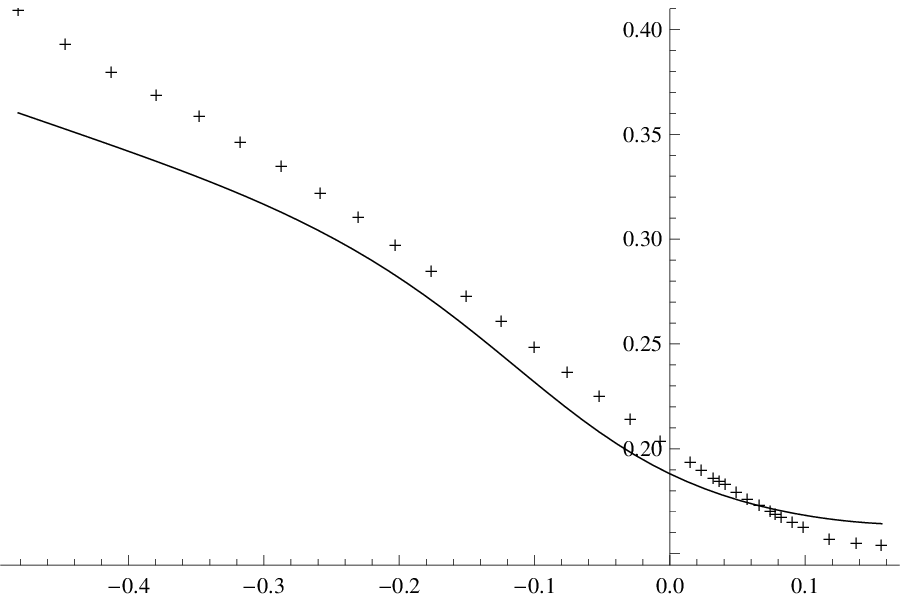} &
\includegraphics[width=.33\textwidth,height=.2\textheight]{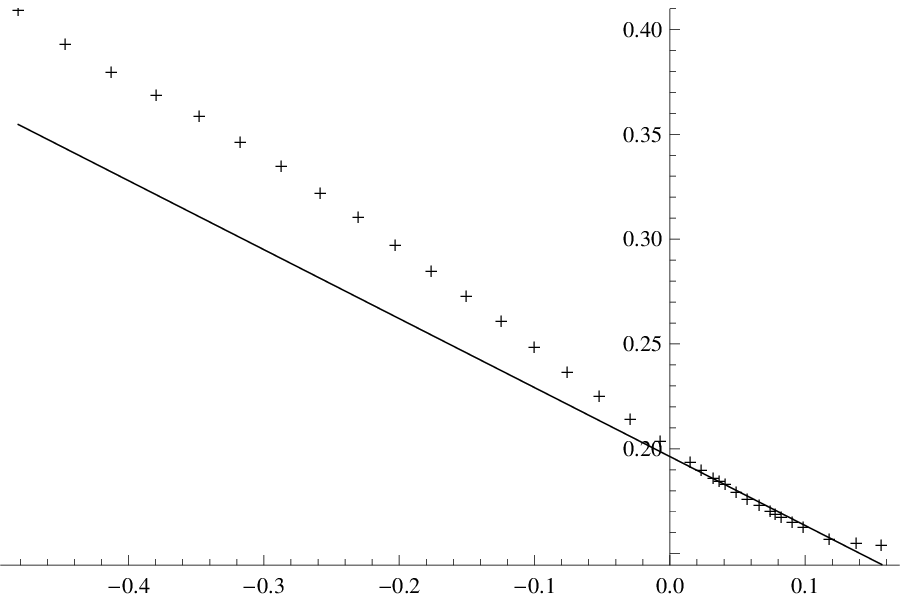} \\ \hline
\end{tabular}
\caption{Implied volatility fit to S\&P500 index options from January 4, 2010.}
\label{fig:ImpVol2010Jan}
\end{figure}
%figure generated from mathematica notebook ``FitToData2.nb''

%\clearpage
\begin{figure}
\centering
\begin{tabular}{ | c | c | c | }
\hline
Extended Merton & Merton & FMR-SV \\ \hline
50 days-to-maturity & 50 days-to-maturity & 50 days-to-maturity  \\
\includegraphics[width=.33\textwidth,height=.2\textheight]{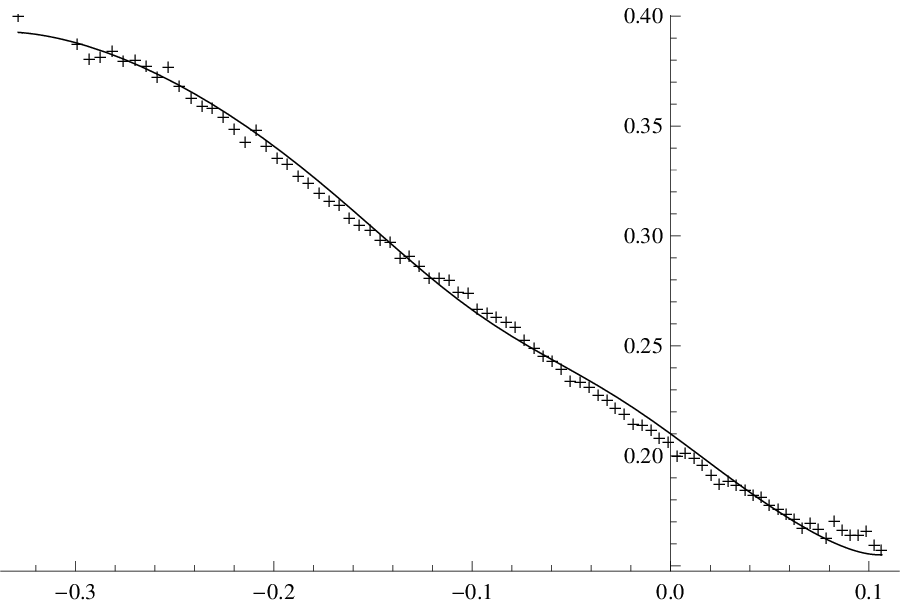} &
\includegraphics[width=.33\textwidth,height=.2\textheight]{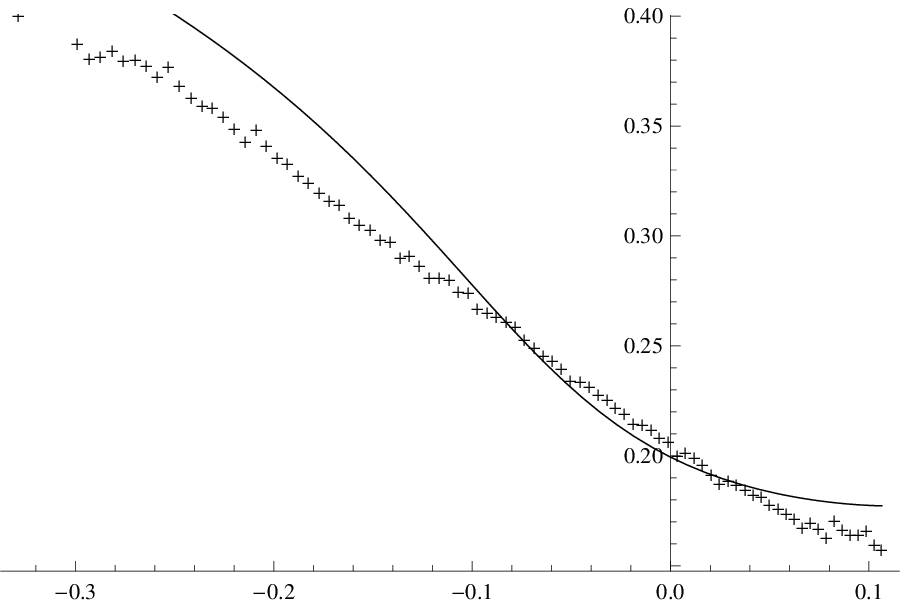} &
\includegraphics[width=.33\textwidth,height=.2\textheight]{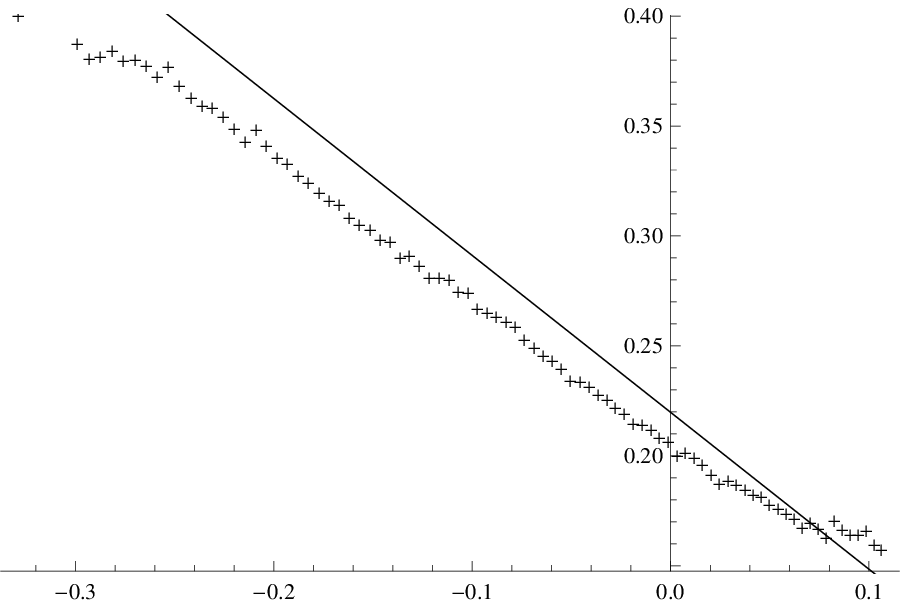} \\ \hline
78 days-to-maturity& 78 days-to-maturity & 78 days-to-maturity  \\
\includegraphics[width=.33\textwidth,height=.2\textheight]{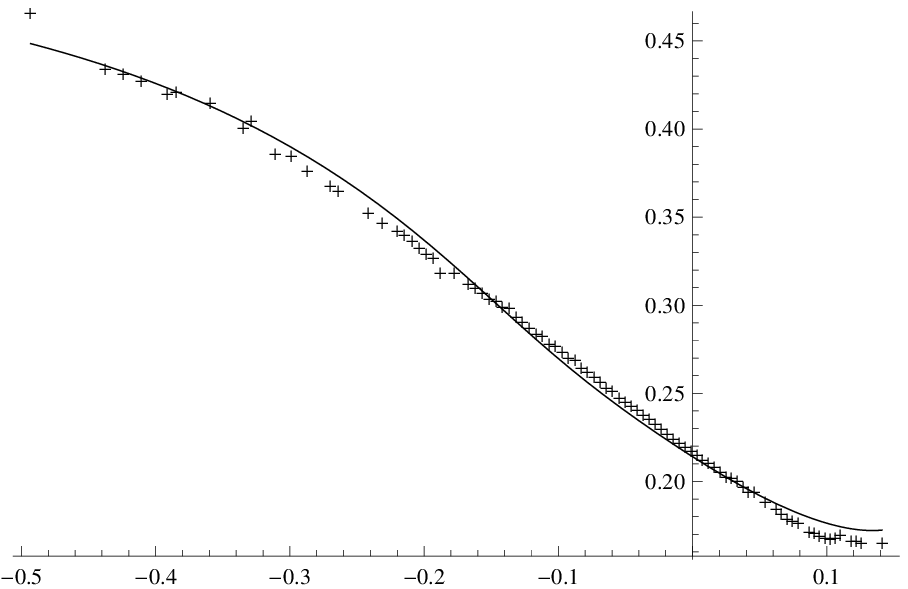} &
\includegraphics[width=.33\textwidth,height=.2\textheight]{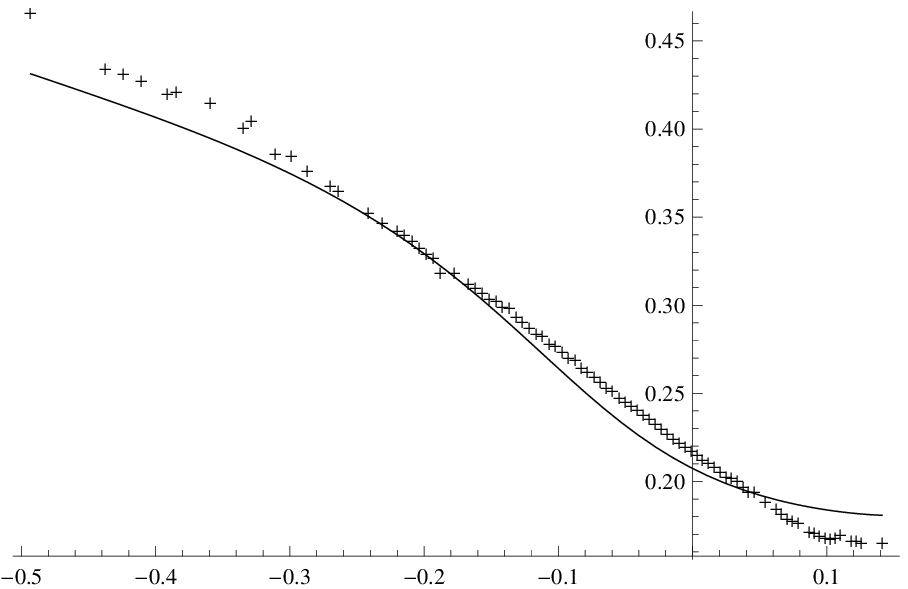} &
\includegraphics[width=.33\textwidth,height=.2\textheight]{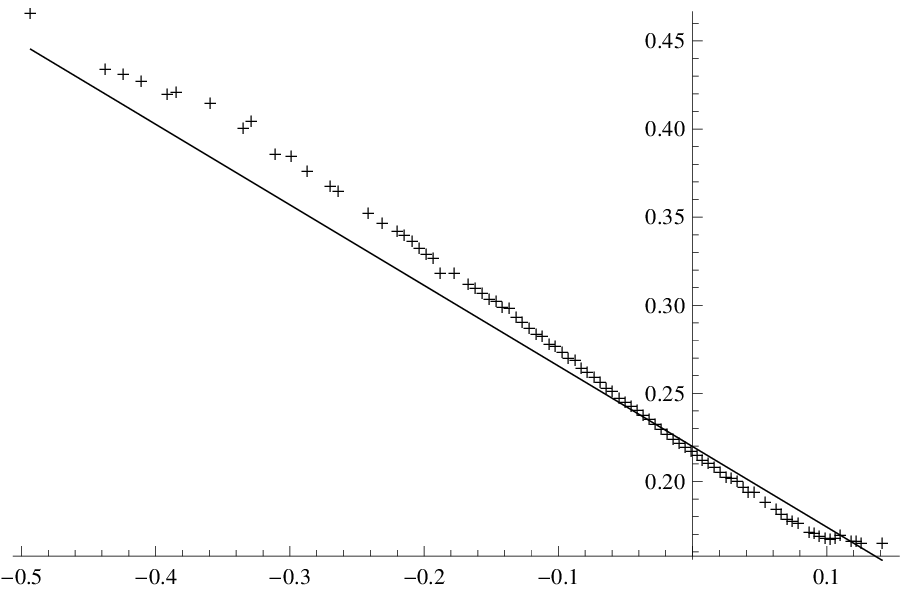} \\ \hline
113 days-to-maturity & 113 days-to-maturity & 113 days-to-maturity  \\
\includegraphics[width=.33\textwidth,height=.2\textheight]{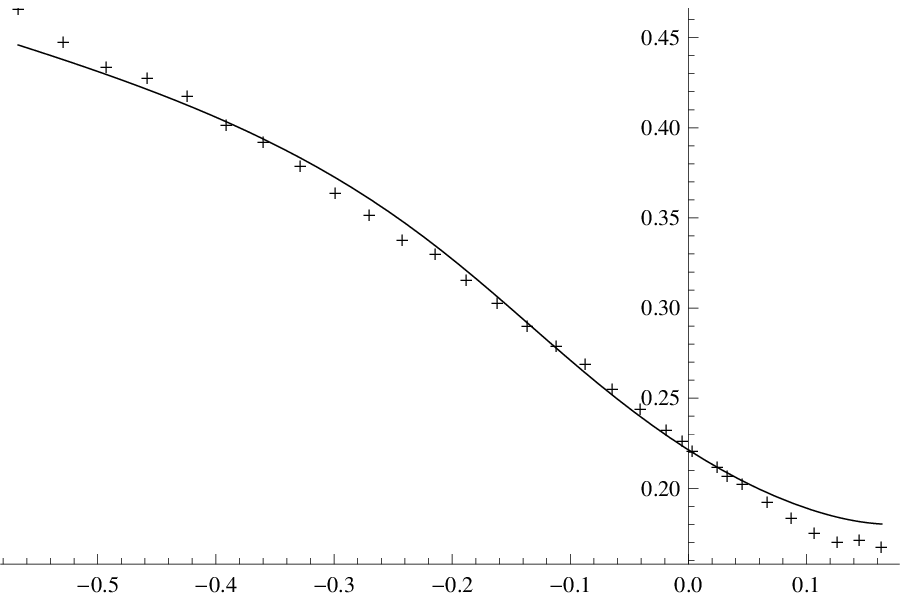} &
\includegraphics[width=.33\textwidth,height=.2\textheight]{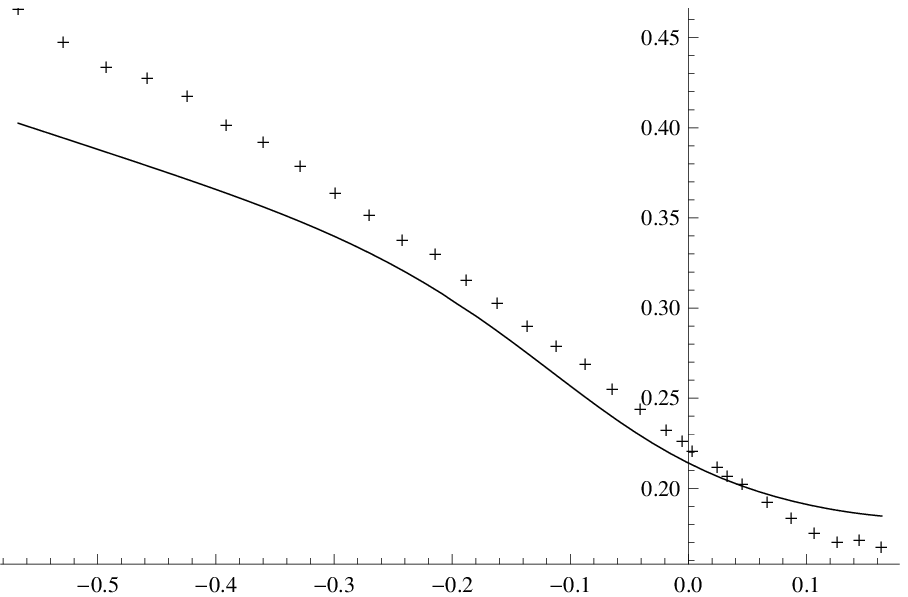} &
\includegraphics[width=.33\textwidth,height=.2\textheight]{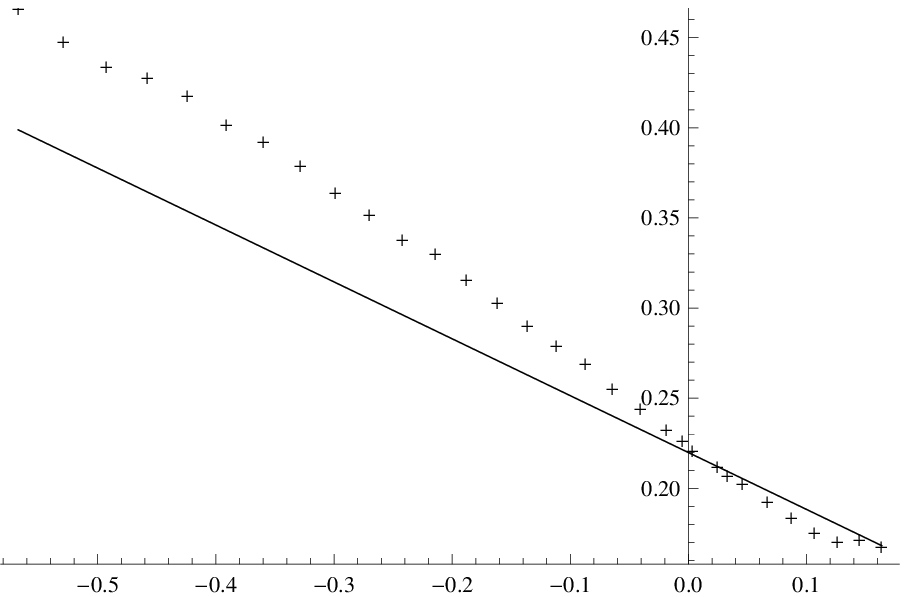} \\ \hline
\end{tabular}
\caption{Implied volatility fit to S\&P500 index options from October 1, 2010.}
\label{fig:ImpVol2010Oct}
\end{figure}
%figure generated from mathematica notebook ``FitToData3.nb''

%\clearpage
\begin{figure}
\centering
\begin{tabular}{ | c | c | c | }
\hline
Extended Merton & Merton & FMR-SV \\ \hline
59 days-to-maturity & 59 days-to-maturity & 59 days-to-maturity  \\
\includegraphics[width=.33\textwidth,height=.118\textheight]{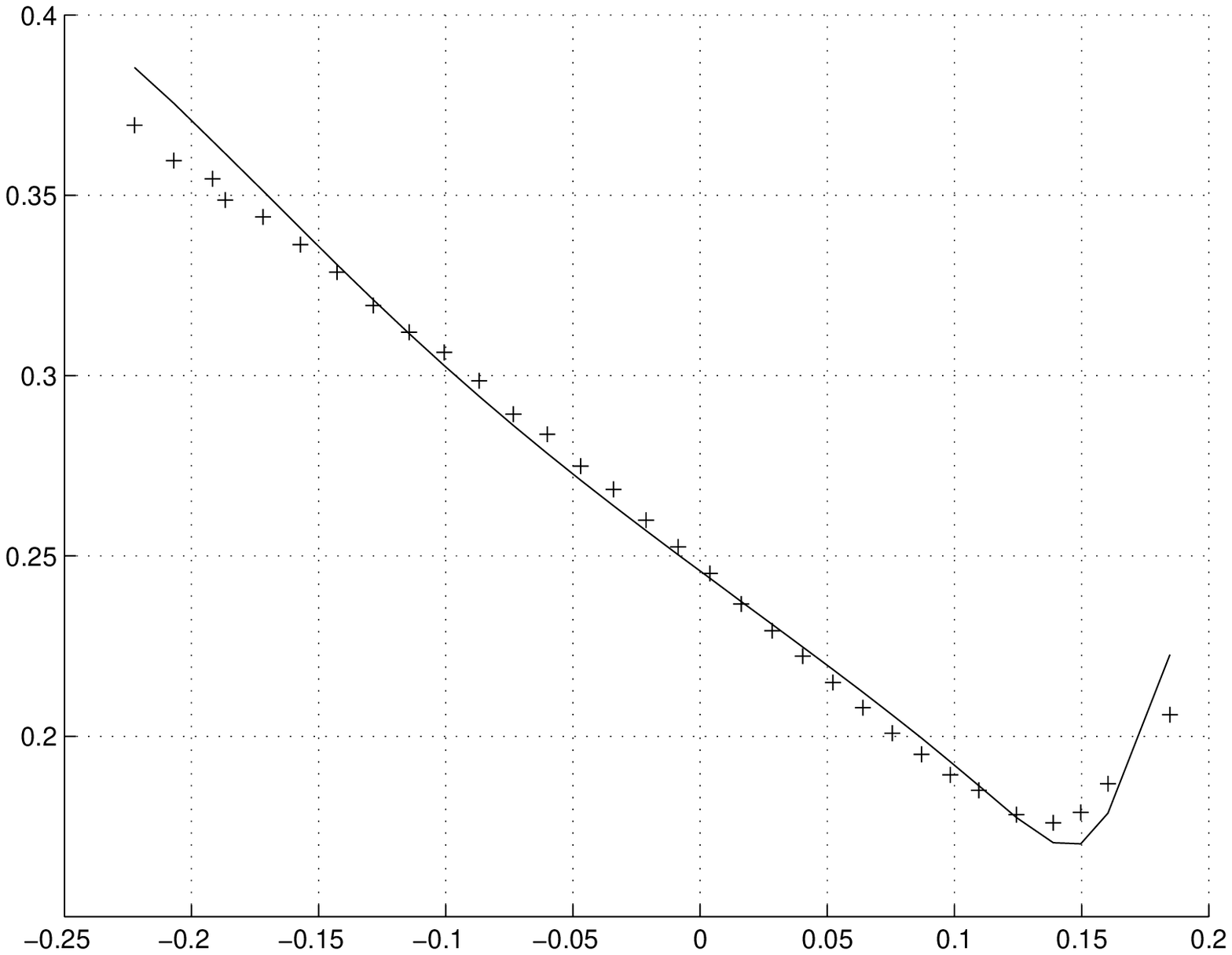} &
\includegraphics[width=.33\textwidth,height=.118\textheight]{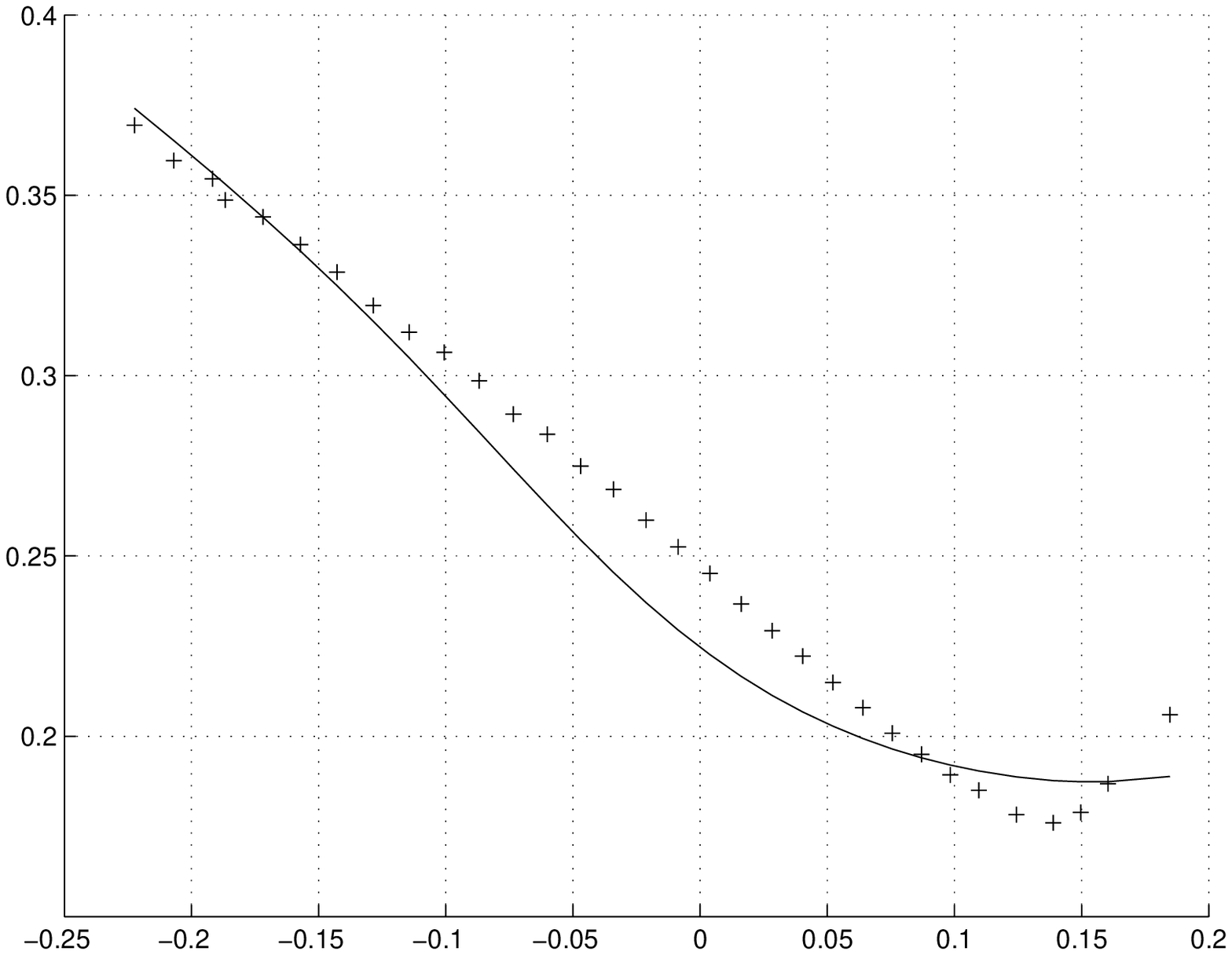} &
\includegraphics[width=.33\textwidth,height=.118\textheight]{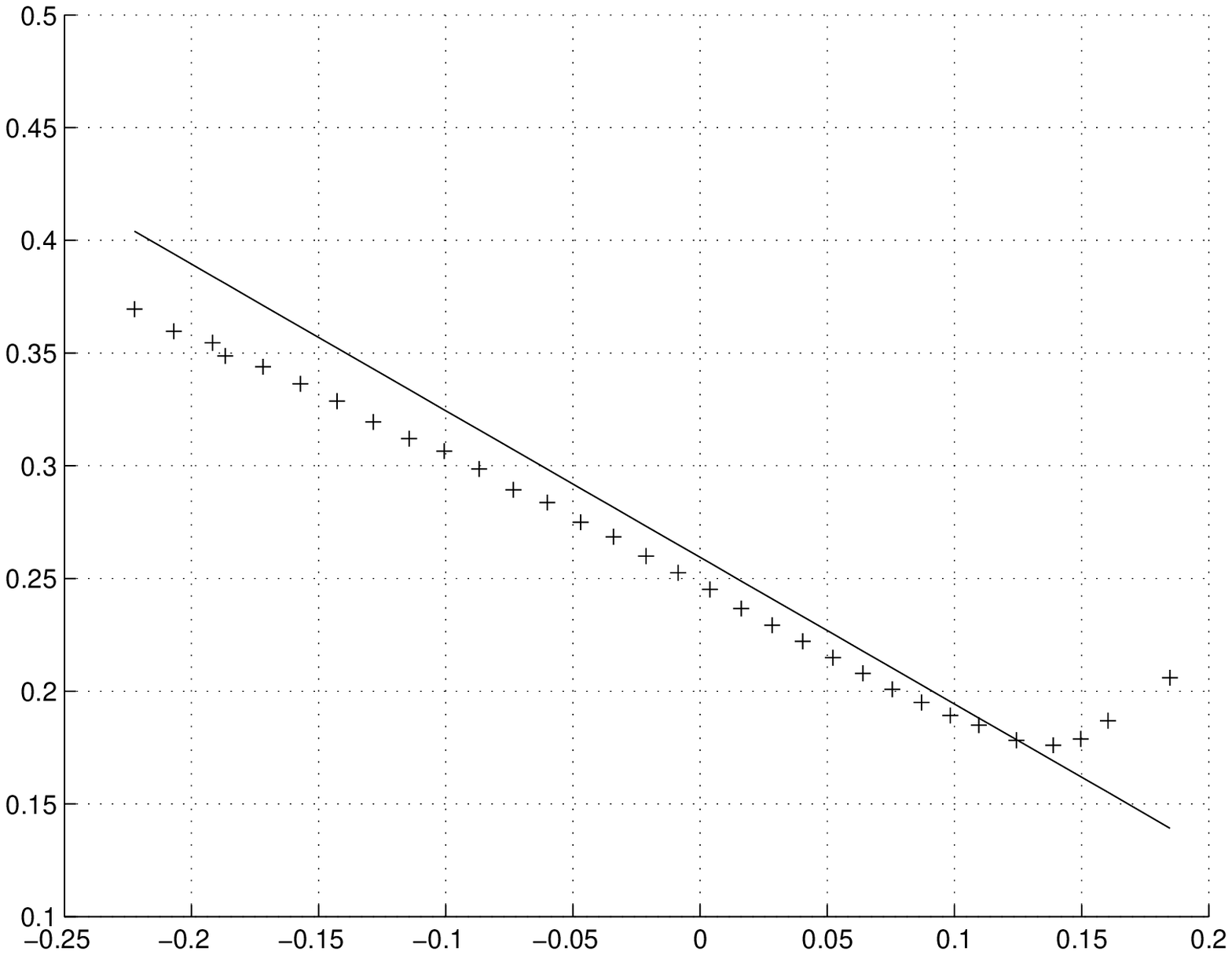} \\ \hline
88 days-to-maturity& 88 days-to-maturity & 88 days-to-maturity  \\
\includegraphics[width=.33\textwidth,height=.118\textheight]{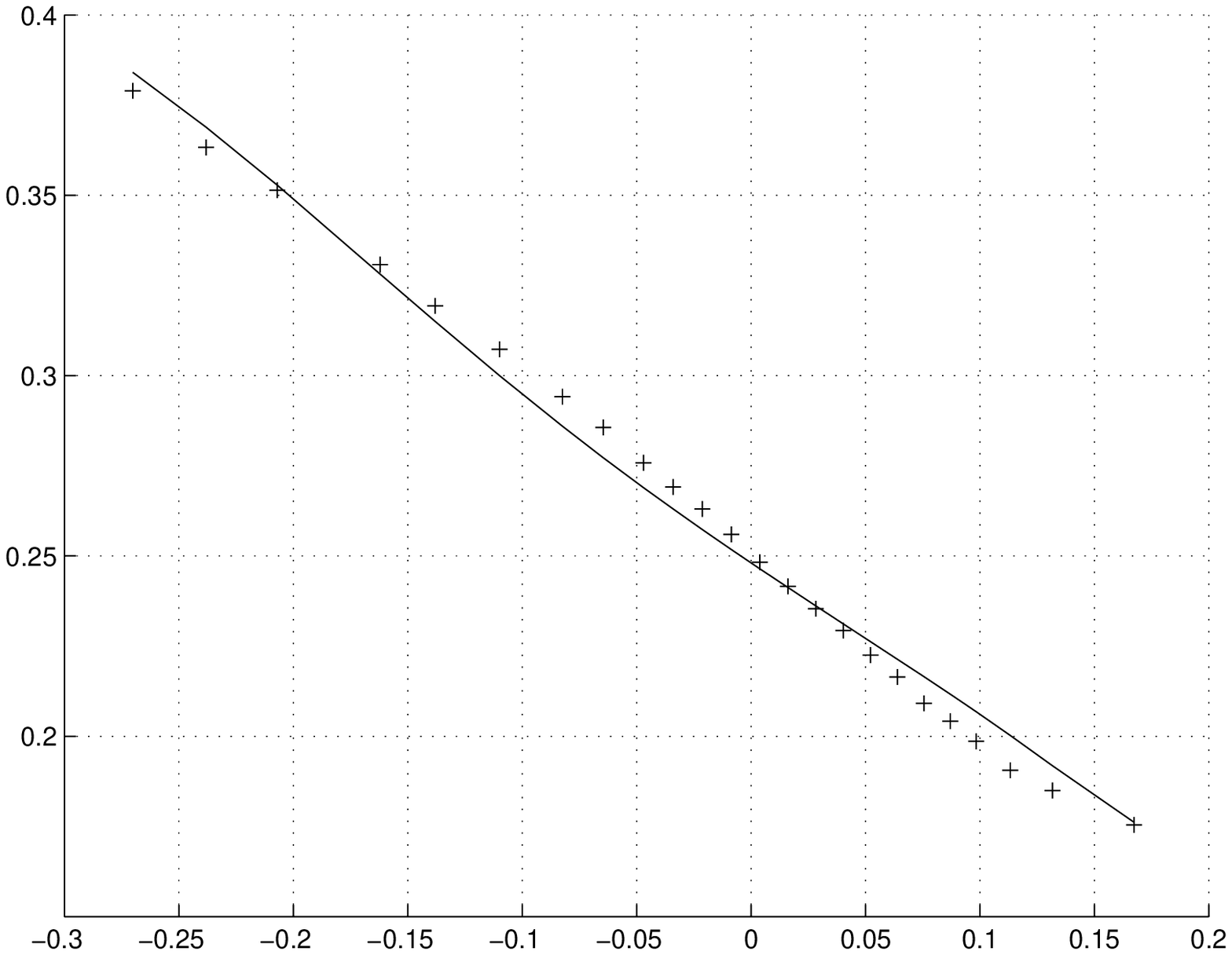} &
\includegraphics[width=.33\textwidth,height=.118\textheight]{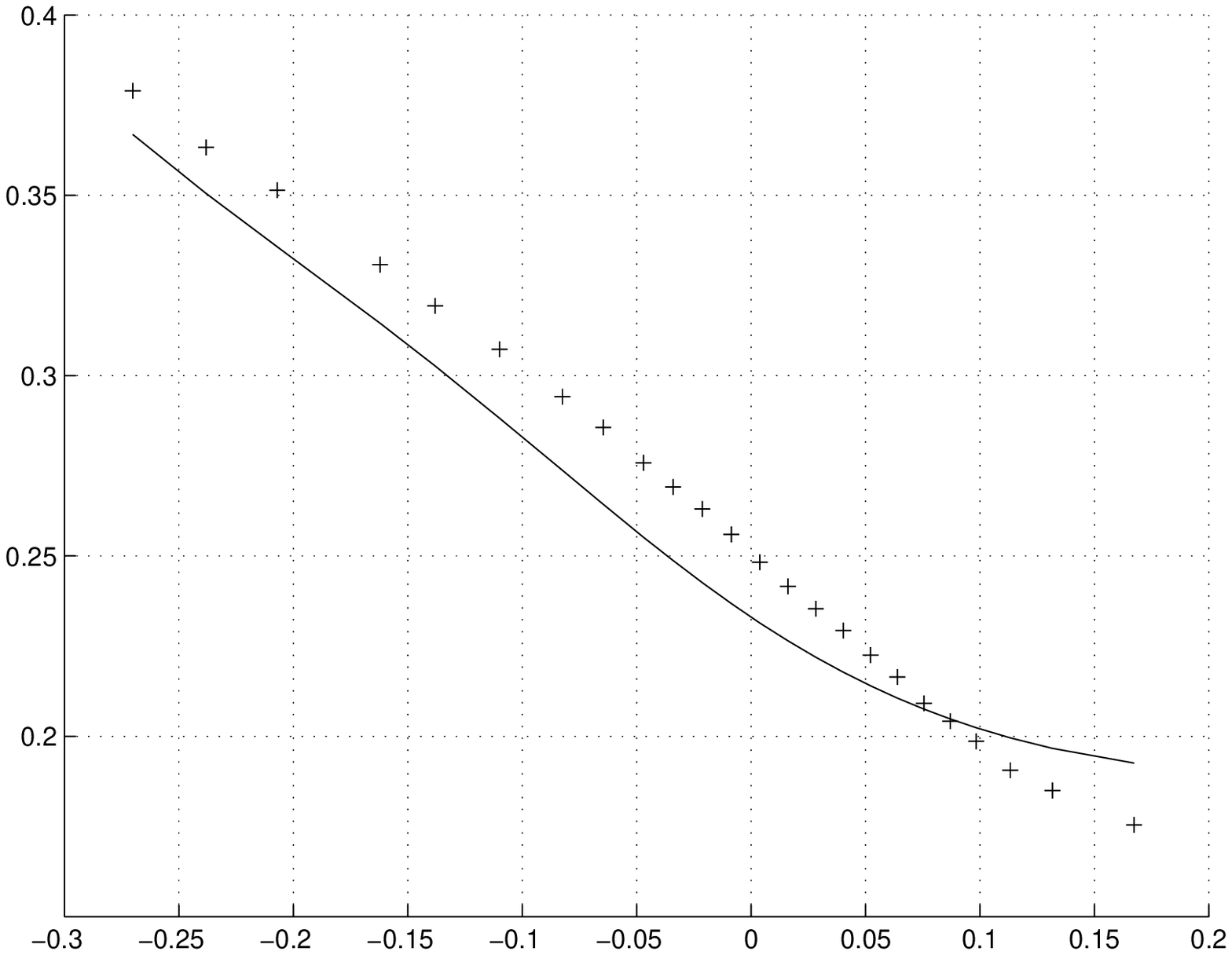} &
\includegraphics[width=.33\textwidth,height=.118\textheight]{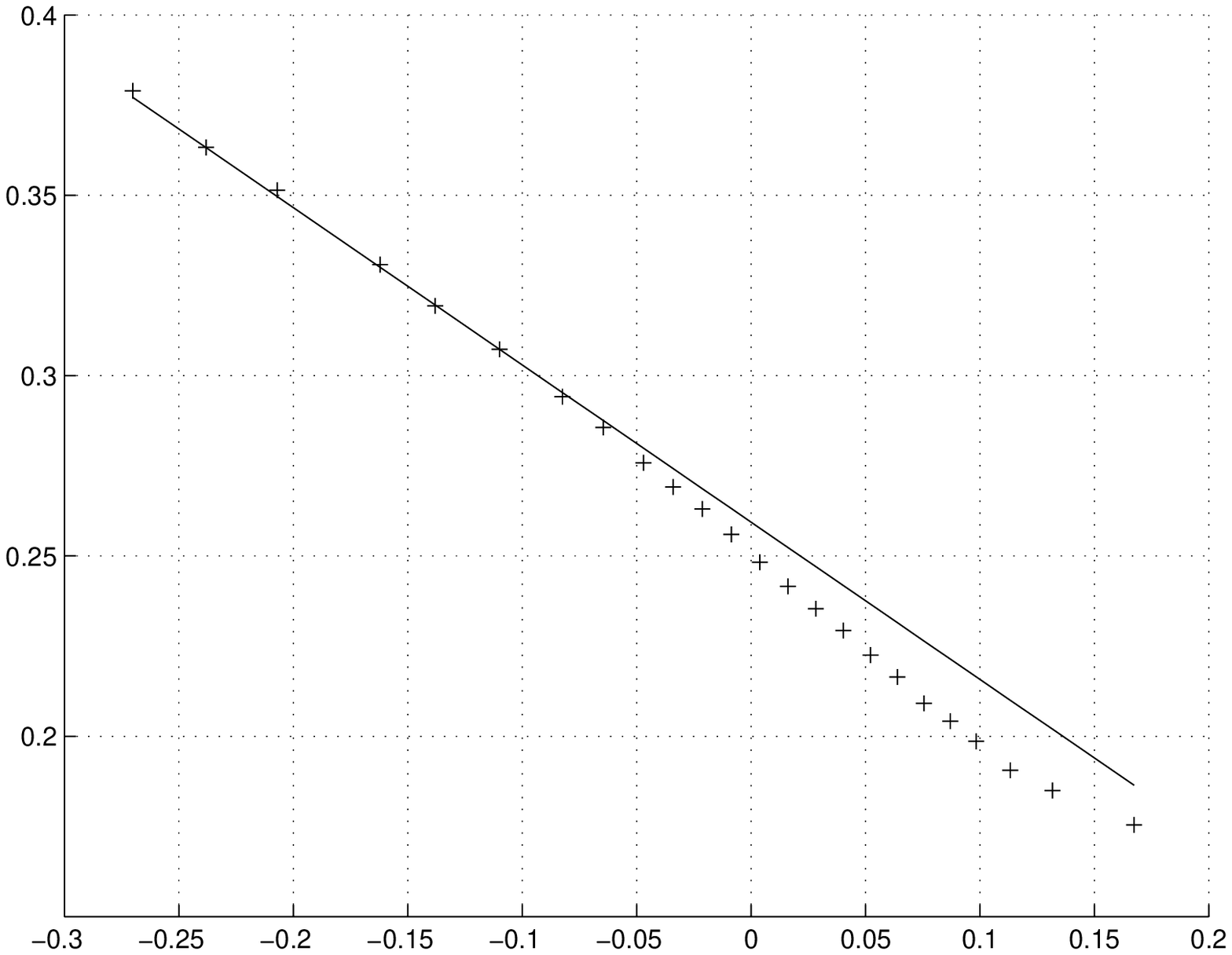} \\ \hline
122 days-to-maturity & 122 days-to-maturity & 122 days-to-maturity  \\
\includegraphics[width=.33\textwidth,height=.118\textheight]{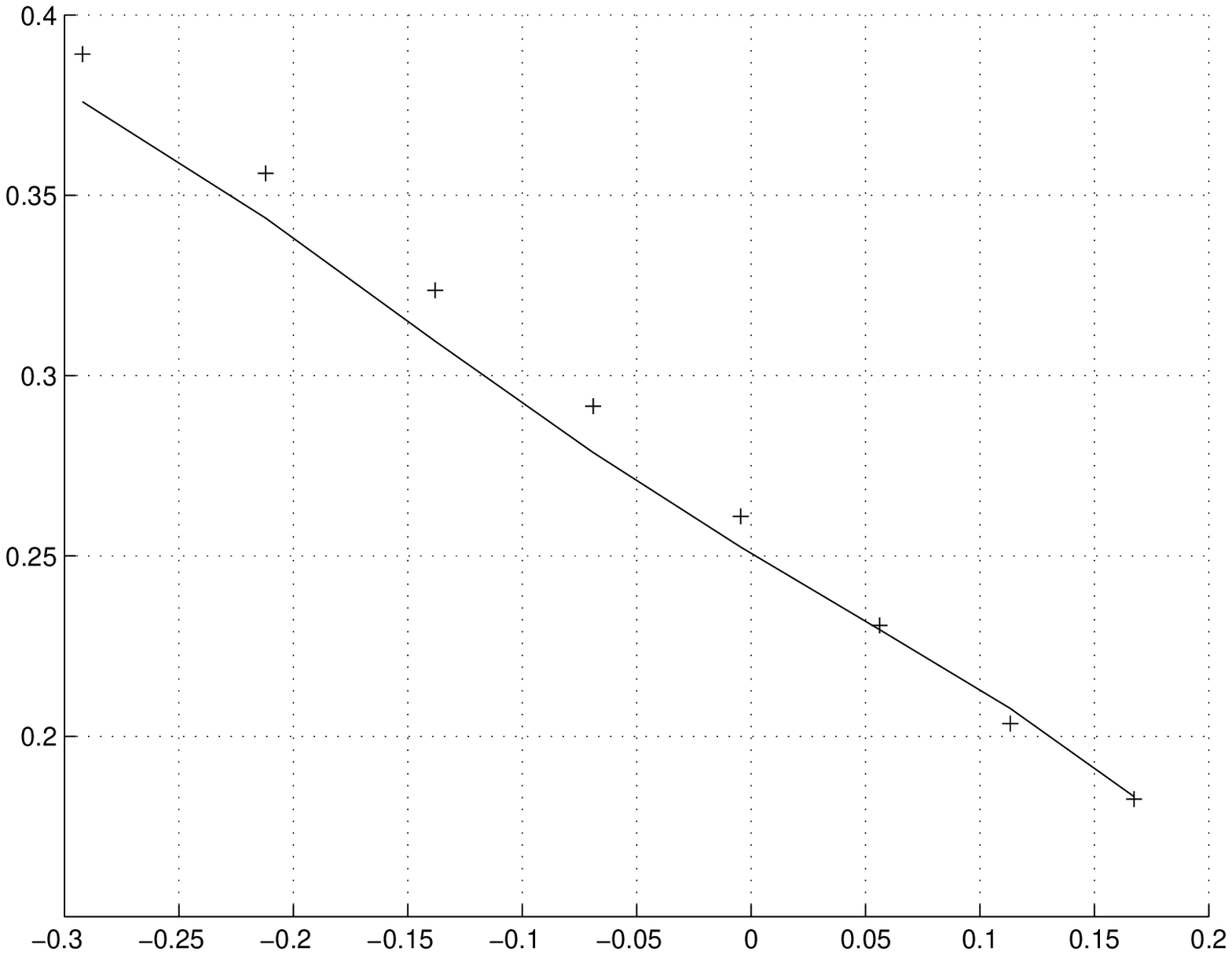} &
\includegraphics[width=.33\textwidth,height=.118\textheight]{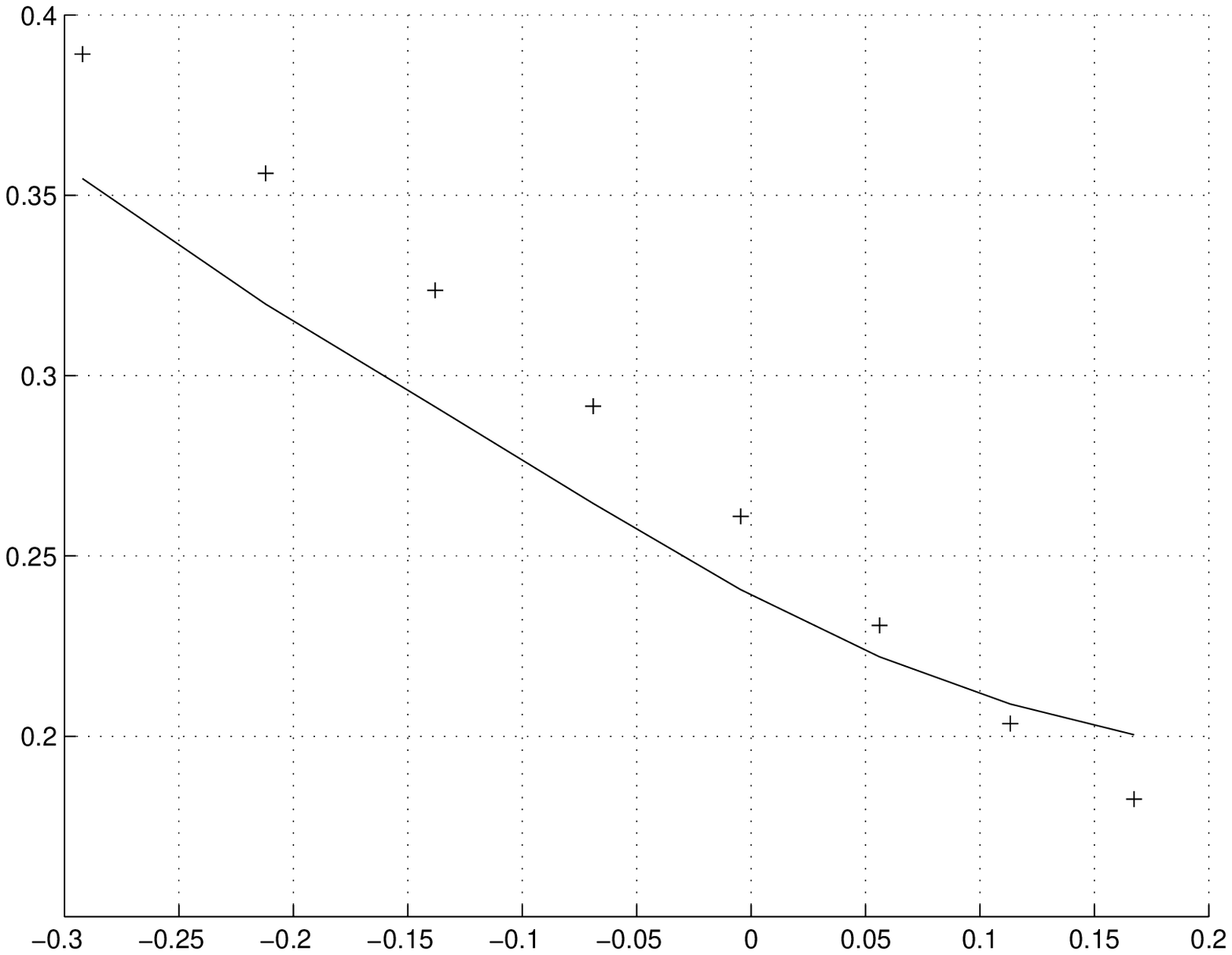} &
\includegraphics[width=.33\textwidth,height=.118\textheight]{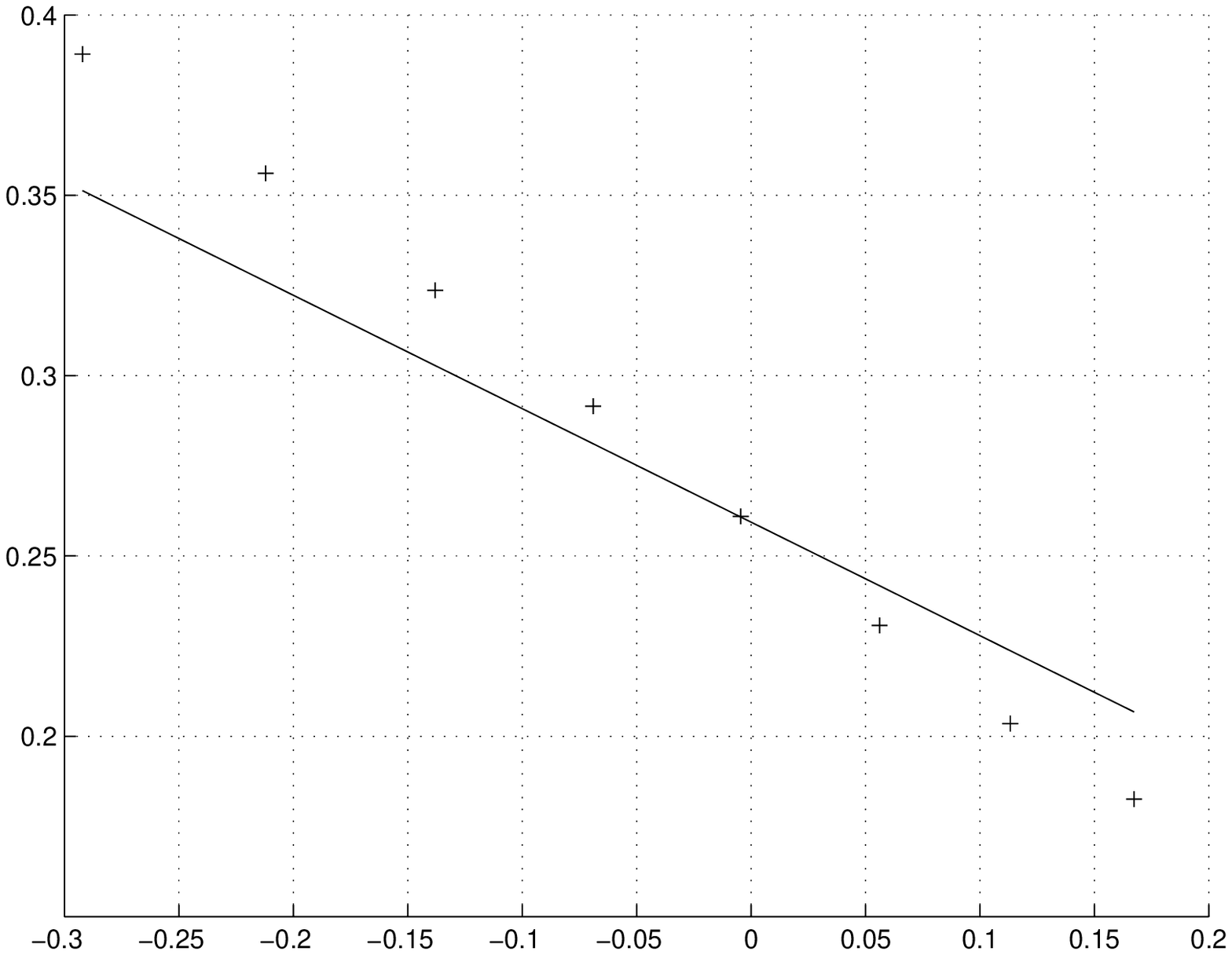} \\ \hline
177 days-to-maturity & 177 days-to-maturity &177 days-to-maturity  \\
\includegraphics[width=.33\textwidth,height=.118\textheight]{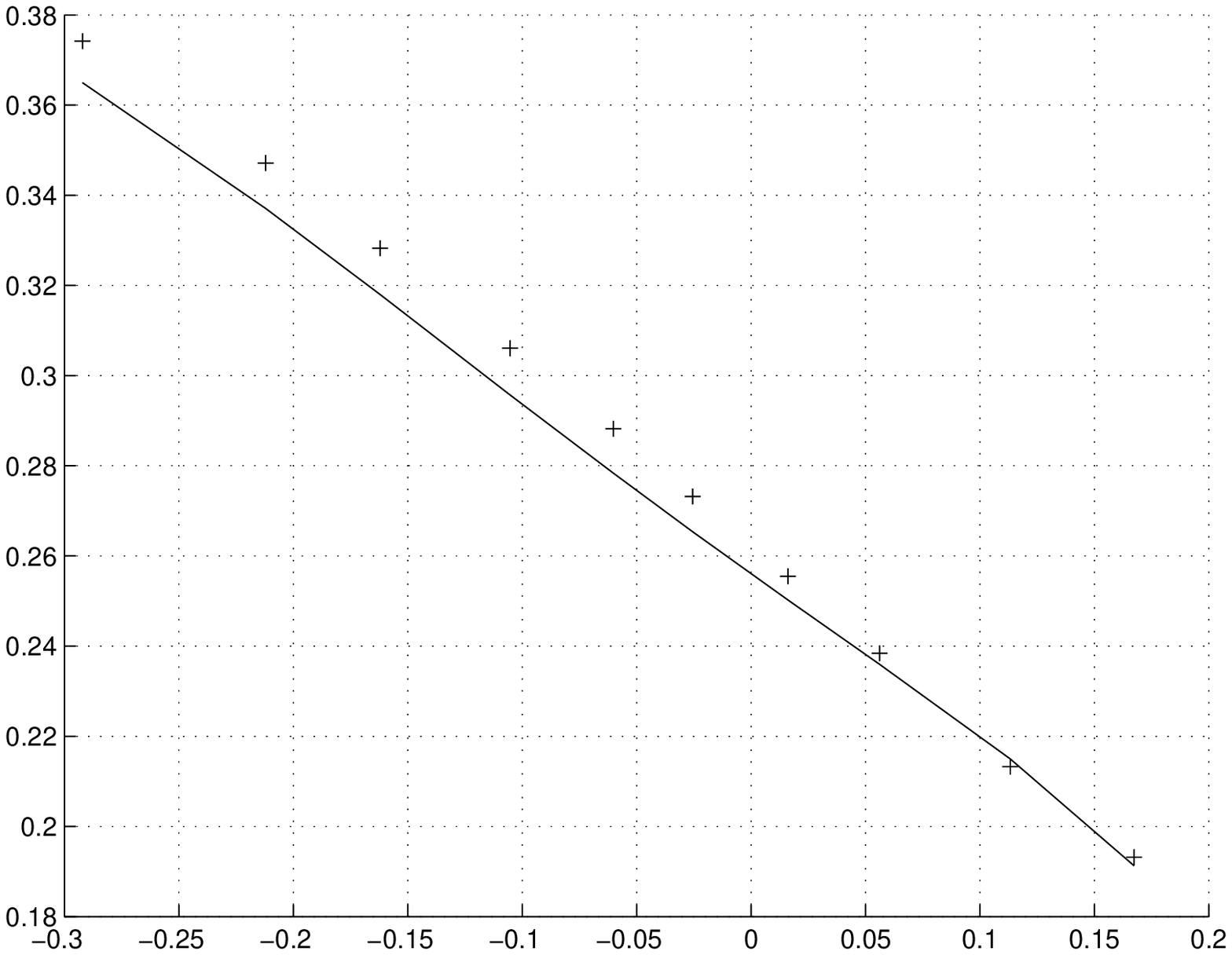} &
\includegraphics[width=.33\textwidth,height=.118\textheight]{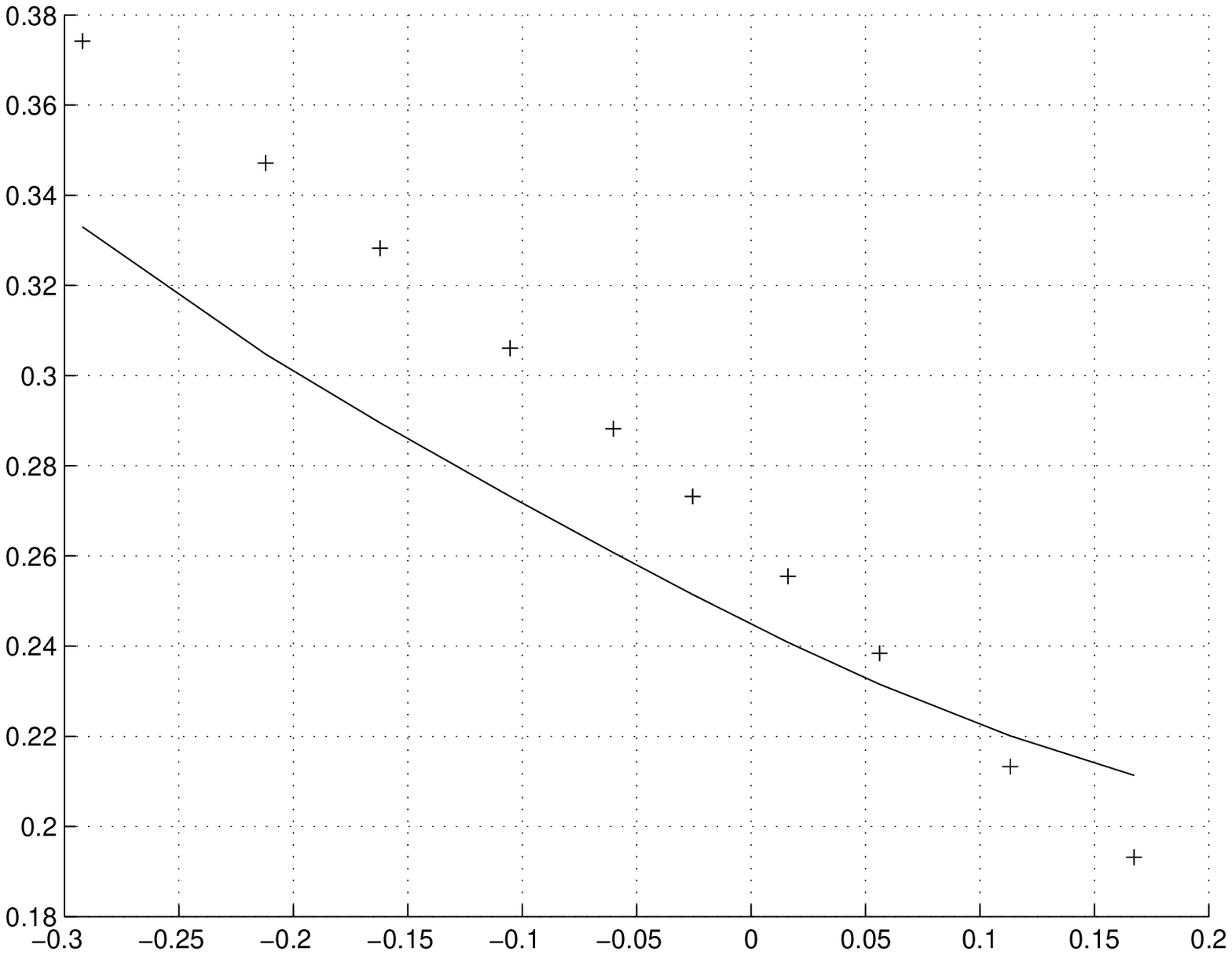} &
\includegraphics[width=.33\textwidth,height=.118\textheight]{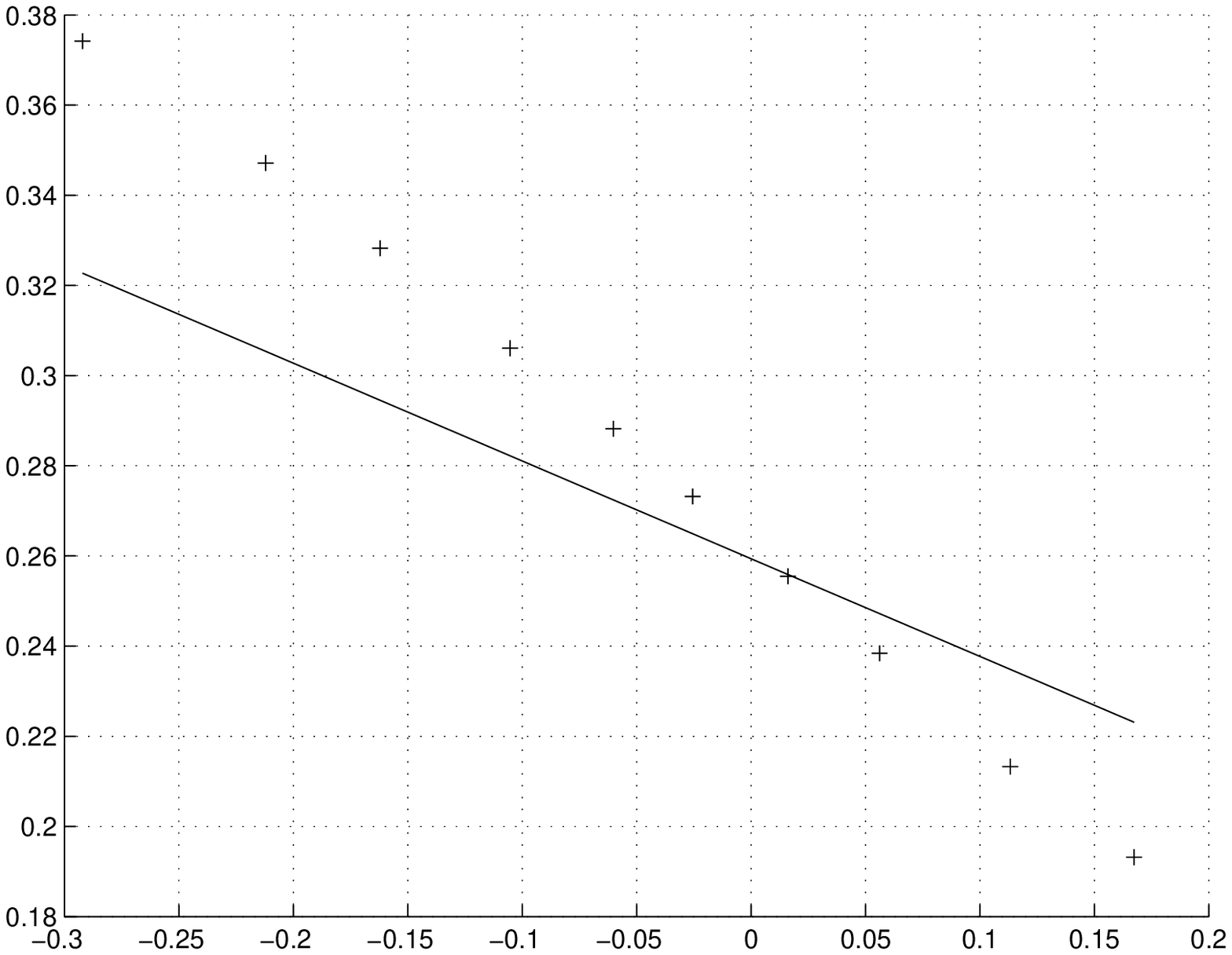} \\ \hline
273 days-to-maturity& 273 days-to-maturity & 273 days-to-maturity  \\
\includegraphics[width=.33\textwidth,height=.118\textheight]{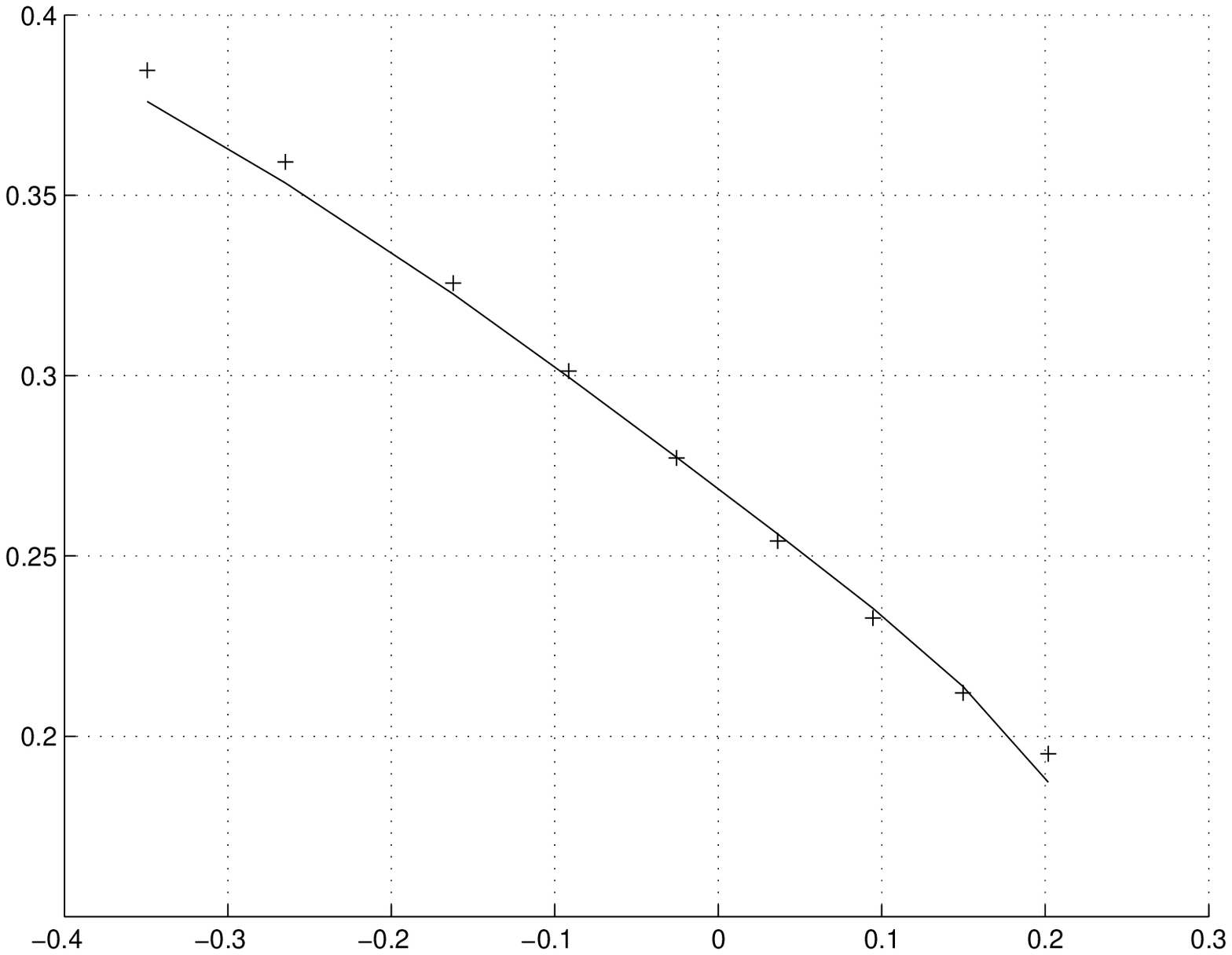} &
\includegraphics[width=.33\textwidth,height=.118\textheight]{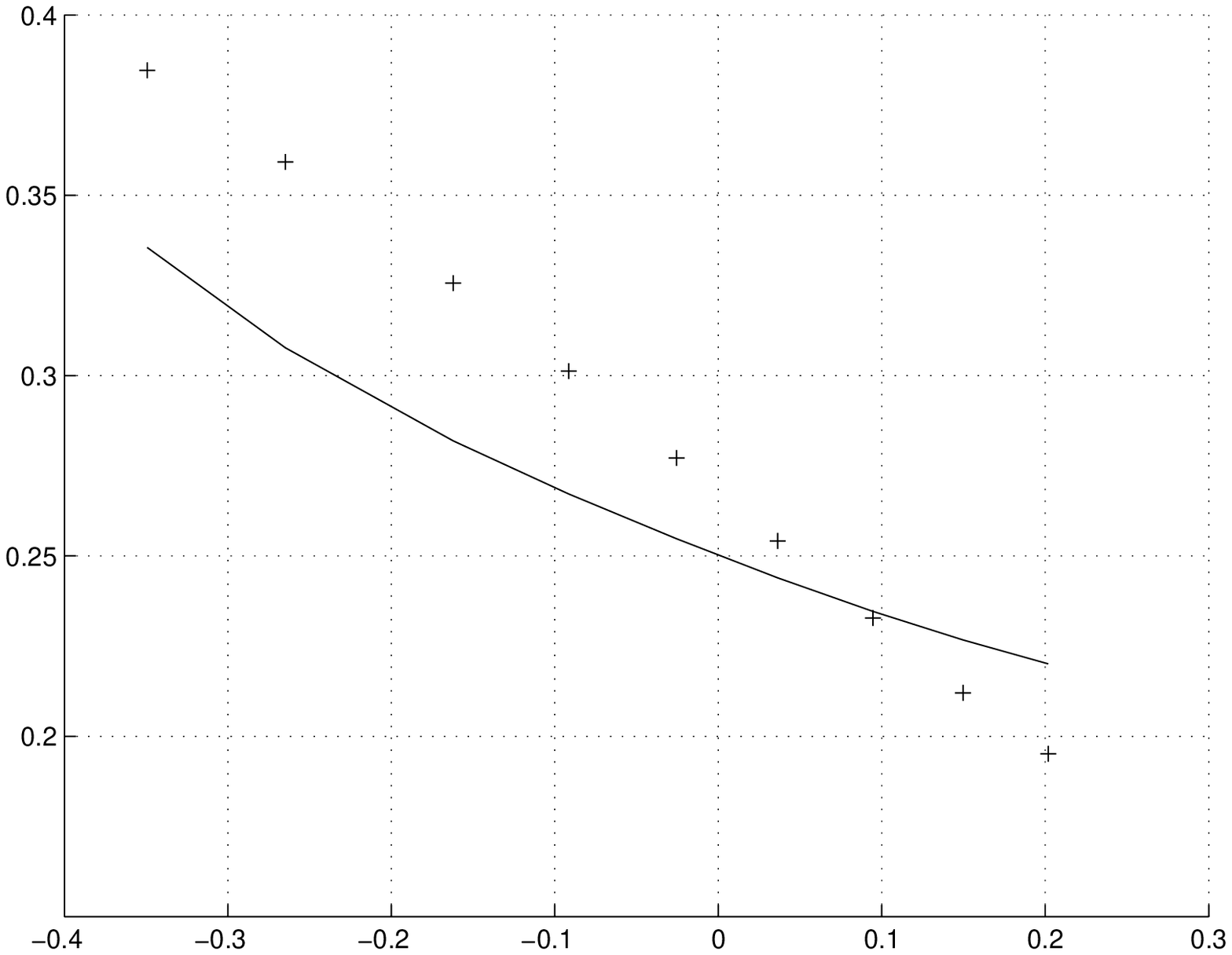} &
\includegraphics[width=.33\textwidth,height=.118\textheight]{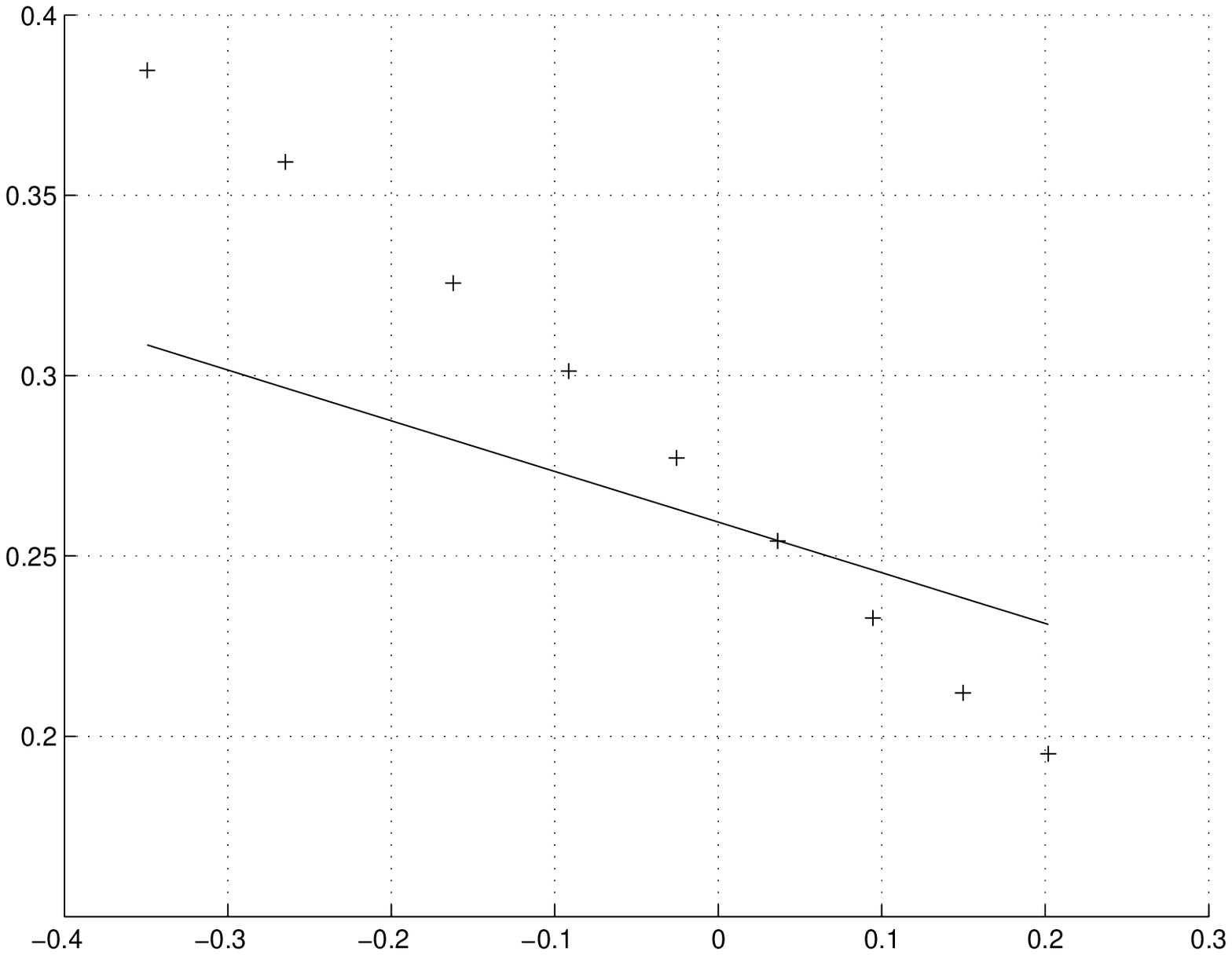} \\ \hline
363 days-to-maturity & 363 days-to-maturity & 363 days-to-maturity  \\
\includegraphics[width=.33\textwidth,height=.118\textheight]{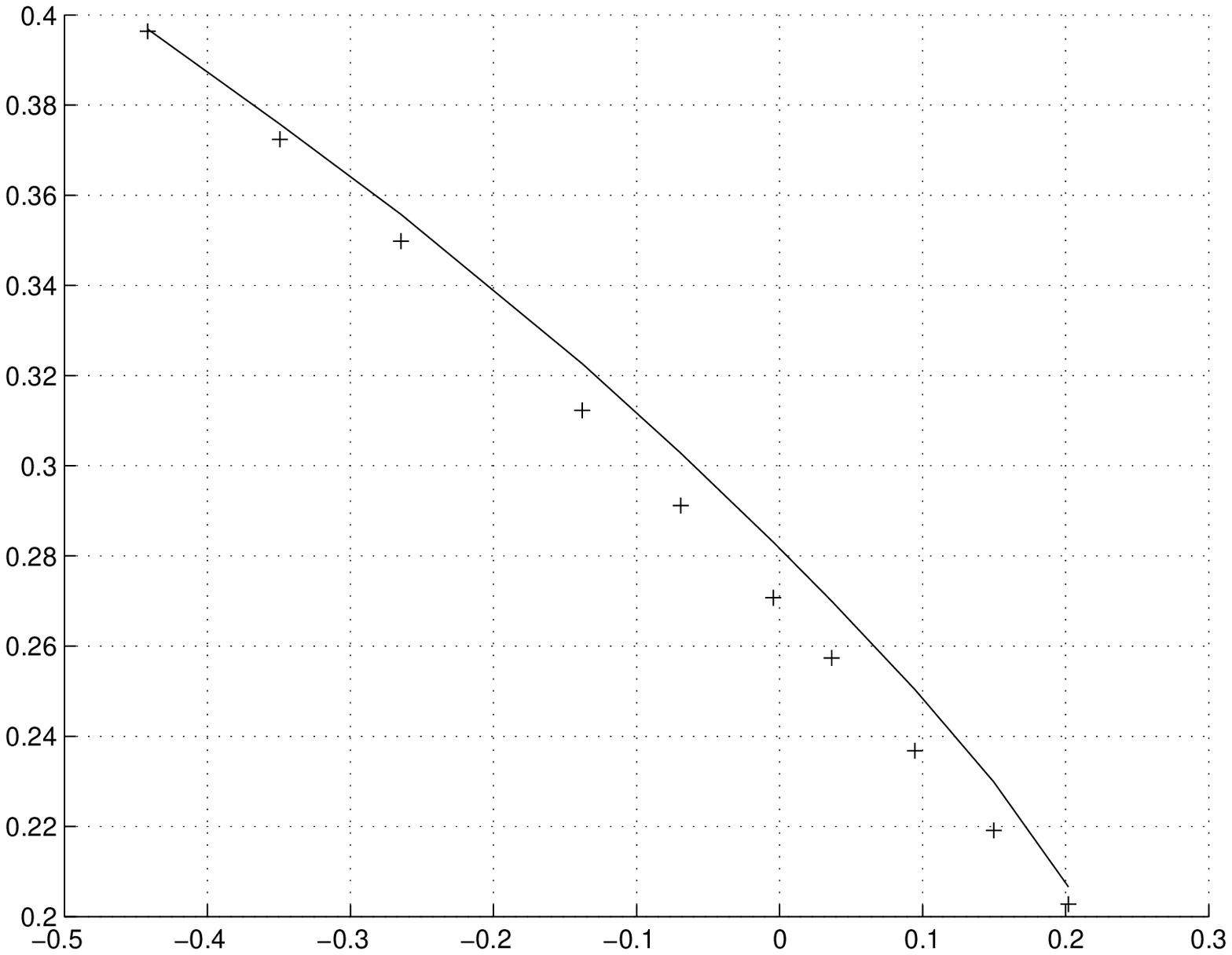} &
\includegraphics[width=.33\textwidth,height=.118\textheight]{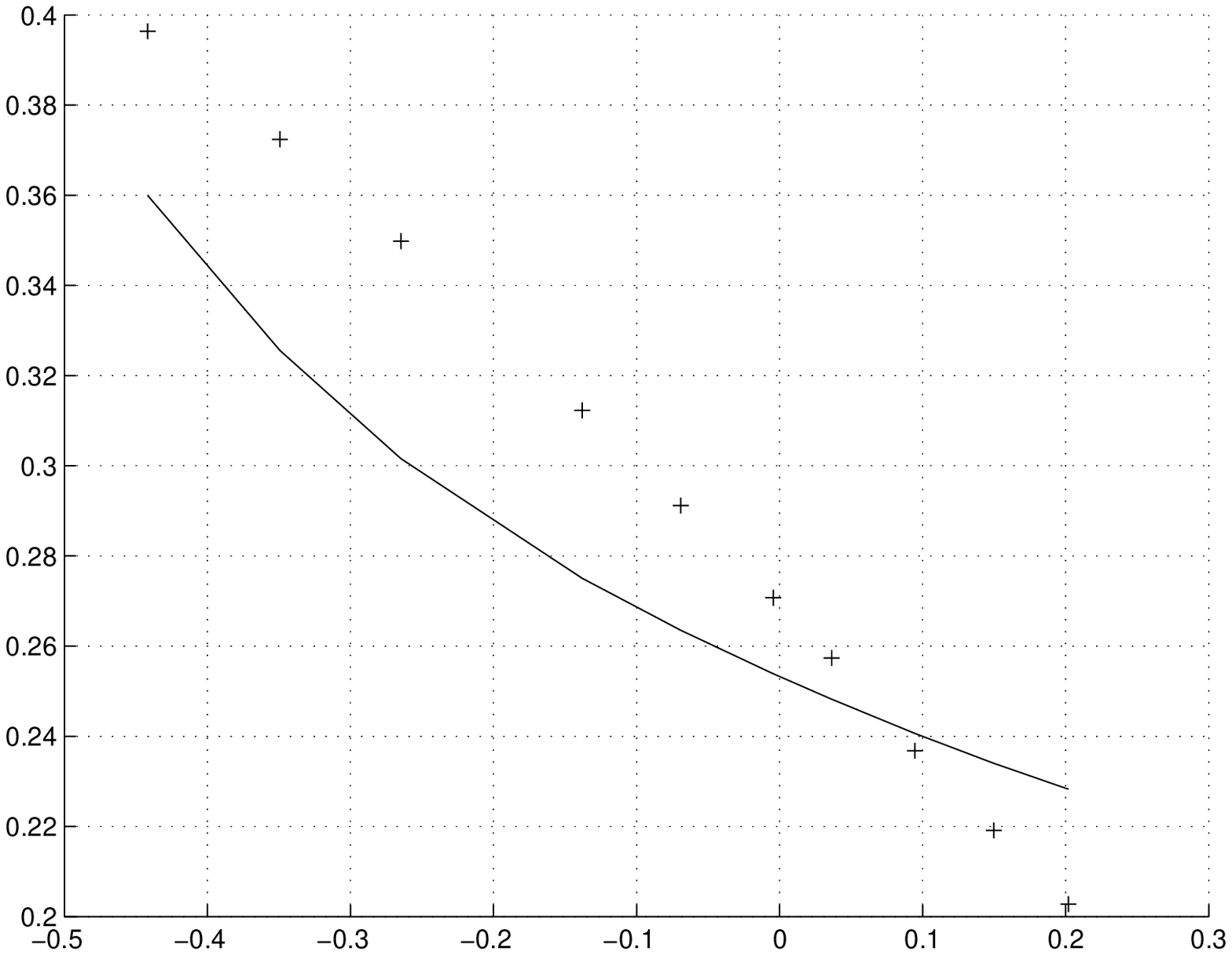} &
\includegraphics[width=.33\textwidth,height=.118\textheight]{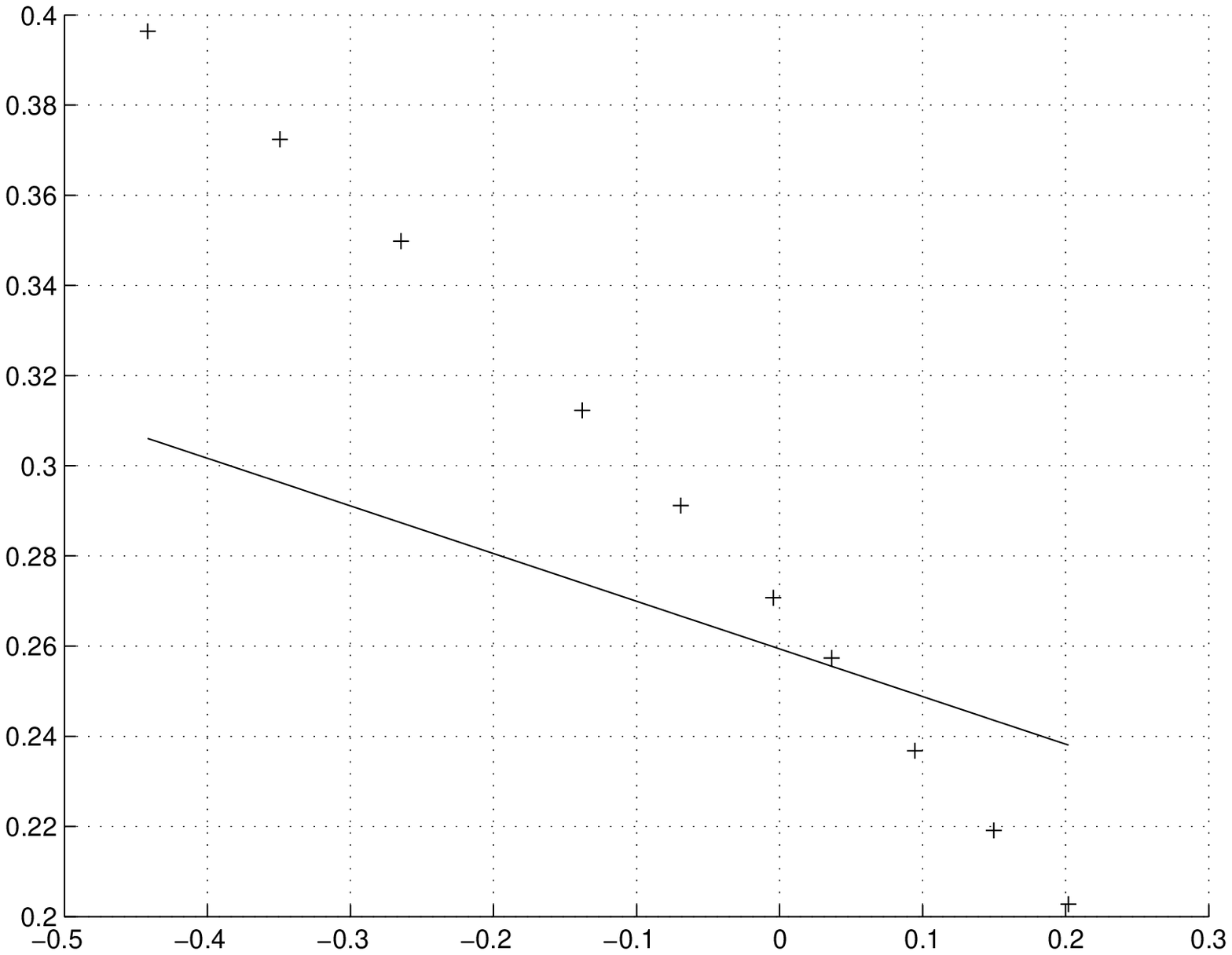} \\ \hline
\end{tabular}
\caption{Implied volatility fit to S\&P500 index options from December 19, 2011.}
\label{fig:merton1}
\end{figure}

%\clearpage
\begin{figure}
\centering
\begin{tabular}{ | c | c | c | }
\hline
Extended Merton & Merton & FMR-SV \\ \hline
66 days-to-maturity & 66 days-to-maturity & 66 days-to-maturity  \\
\includegraphics[width=.33\textwidth,height=.12\textheight]{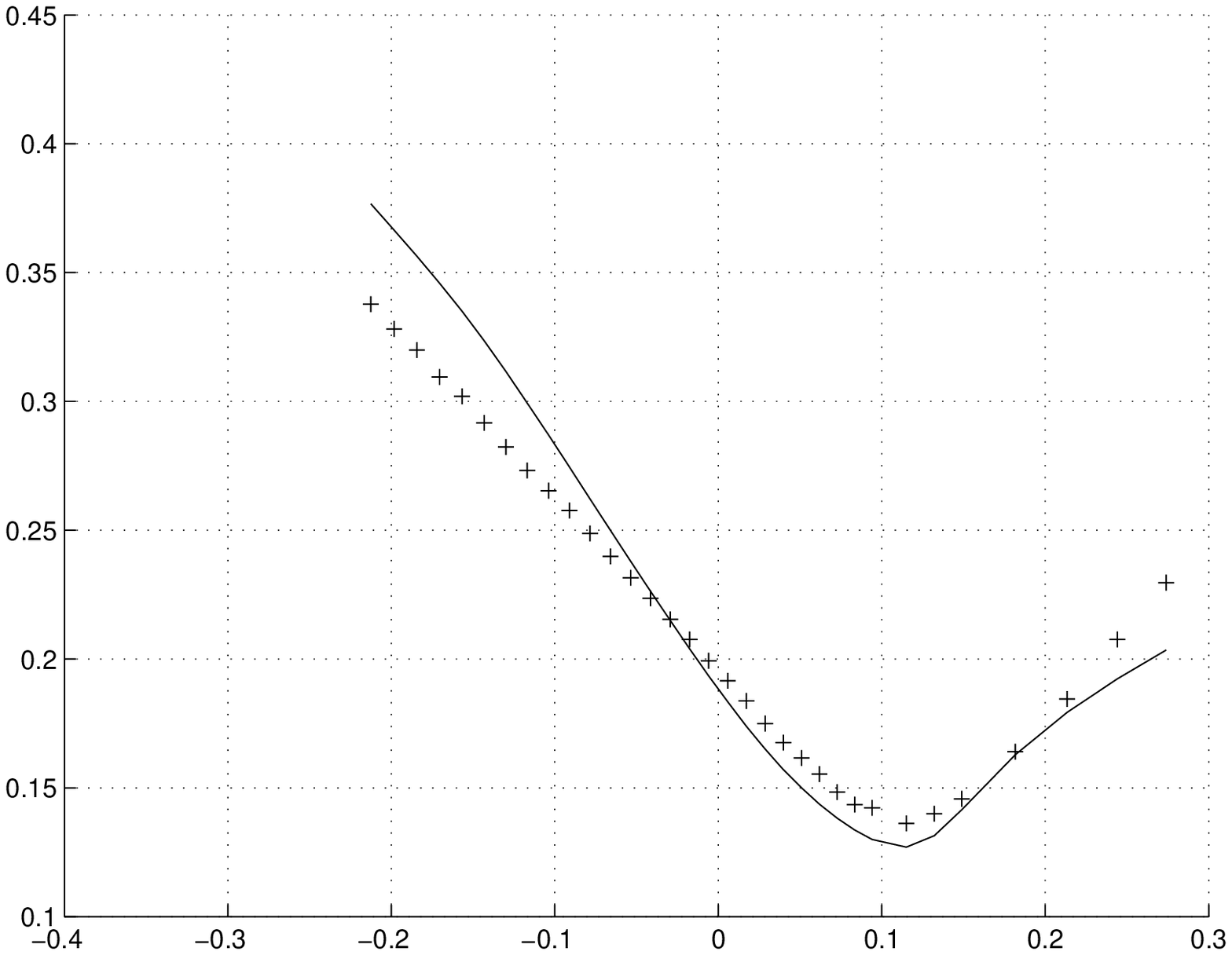} &
\includegraphics[width=.33\textwidth,height=.12\textheight]{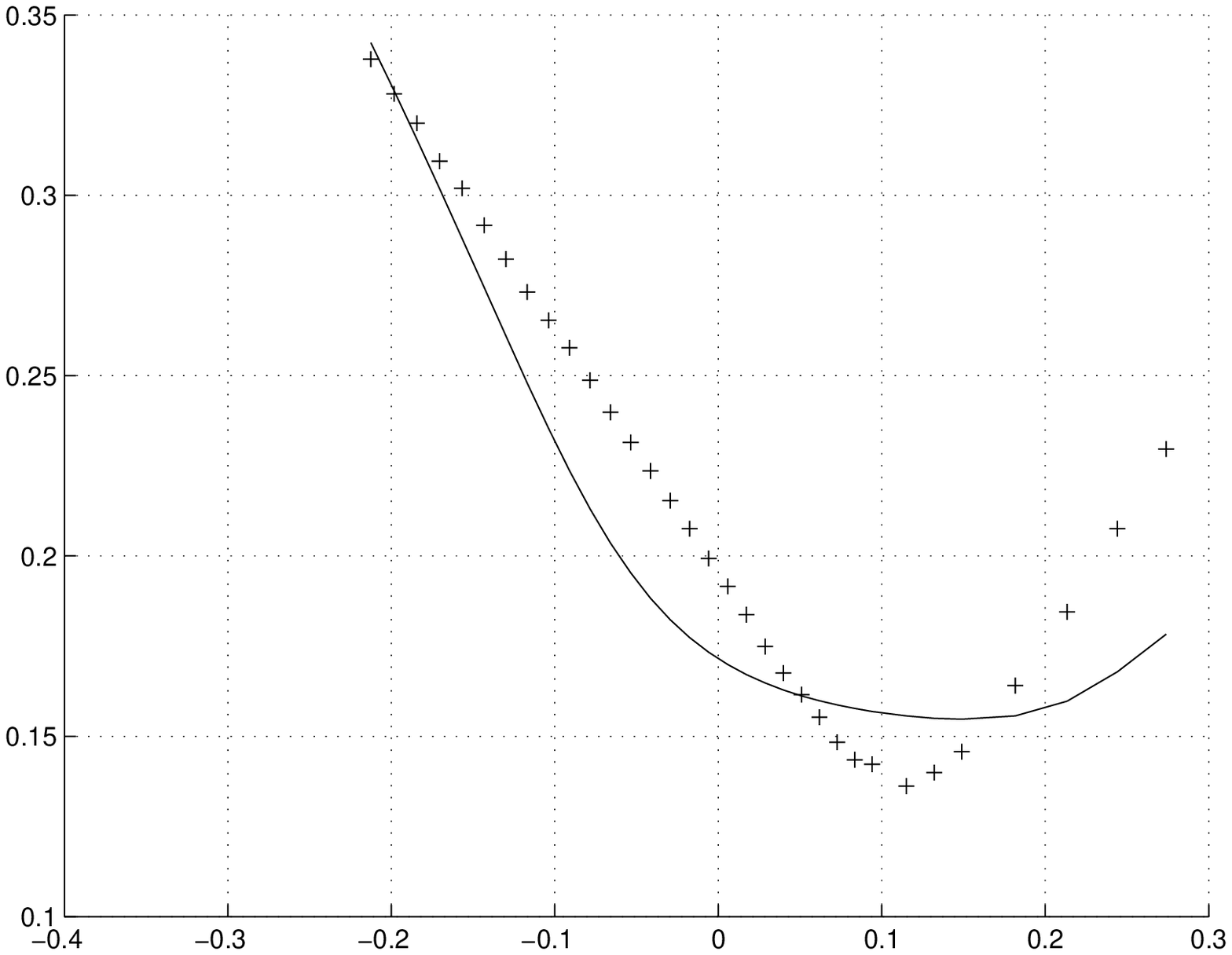} &
\includegraphics[width=.33\textwidth,height=.12\textheight]{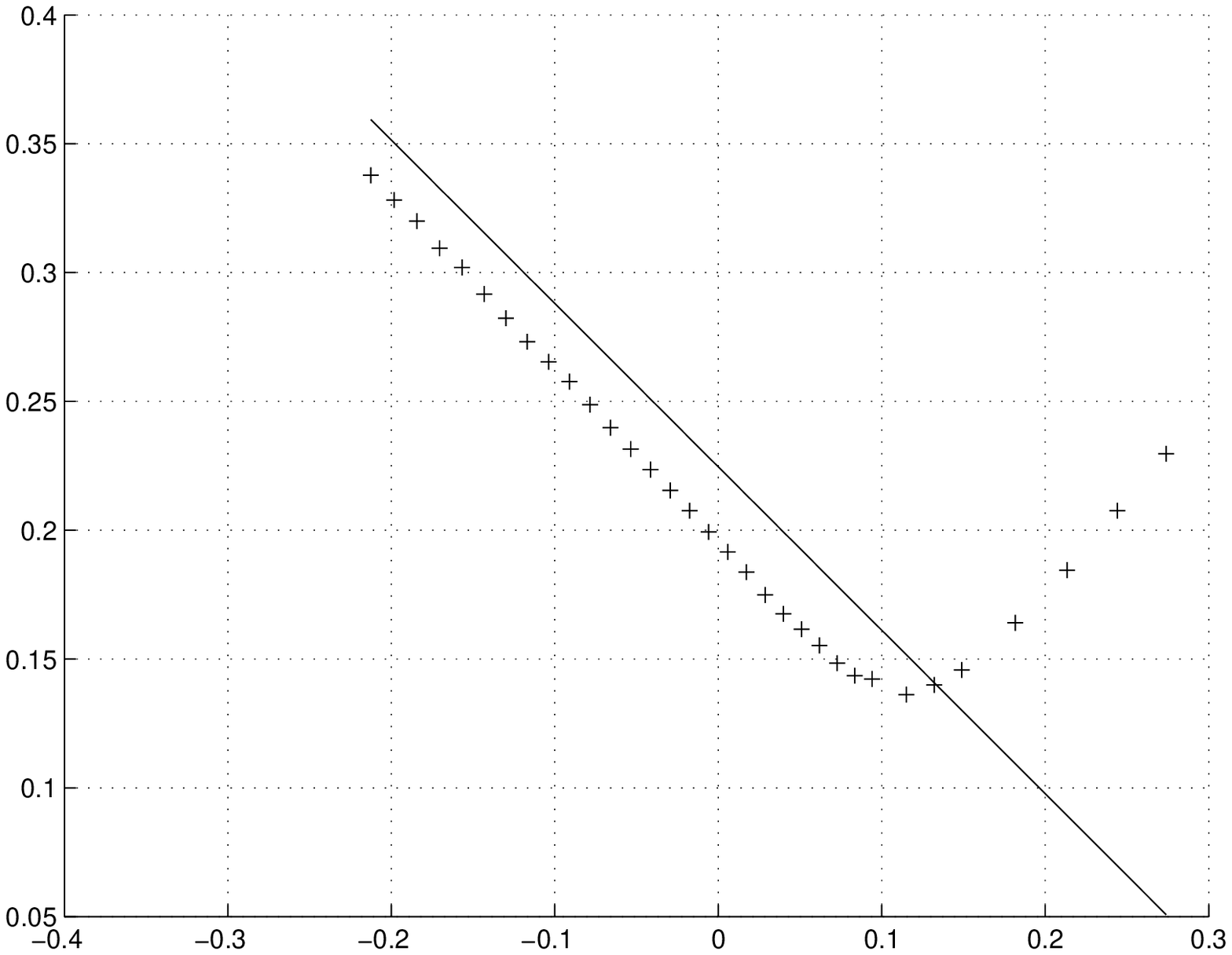} \\ \hline
100 days-to-maturity& 100 days-to-maturity & 100 days-to-maturity  \\
\includegraphics[width=.33\textwidth,height=.12\textheight]{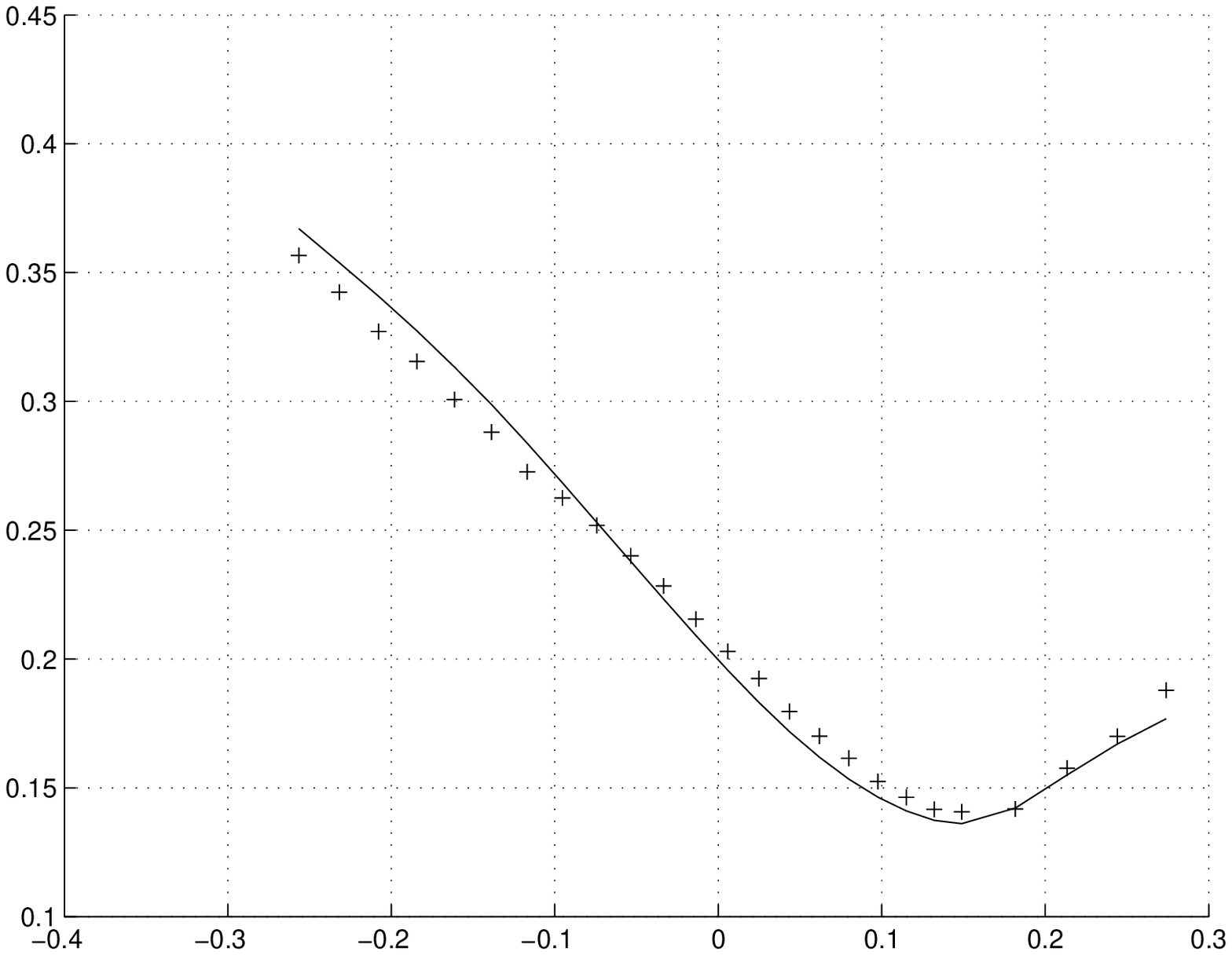} &
\includegraphics[width=.33\textwidth,height=.12\textheight]{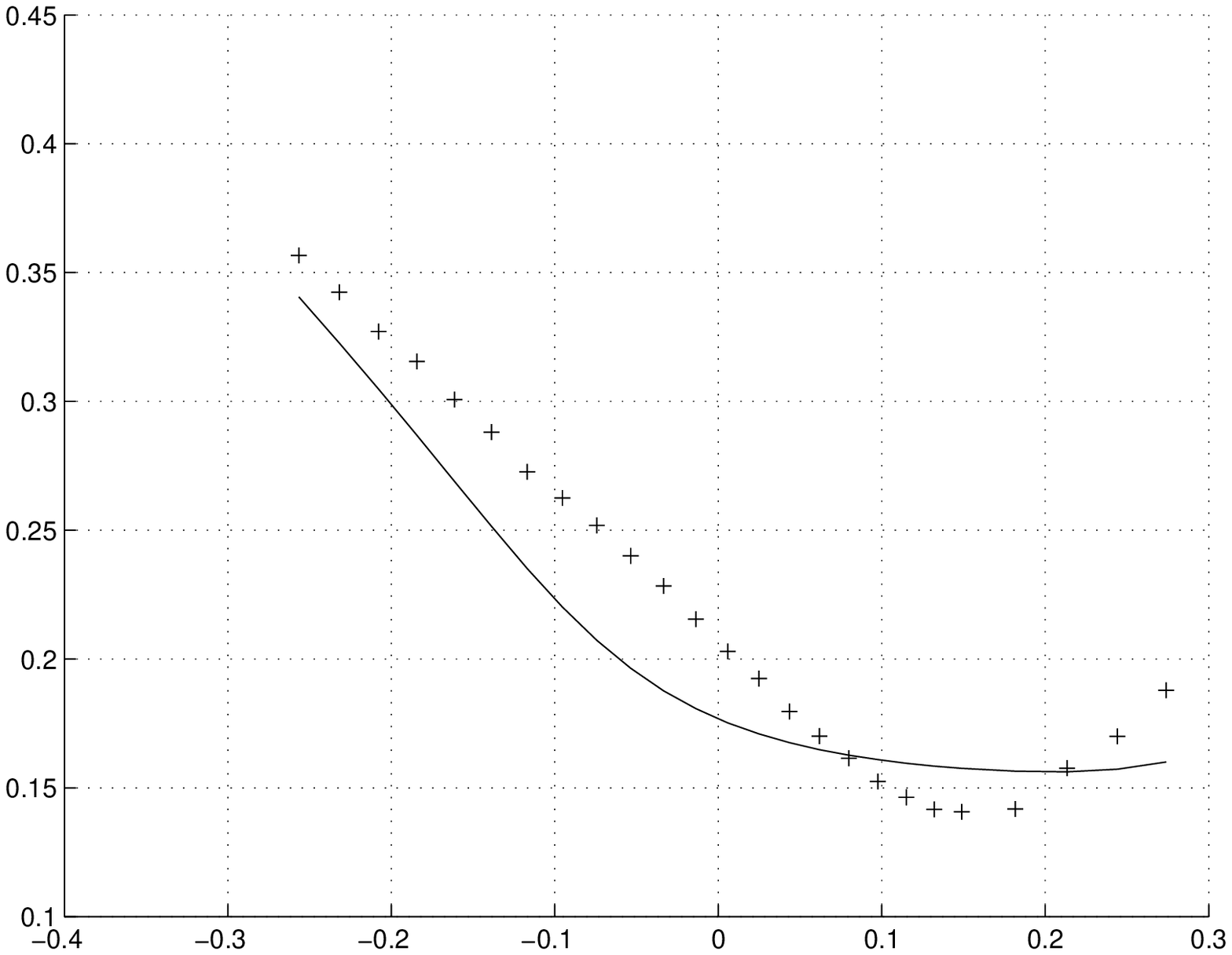} &
\includegraphics[width=.33\textwidth,height=.12\textheight]{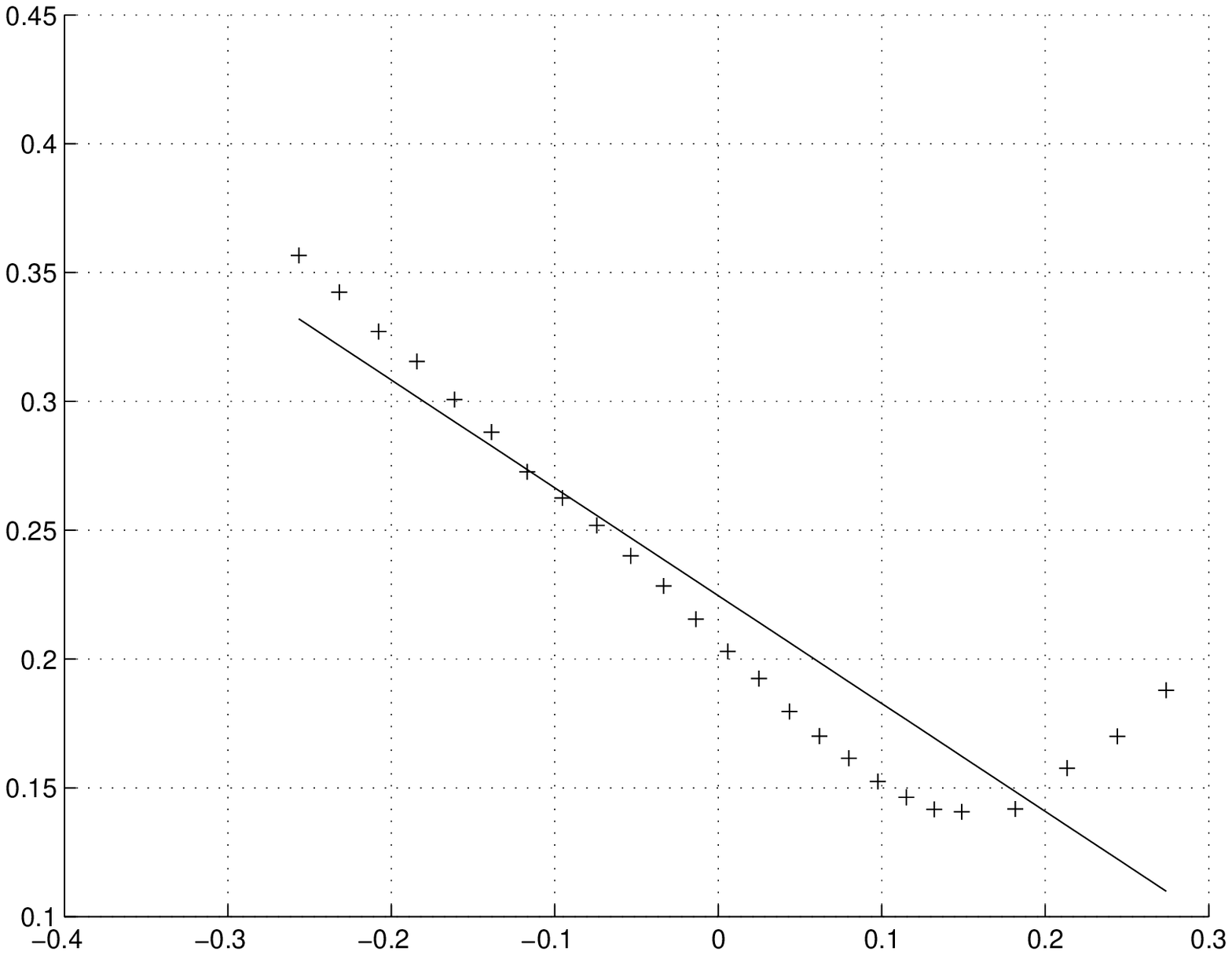} \\ \hline
155 days-to-maturity & 155 days-to-maturity & 155 days-to-maturity  \\
\includegraphics[width=.33\textwidth,height=.12\textheight]{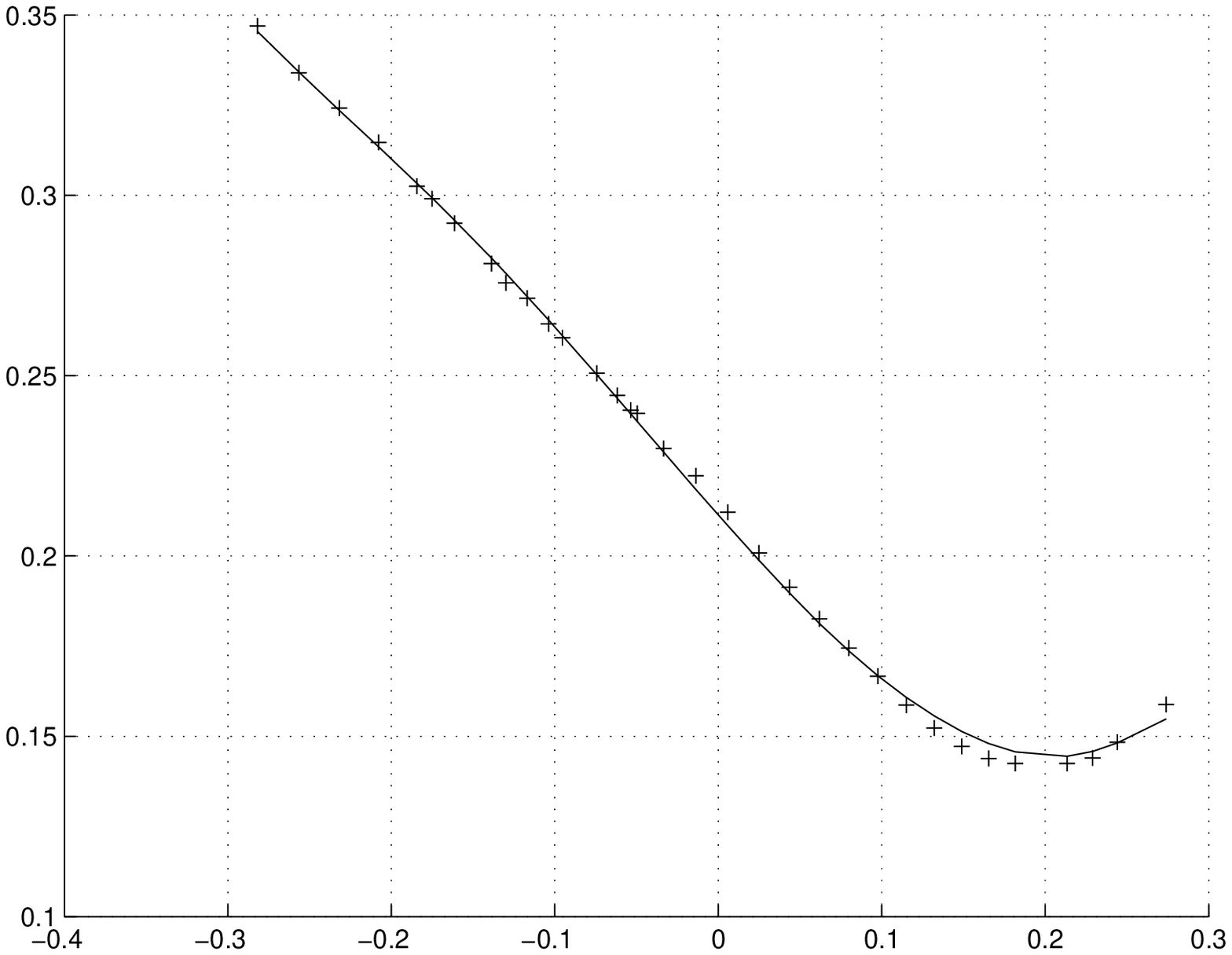} &
\includegraphics[width=.33\textwidth,height=.12\textheight]{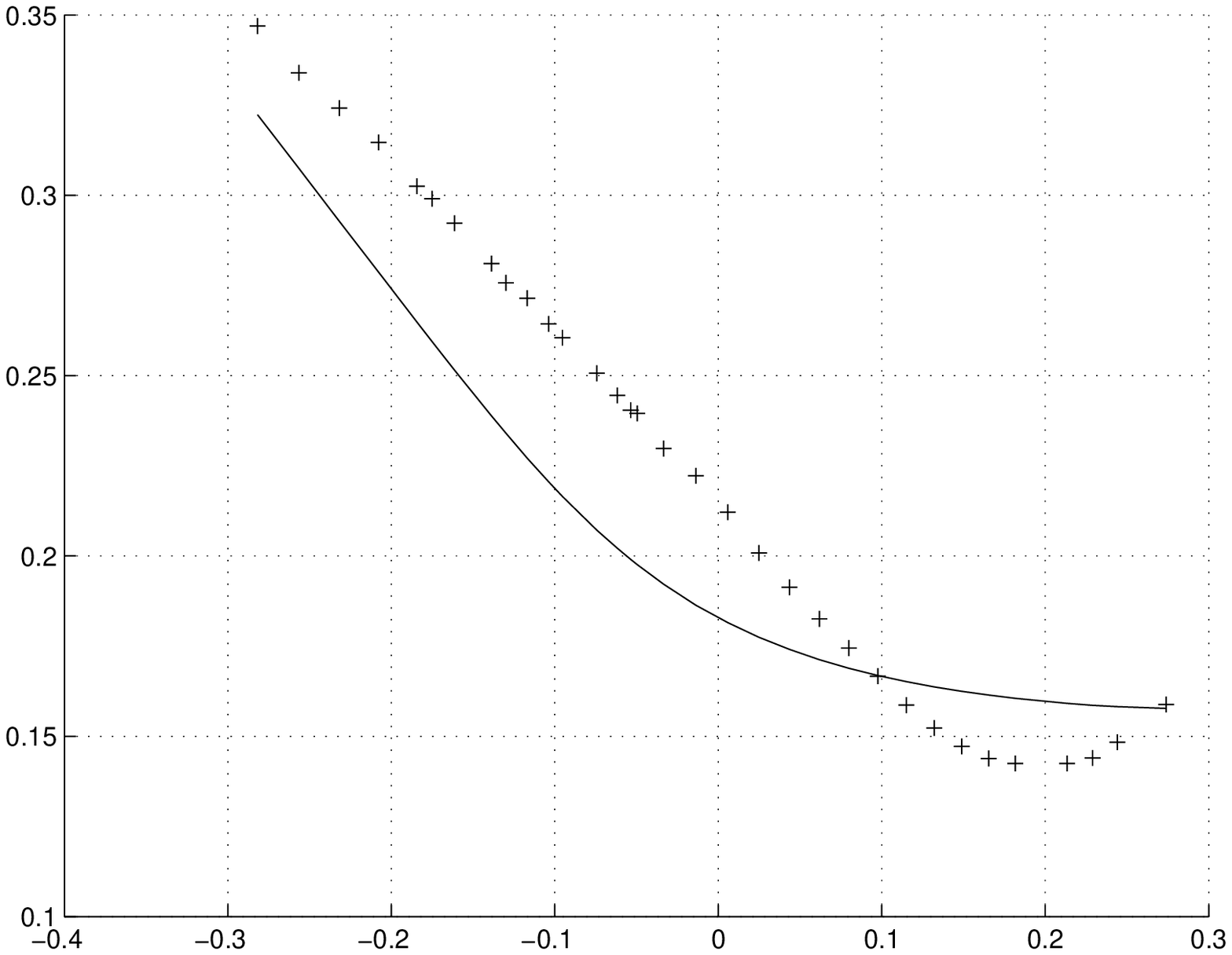} &
\includegraphics[width=.33\textwidth,height=.12\textheight]{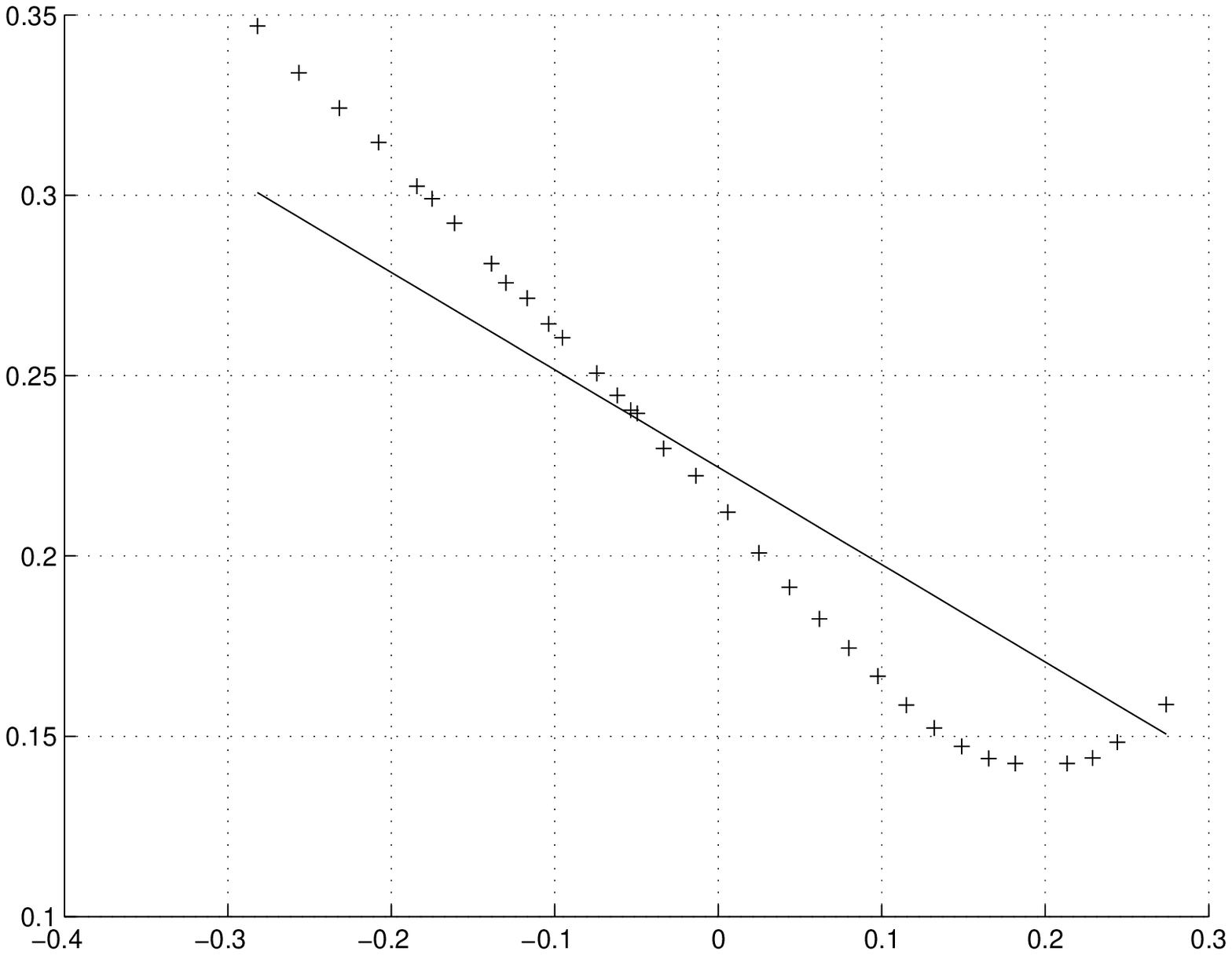} \\ \hline
251 days-to-maturity & 251 days-to-maturity &251 days-to-maturity  \\
\includegraphics[width=.33\textwidth,height=.12\textheight]{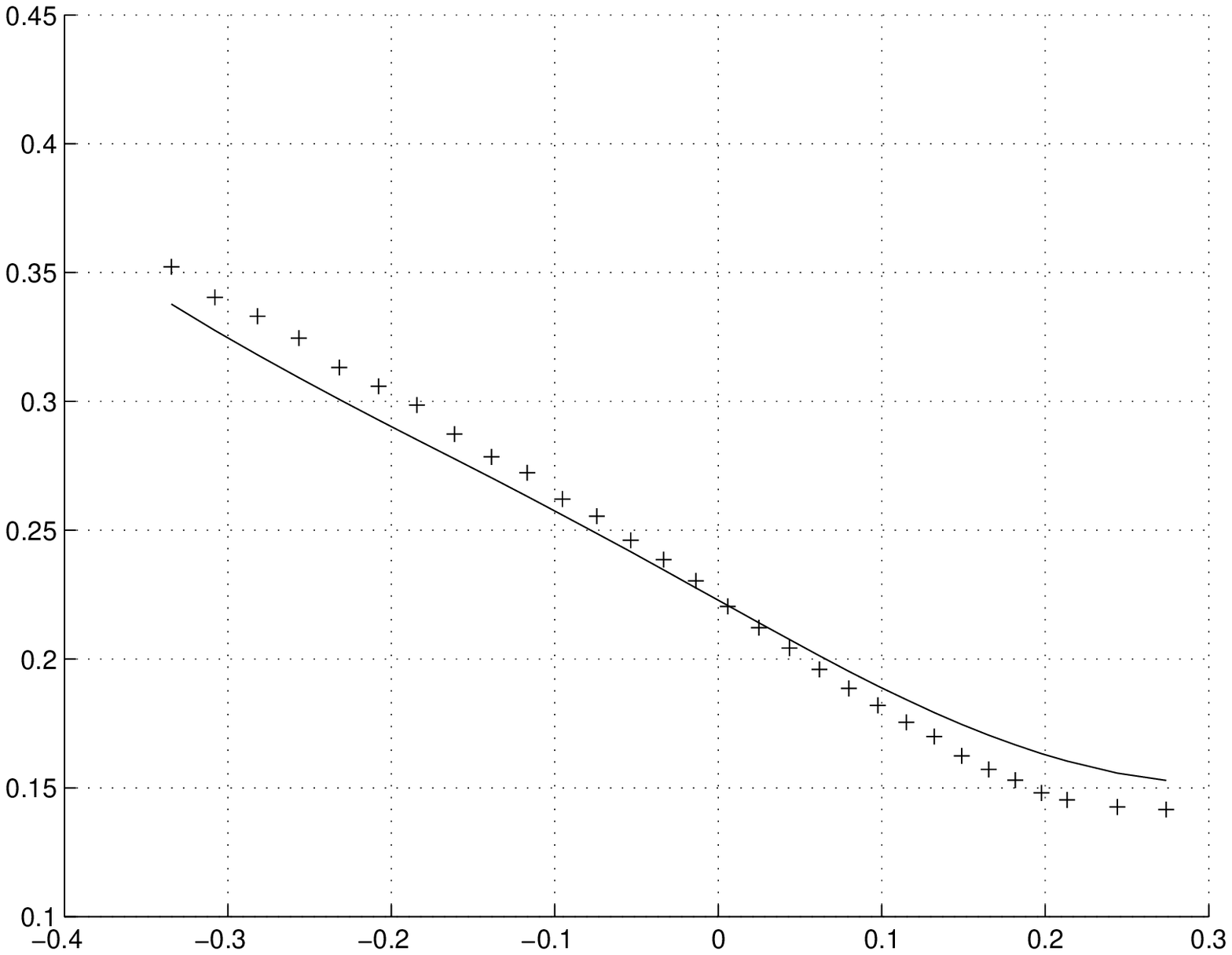} &
\includegraphics[width=.33\textwidth,height=.12\textheight]{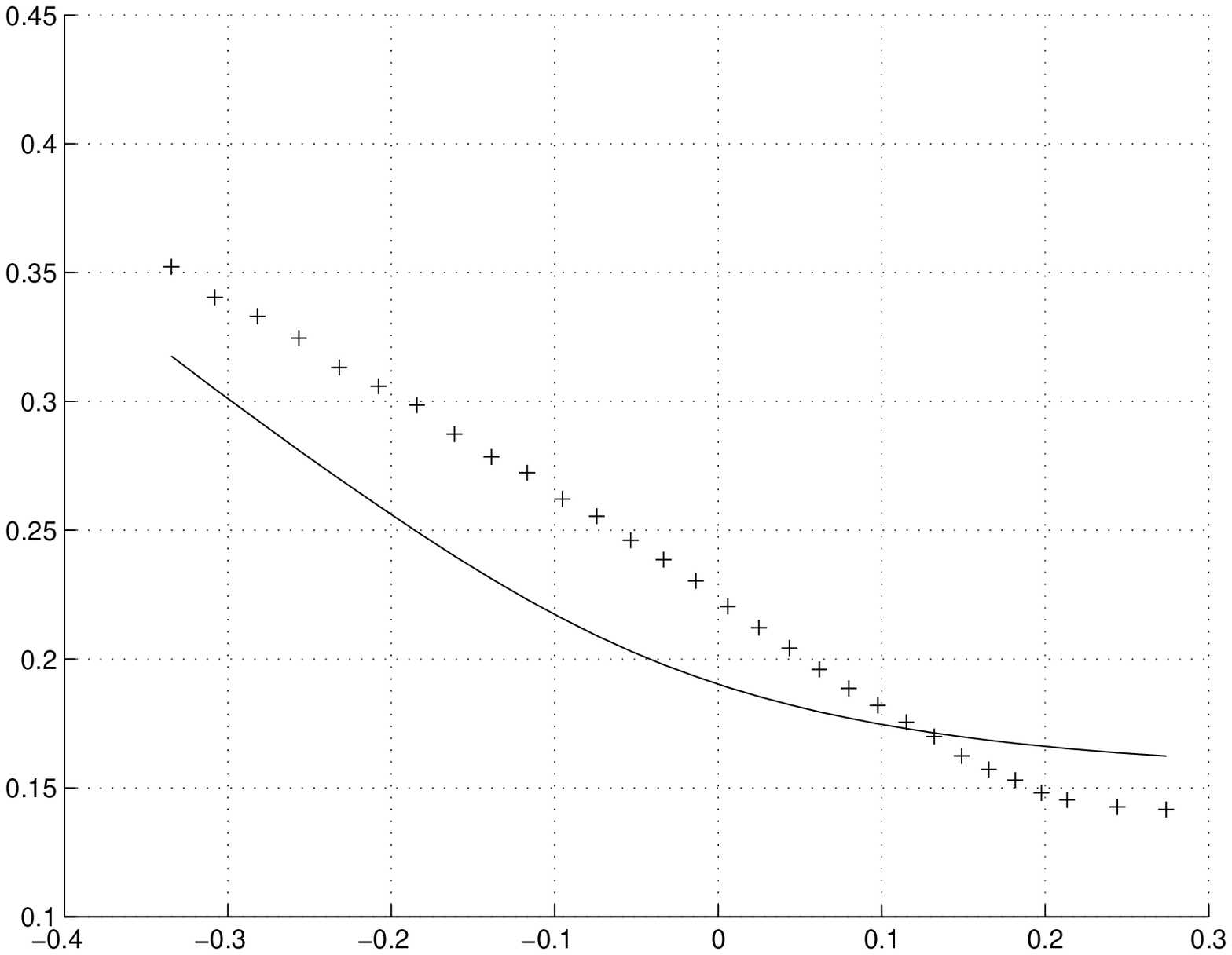} &
\includegraphics[width=.33\textwidth,height=.12\textheight]{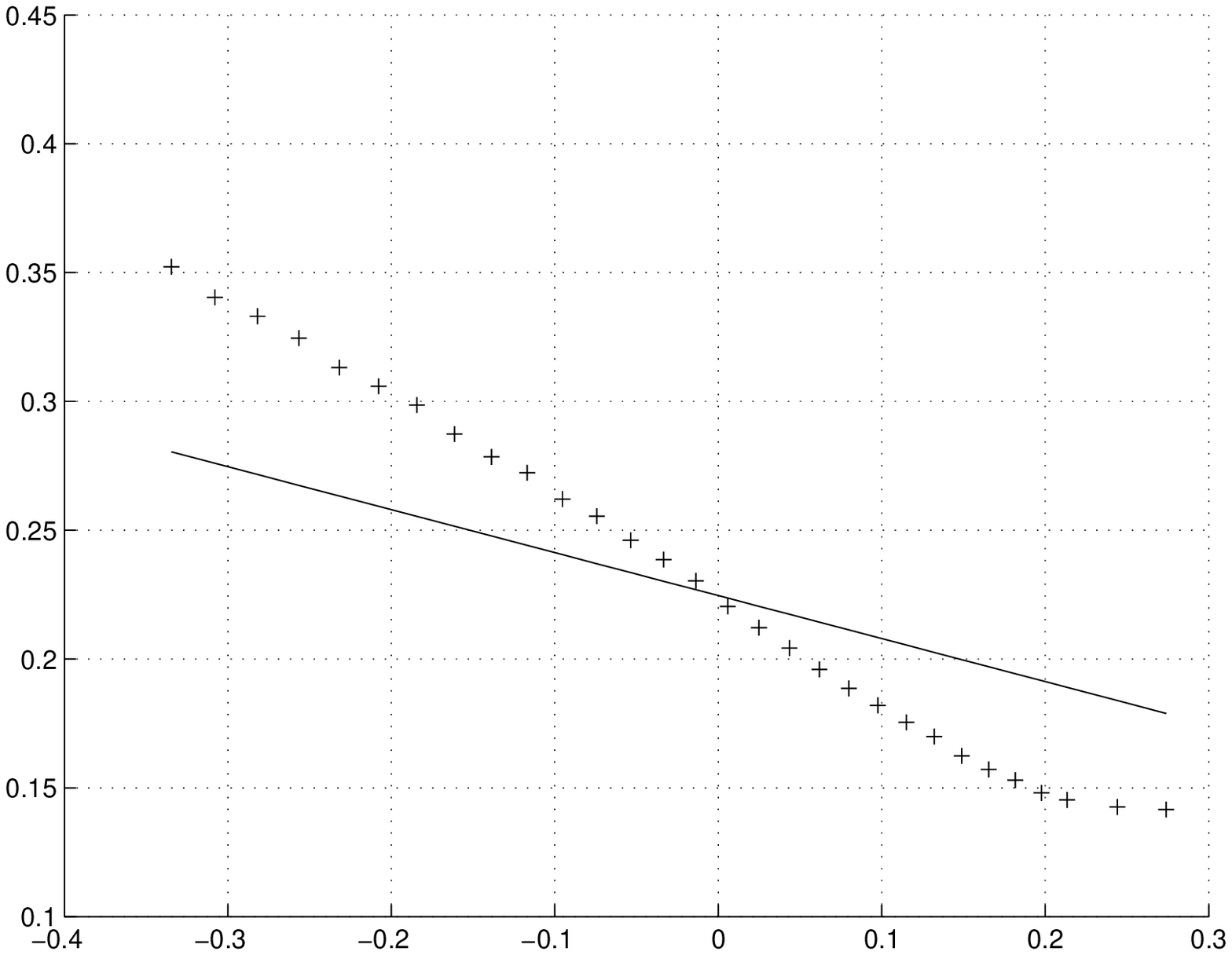} \\ \hline
341 days-to-maturity& 341 days-to-maturity & 341 days-to-maturity  \\
\includegraphics[width=.33\textwidth,height=.12\textheight]{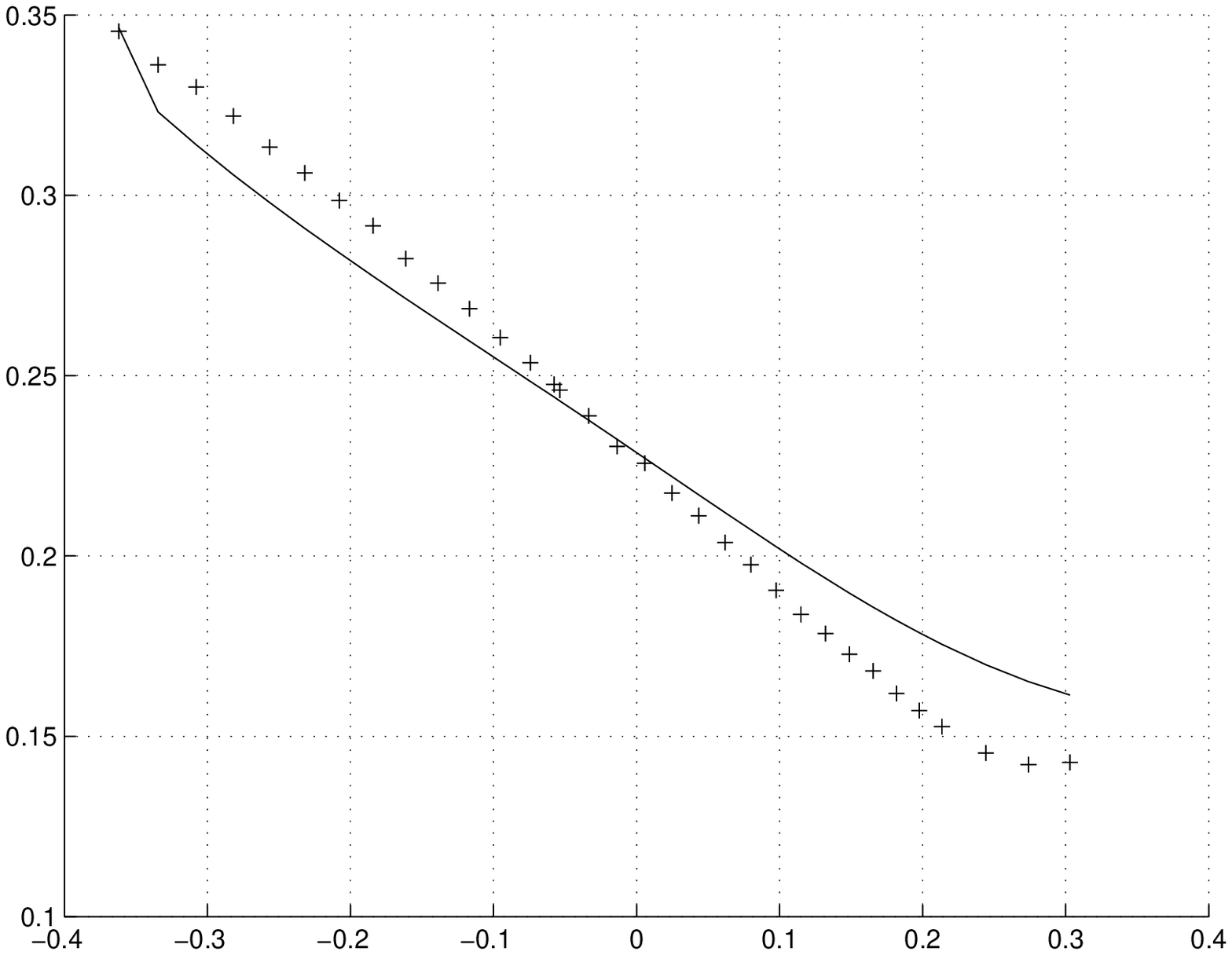} &
\includegraphics[width=.33\textwidth,height=.12\textheight]{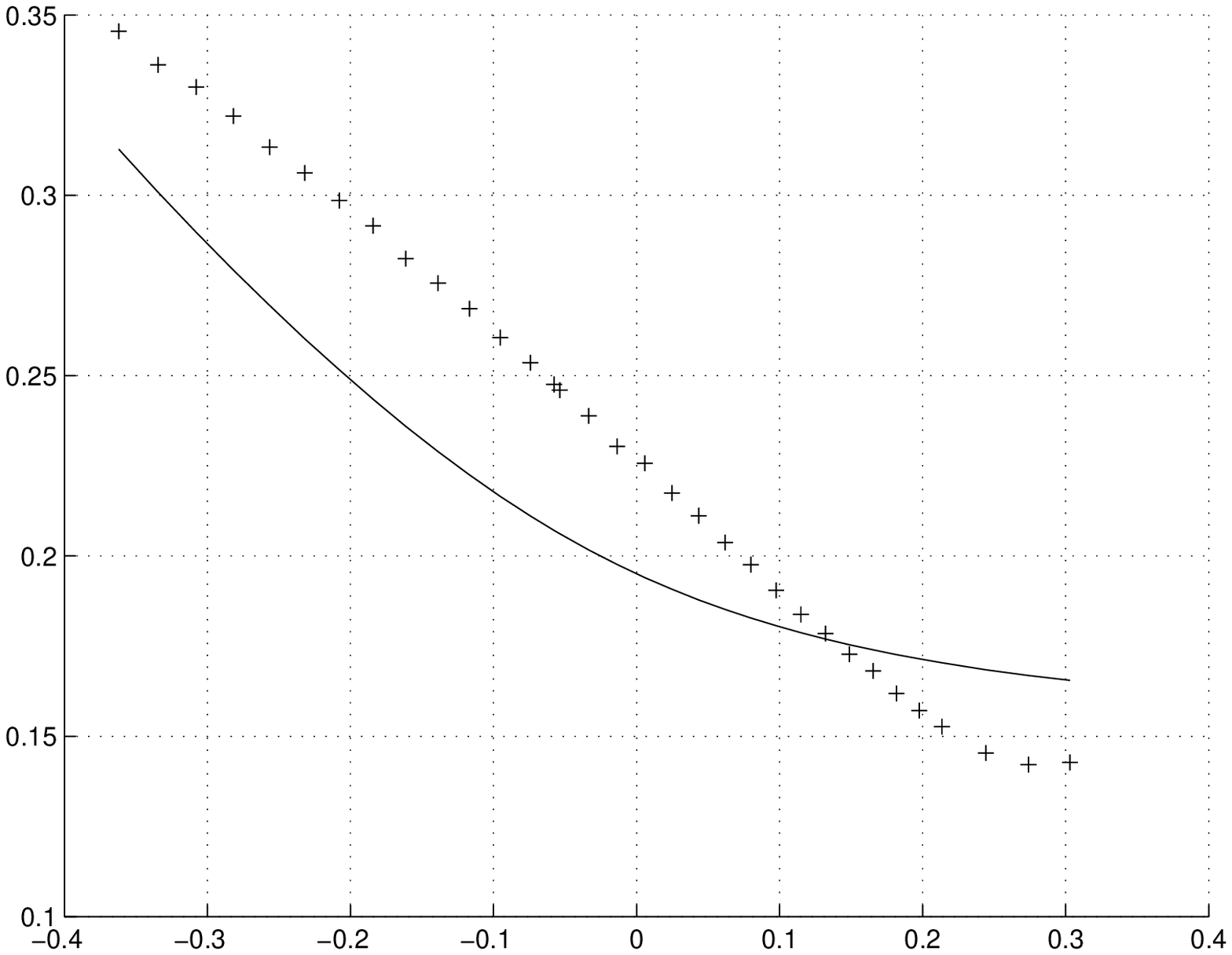} &
\includegraphics[width=.33\textwidth,height=.12\textheight]{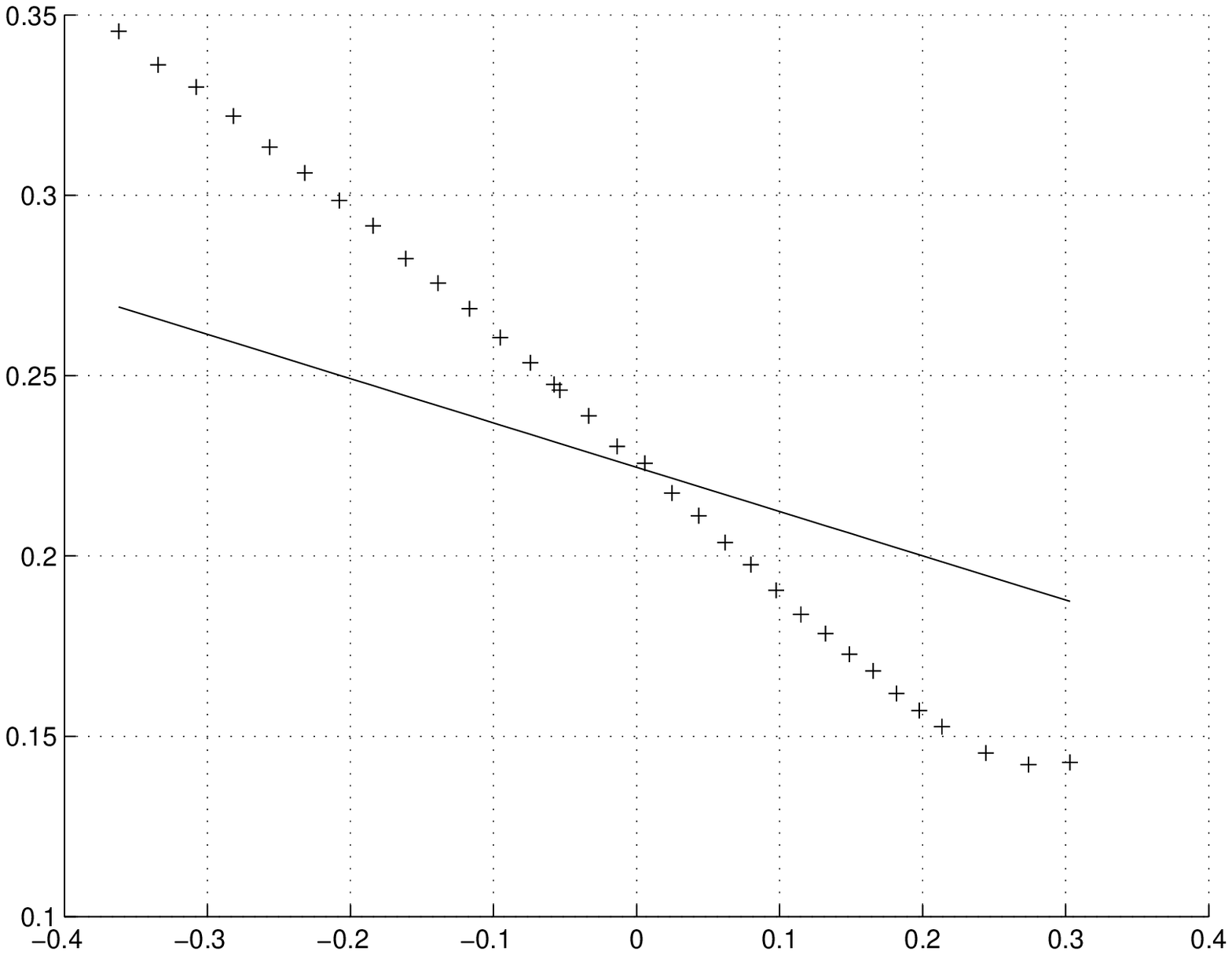} \\ \hline
\end{tabular}
\caption{Implied volatility fit to S\&P500 index options from January 11, 2012.}
\label{fig:merton2}
\end{figure}

%%%%%%%%%%%%%%%%%%%%%%%%%%%%%%%%%%%%%%%%%%%%%%%%%%%%%%%%%
%
%								Extension to Multiscale
%
%%%%%%%%%%%%%%%%%%%%%%%%%%%%%%%%%%%%%%%%%%%%%%%%%%%%%%%%%

\section{Extension to multiscale stochastic volatility and jump intensity}
\label{sec:multiscale}
The results of this paper can be extended in a straightforward manner to include \emph{multiscale} stochastic volatility and jump intensity.  We briefly describe how this may be done.  Our intent in this section is not to be rigorous, but rather to give a flavor of the computations involved in this extension.  To begin, we modify the dynamics of $S$ slightly.  Letting $S=e^X$ we have
\begin{align}
\left.
\begin{aligned}
dX_t
		&=		\gam(Y_t,Z_t) \, dt + \sig(Y_t,Z_t) \, d\Wt_t^x + \int_\Rb s \, d\Nt_t(Y_t,Z_t,ds) , &
X_0
		&= 		x , \\
dY_t
		&=		\( \frac{1}{\eps^2} \alpha(Y_t) - \frac{1}{\eps} \Lam(Y_t,Z_t) \, \beta(Y_t) \) dt + \frac{1}{\eps} \beta(Y_t) d\Wt_t^y , &
Y_0
		&=		y , \\
dZ_t
		&=		\Big( \del^2 c(Z_t) - \del \, \Gamma(Y_t,Z_t) \, g(Z_t) \Big) dt + \del \, g(Z_t) d\Wt_t^z , &
Z_0
		&=		z .
\end{aligned} \right\} \qquad (\text{under $\Pt$}) \label{eq:dXdYdZ}
\end{align}
Here, $Z$ is a \emph{slow-varying} factor, in the sense that its infinitesimal generator under $\Pb$ is scaled by $\del^2$, which is assumed to be a small parameter: $\del^2 << 1$.  The Brownian motions $\Wt^x$, $\Wt^y$, $\Wt^z$ have correlations $\rho_{xy}$, $\rho_{xz}$ and $\rho_{yz}$ (which must be such that the covariance matrix is positive definite), the compensated Poisson random measure $\Nt(Y,Z,ds)$ satisfies
\begin{align}
d\Nt_t(Y_t,Z_t,ds)
		&=		dN_t(Y_t,Z_t,ds) - \zeta(Y_t,Z_t)\nu(ds) dt , \\
\Et[dN_t(Y_t,Z_t,ds)|Y_t,Z_t]
		&=		\zeta(Y_t,Z_t) \nu(ds) dt ,
\end{align}
and the drift $\gamma(Y_t,Z_t)$ is given by
\begin{align}
\gam(Y_t,Z_t)
		&=		-\frac{1}{2}\sig^2(Y_t,Z_t) - \zeta(Y_t,Z_t) \int_\Rb (e^s-1-s ) \nu(ds) .
\end{align}
Using risk-neutral pricing, the value $u^{\eps,\del}(t,x,y,z)$ of a European option in this setting is
\begin{align}
u^{\eps,\del}(t,x,y,z)
		&=		\Et_{x,y,z} \[ h(X_t) \] , &
h(x)
		&:=		H(e^x) .
\end{align}
From the Kolmogorov backward equation, the function $u^{\eps,\del}$ satisfies the following PIDE and BC
\begin{align}
\(-\d_t + \Ac^{\eps,\del} \) u^{\eps,\del}
		&=		0 , &
u^{\eps,\del}(0,x,y,z)
		&=		h(x) , \label{eq:PIDE.2}
\end{align}
where the partial integro-differential operator $\Ac^{\eps,\del}$ is the generator of $(X,Y,Z)$.  The operator $\Ac^{\eps,\del}$ has the following form
\begin{align}
\Ac^{\eps,\del}
		&=		\frac{1}{\eps^2} \Ac_0 + \frac{1}{\eps} \Ac_1 + \Ac_2 + \frac{\del}{\eps} \Mc_3 + \del \, \Mc_1 + \del^2 \, \Mc_2 .
\end{align}
Terms containing $\del$ in \eqref{eq:PIDE.2} are small in the small-$\del$ limit, giving rise to a \emph{regular} perturbation.  Thus, \eqref{eq:PIDE.2} has the form of a combined singular-regular perturbation about the $\Oc(1)$ operator $(-\d_t + \Ac_2 )$.  Following \cite{fpss} we seek a solution $u^{\eps,\del}$ of the form
\begin{align}
u^{\eps,\del}
		&=		\sum_{n=0}^\infty \sum_{m=0}^\infty \eps^n \del^m u_{n,m} . \label{eq:uexpand.2}
\end{align}
Our goal is to find an approximation $u^{\eps,\del} = u_{0,0} + \eps \, u_{1,0} + \del \, u_{0,1} + \Oc(\eps^2 + \del^2)$.
A formal asymptotic analysis yields the following PIDEs for $u_{0,0}$, $u_{1,0}$ and $u_{0,1}$
\begin{align}
\Oc(1):&&
(-\d_t + \<\Ac_2\>)u_{0,0} 
		&=		0, &
u_{0,0}(0,x,z)
		&=		h(x) , \\
\Oc(\eps):&&
(-\d_t + \<\Ac_2\>)u_{1,0} 
		&=		- \Bc u_{0,0} , &
u_{1,0}(0,x,z)
		&=		0 , \\
\Oc(\del):&&
(-\d_t + \<\Ac_2\>)u_{0,1} 
		&=		- \<\Mc_1\> u_{0,0} , &
u_{0,1}(0,x,z)
		&=		0 ,
\end{align}
where, as in Section \ref{sec:asymptotics}, the $y$-dependence has disappeared from $u_{0,0}$, $u_{1,0}$ and $u_{0,1}$.  The operators $\< \Ac_2 \>$, $\Bc$ and $\< \Mc_1 \>$ are given by
\begin{align}
\< \Ac_2 \>
		&=		\< \gam(\cdot,z)\> \d_x + \frac{1}{2} \< \sig^2(\cdot,z) \> \d_{xx}^2 
					+ \< \zeta(\cdot,z)\> \int_\Rb \(e^{s \d_x } - 1 - s \d_x \) \nu(ds) , \\
\Bc
		&=		V_3(z) \( \d_{xxx}^3 - \d_{xx}^2 \) 
					+ U_3(z) \( - \int_\Rb \Big( e^s - 1 - s \Big) \nu(ds) \d_{xx}^2
						+ \int_\Rb \Big( \theta_s - 1 - s \d_x \Big) \d_x \nu(ds) \) \\ &\qquad
					+ V_2(z) \( \d_{xx}^2 - \d_x \)
					+ U_2(z) \( - \int_\Rb \Big( e^s - 1 - s \Big) \nu(ds) \d_x
						+ \int_\Rb \Big( \theta_s - 1 - s \d_x \Big) \nu(ds) \) , \\
\< \Mc_1 \>
		&=		- g(z) \< \Gamma(\cdot,z)\> \d_z + g(z) \rho_{xz} \< \sig(\cdot,z)\> \d_{xz}^2,
\end{align}
where the $z$-dependent parameters ($V_3(z)$, $U_3(z)$, $V_2(z)$, $U_2(z)$) are
\begin{align}
V_3(z)
		&= 		- \frac{\rho_{xy}}{2}\< \beta(\cdot) \sigma(\cdot,z) \d_y \eta(\cdot,z) \> , &
U_3(z)
		&=		- \rho_{xy} \< \beta(\cdot) \sigma(\cdot,z) \d_y \xi(\cdot,z) \> , \\
V_2(z)
		&=		\frac{1}{2} \< \beta(\cdot) \Lam(\cdot,z) \d_y \eta(\cdot,z) \> , &
U_2(z)
		&=		\< \beta(\cdot) \Lam(\cdot,z) \d_y \xi(\cdot,z) \> .		
\end{align}
The expressions for $u_{0,0}$ and $u_{1,0}$ are analogous to those given for $u_0$ and $u_1$ in Theorem \ref{thm:u0u1}.  An expression for $u_{0,1}$ is obtained using Fourier transforms
\begin{align}
\vh(s,\lam,z)
		&:=		\int_\Rb dx \frac{1}{\sqrt{2\pi}}e^{-i \lam x} \< \Mc_1 \> u_{0,0}(s,x,z) \\
u_{0,1}(t,x,z)		
		&=		\int_\Rb d\lam \,  \frac{1}{\sqrt{2\pi}}e^{i \lam x} 
					\int_0^t ds \, e^{(t-s)\phi_\lam(z)} \vh(s,\lam,z) .
\end{align}
Note, care must be taken when computing $\< \Mc_1 \>u_{0,0}$ as both terms in $\< \Mc_1 \>$ contain the operator $\d_z$ and $u_{0,0}$ depends on $z$ through both $\<\sig^2(\cdot,z)\>$ and $\<\zeta(\cdot,z)\>$.  A careful computation shows that $u_{0,1}$ is linear in the following four parameters
\begin{align}
V_1(z)
		&= 		 g(z) \rho_{xz} \< \sig(\cdot,z)\> \d_z\<\sig^2(\cdot,z)\> , &
V_0(z)
		&=		 - g(z) \< \Gamma(\cdot,z)\> \d_z\<\sig^2(\cdot,z)\>, \\
U_1(z)
		&=		 g(z) \rho_{xz} \< \sig(\cdot,z)\> \d_z\<\zeta(\cdot,z)\>, &
U_0(z)
		&=		 - g(z) \< \Gamma(\cdot,z)\> \d_z\<\zeta(\cdot,z)\>.
\end{align}
From the arguments in \cite{fpss}, Chapter 4, it follows that for fixed $(t,x,y)$ there exists a constant $C$ such that $|u^{\eps,\del} - (u_{0,} + \eps \, u_{1,0} + \del \, u_{0,1})| < C ( \eps^2 + \del^2)$ when $h$ is smooth.
%Finally, for smooth payoffs, the accuracy of the multiscale pricing approximation $u_{0,0} + \eps \, u_{1,0} + \del \, u_{0,1}$ is as follows: for fixed $(t,x,y,z)$ there exists a constant $C$ such that for any $\eps \leq 1$, $\del \leq 1$ we have
%\begin{align}
%\left| u^{\eps,\del} - \( u_{0,0} + \eps \, u_{1,0} + \del \, u_{0,1} \) \right| &\leq C ( \eps^2 + \del^2 ) .
%\end{align}
%The proof of this error bound is analogous to the proof found in Chapter 4 of \cite{fpss}.

%%%%%%%%%%%%%%%%%%%%%%%%%%%%%%%%%%%%%%%%%%%%%%%%%%%%%%%%%
%
%								Conclusion
%
%%%%%%%%%%%%%%%%%%%%%%%%%%%%%%%%%%%%%%%%%%%%%%%%%%%%%%%%%

\section{Conclusion}
In this paper, we have introduced a class of exponential L\'evy-type models in which the volatility and jump-intensity are driven stochastically by two factors -- one fast-varying and one slow-varying.  Using techniques from the theory of generalized Fourier transforms, singular and regular perturbation theory we have derived a general formula for the approximate price of any European-style derivative whose payoff function has a generalized Fourier transform.  We test five specific examples of exponential L\'evy-type models with stochastic volatility and jump-intensity (the Extended Merton, Gumbel, Dirac, Variance Gamma, and Uniform) and we show that these model classes provide a closer fit to the S\&P500 implied volatility surface than either the Merton model or the class of fast mean-reverting stochastic volatility models.  Other exponential L\'evy models can be extended in a similar fashion by choosing the appropriate L\'evy measure $\nu$.  We hope this work motivates further research into exponential L\'evy-type models.  A possible extension of this paper, for example, would be to allow the jump \emph{distribution} (rather than just the jump \emph{intensity}) to vary stochastically in time.

\subsection*{Acknowledgments}
The author would like to extend his sincerest thanks to Ramon van Handel, Ronnie Sircar, Jean-Pierre Fouque, Rama Cont, Jose-Luis Menaldi and Erhan Bayraktar, whose comments and suggestions have improved the quality and readability of this manuscript.

%%%%%%%%%%%%%%%%%%%%%%%%%%%%%%%%%%%%%%%%%%%%%%%%%%%%%%%%%
%
%								APPENDIX
%
%%%%%%%%%%%%%%%%%%%%%%%%%%%%%%%%%%%%%%%%%%%%%%%%%%%%%%%%%

\appendix

%%%%%%%%%%%%%%%%%%%%%%%%%%%%%%%%%%%%%%%%%%%%%%%%%%%%%%%%%
%
%								Fourier
%
%%%%%%%%%%%%%%%%%%%%%%%%%%%%%%%%%%%%%%%%%%%%%%%%%%%%%%%%%

\section{Proof the Theorem \ref{thm:u0u1}}
\label{sec:u0u1}
We wish to solve PIDEs \eqref{eq:PIDEu0} and \eqref{eq:PIDEu1} with BCs \eqref{eq:BCu0} and \eqref{eq:BCu1} respectively.  For simplicity, we solve these equations for a payoff $h \in L^1(\Rb,dx)$.  The results extend to any $h$ with a generalized Fourier transform in a straightforward manner.
\par
To begin, we recall that the \emph{Fourier transform} $\fh$ and \emph{inverse transform} of a function $f \in L^1(\Rb,dx)$ are defined as a pair
\begin{align}
\text{Fourier Transform}:&&
\fh(\lam)
	&:=		\frac{1}{\sqrt{2\pi}} \int_\Rb dx \, e^{-i \lam x} f(x) , & \lam \in \Rb , \\
\text{Inverse Transform}:&&
f(x)
	&=		\frac{1}{\sqrt{2\pi}}  \int d\lam \, e^{i \lam x} \fh(\lam) .
\end{align}
Next, we introduce $\<\Ac_2\>^*$, the formal adjoint of $\<\Ac_2\>$, which satisfies
\begin{align}
\int_\Rb dx \, e^{-i\lam x} \< \Ac_2 \> f(x) dx
	&=	\int_\Rb dx \, f(x)  \< \Ac_2 \>^* e^{-i \lam x} dx .
\end{align}
The operator $\< \Ac_2 \>^*$ can be obtained through by integration by parts, which leads to
\begin{align}
\< \Ac_2 \>^*
		&=		- \< \gam\> \d_x + \frac{1}{2} \< \sig^2 \> \d_{xx}^2 + \<\zeta\> \int_\Rb \(\theta_{-z} - 1 + z \d_x \) \nu(dz) .
\end{align}
We note that
\begin{align}
\<\Ac_2\>^*e^{-i\lam x}
	&= \phi_\lam e^{ -i\lam x}, &
\Bc e^{i\lam x}
	&=	B_\lam e^{i \lam x} ,
\end{align}
where $\phi_\lam$ and $B_\lam$ are given by \eqref{eq:phi} and \eqref{eq:B.lambda} respectively.  To find an expression for $u_0(t,x)$ we Fourier transform PIDE \eqref{eq:PIDEu0} and BC \eqref{eq:BCu0}.  We have, 
\begin{align}
\d_t u_0(t,x)
	&= \< \Ac_2 \> u_0(t,x) & &\Rightarrow &
\frac{1}{\sqrt{2\pi}}  \int_\Rb dx \, e^{-i\lam x} \d_t u_0(t,x)
	&=	\frac{1}{\sqrt{2\pi}}  \int_\Rb dx \, e^{-i\lam x} \< \Ac_2 \> u_0(t,x)  \\ & & & &
	&=	\frac{1}{\sqrt{2\pi}}  \int_\Rb dx \, u_0(t,x) \< \Ac_2 \>^{*} e^{-i\lam x}  \\ & & & &
	&=	\phi_\lam \frac{1}{\sqrt{2\pi}} \int_\Rb dx \, u_0(t,x) e^{-i\lam x} \\ & & &\Rightarrow &
\d_t \uh_0(t,\lam)
	&=	\phi_\lam \uh_0(t,\lam) , \label{eq:uhatODE} \\
u_0(0,x)
	&=	h(x)  & &\Rightarrow &
\uh_0(0,\lam)
	&=	\hh(\lam) . \label{eq:uhatBC}
\end{align}
Note that \eqref{eq:uhatODE} is an ODE in $t$ for $\uh_0(t,\lam)$ with an initial condition \eqref{eq:uhatBC}.  Thus, one deduces
\begin{align}
\uh_0(t,\lam)
	&=	e^{t \phi_\lam} \hh(\lam) & &\Rightarrow &
u_0(t,x)
	&=	\frac{1}{\sqrt{2 \pi}} \int_\Rb d\lam e^{t \phi_\lam} \hh(\lam) e^{i \lam x} ,
\end{align}
which established \eqref{eq:u0.FT}.  Next, to find an expression for $u_1(t,x)$ we first observe that
\begin{align}
\Bc u_0(t,x)
	&=		\frac{1}{\sqrt{2 \pi}} \int_\Rb d\mu \, e^{t \phi_\mu} \hh(\mu) \Bc   e^{i \mu x} 
	=			\frac{1}{\sqrt{2 \pi}} \int_\Rb d\mu \, e^{t \phi_\mu} \hh(\mu)  B_\mu e^{i \mu x} .
\end{align}
Therefore, we have
\begin{align}
\frac{1}{\sqrt{2\pi}} \int_\Rb dx \, e^{-i\lam x} \Bc u_0(t,x)
	&=	\int_\Rb d\mu \, e^{t \phi_\mu} \hh(\mu)  B_\mu
			\frac{1}{2\pi} \int_\Rb dx \, e^{-i (\lam - \mu) x} \\
	&=	\int_\Rb d\mu \, e^{t \phi_\mu} \hh(\mu)  B_\mu \del(\mu - \lam) \\
	&=	e^{t \phi_\lam} \hh(\lam)  B_\lam ,
\end{align}
where we have used the Fourier representation of a Dirac delta function: $\frac{1}{2\pi} \int_\Rb dx \, e^{-i(\lam - \mu )x} = \del(\lam-\mu)$.
Fourier Transforming PIDE \eqref{eq:PIDEu1} and BC \eqref{eq:BCu1} one finds
\begin{align}
(-\d_t + \phi_\lam ) \uh_1(t,\lam)
	&=	- e^{t \phi_\lam} \hh(\lam)  B_\lam &\text{and}& &
\uh_1(0,\lam)
	&=	0 .
\end{align}
Once again, have an (inhomogeneous) ODE in $t$ for $\uh_1(t,\lam)$.  Solving the ODE explicitly for $\uh_1(t,\lam)$ and inverse transforming from $\lam$ to $x$ yields.
\begin{align}
\uh_1(t,\lam)
	&=	t e^{t \phi_\lam} \hh(\lam) B_\lam &\Rightarrow& &
u_1(t,x)
	&=	\frac{1}{\sqrt{2 \pi}} \int_\Rb d\lam \, t e^{t \phi_\lam} \hh(\lam) B_\lam e^{i \lam x} ,
\end{align}
which establishes \eqref{eq:u1.FT}.  This completes the proof.

\bibliographystyle{chicago}
%\bibliography{../../BibTeX-Master/BibTeX-Master}	
\bibliography{BibTeX-Master}

%%%%%%%%%%%%%%%%%%%%%%%%%%%%%%%%%%%%%%%%%%%%%%%%%%%%
%
%			Figures
%
%%%%%%%%%%%%%%%%%%%%%%%%%%%%%%%%%%%%%%%%%%%%%%%%%%%%

\clearpage
\begin{landscape}
{\footnotesize
\captionof{table}{Calibration results for January 4, 2010,  October 1, 2010 and January 11, 2012. The first date encompasses the following maturities: 47, 75, and 103 days. The second date contains the maturities 50, 78 and 113 days. The third contains the maturities 66, 10, 155, 251, 341 days.}
\begin{tabular}{ | c | c | c | c | c | c | c | c | c | c | c | c | c | c |}
\hline
\textbf{2010/01/04} & & & & & & & & & DTM:47-103\\ \hline\hline
FMRSV &$\langle \sigma^2\rangle$ & & & & $V_3$ & & $V_2$ & & RMSE\\ \hline
& $0.1981^2$& & & & $-7\cdot 10^{-4}$& & $10^{-5}$ & & 0.0189 \\\hline
Merton &$\langle \sigma^2\rangle$ &$\langle \zeta \rangle $ & $m$ & $s^2$ & & &  & &RMSE\\ \hline
$\nu(dz)=(2\pi s^2)^{-1/2}\exp\left(-\frac{(z-m)^2}{2s^2}\right)$&$0.1409^2$&	0.1708&	$-$0.1708&	$0.1371^2$ & & & & & 0.0144\\ \hline
ExtendedMerton &$\langle \sigma^2\rangle$ &$\langle \zeta \rangle $ & $m$ & $s^2$ & $V_3^\varepsilon$ & $U_3^\varepsilon$& $V_2^\varepsilon$ &$U_2^\varepsilon$ & RMSE\\ \hline
$\nu(dz)=(2\pi s^2)^{-1/2}\exp\left(-\frac{(z-m)^2}{2s^2}\right)$ & $0.1346^2$&	0.653	&$-$0.208&	$0.1729^2$&		$-2\cdot 10^{-5}$	&0.0072	&0.001	&$-$0.0058&	0.0049\\ \hline\hline
\textbf{2010/10/01} & & & & &  & & & &DTM: 50-113\\ \hline\hline
FMRSV &$\langle \sigma^2\rangle$ & & & & $V_3$ & & $V_2$ & & RMSE\\ \hline
& $0.2222^2$& & & & $-0.0011$& & $10^{-5}$ & & 0.0222 \\\hline
Merton &$\langle \sigma^2\rangle$ &$\langle \zeta \rangle $ & $m$ & $s^2$ & & &  & & RMSE\\ \hline
$\nu(dz)=(2\pi s^2)^{-1/2}\exp\left(-\frac{(z-m)^2}{2s^2}\right)$&$0.1538^2$&	0.6617&	$-$0.1895&	$0.1525^2$ & & & & & 0.0183\\ \hline
ExtendedMerton &$\langle \sigma^2\rangle$ &$\langle \zeta \rangle $ & $m$ & $s^2$ & $V_3^\varepsilon$ & $U_3^\varepsilon$& $V_2^\varepsilon$ &$U_2^\varepsilon$ & RMSE\\ \hline
$\nu(dz)=(2\pi s^2)^{-1/2}\exp\left(-\frac{(z-m)^2}{2s^2}\right)$ & $0.1437^2$&	0.5187	&$-$0.295&	$0.2159^2$&		$-3\cdot 10^{-5}$	&0.0132	&$-$0.0027&0.0142&	0.0052\\ \hline \hline
\textbf{2012/01/11} & & & & & & & & & DTM:66-341\\ \hline\hline
FMRSV &$\langle \sigma^2\rangle$ & & & & $V_3$ & & $V_2$ & & RMSE\\ \hline
& $0.35^2$& & & & $0.0049$& & $0.0513$ & & 0.0373 \\\hline
Merton &$\langle \sigma^2\rangle$ &$\langle \zeta \rangle $ & $m$ & $s^2$ & & &  & &RMSE\\ \hline
$\nu(dz)=(2\pi s^2)^{-1/2}\exp\left(-\frac{(z-m)^2}{2s^2}\right)$&$0.1422^2$&	0.2319&$-0.2858$&	$0.1705^2$ & & & & & 0.0298\\ \hline
ExtendedMerton &$\langle \sigma^2\rangle$ &$\langle \zeta \rangle $ & $m$ & $s^2$ & $V_3^\varepsilon$ & $U_3^\varepsilon$& $V_2^\varepsilon$ &$U_2^\varepsilon$ & RMSE\\ \hline
$\nu(dz)=(2\pi s^2)^{-1/2}\exp\left(-\frac{(z-m)^2}{2s^2}\right)$ & $0.1073^2$&	0.7606&$-0.1988$&$0.1543^2$	&$-6.0052\cdot 10^{-5}$	&0.0010	&$6.8192\cdot 10^{-4}$&$-3.045\cdot 10^{-4}$&	0.0125\\ \hline\hline

\end{tabular}
\label{tab:stats}
}
\end{landscape}

\clearpage
\begin{landscape}
{\footnotesize
\captionof{table}{Calibration results for December 19, 2011. The dataset contains the following maturities: 59, 88, 122, 177, 273 and 363 days.}
\begin{tabular}{ | c | c | c | c | c | c | c | c | c | c | c | c | c | c |}
\hline
\textbf{2011/12/19} & & & & & & & & & & DTM:59-363\\ \hline\hline
FMRSV &$\langle \sigma^2\rangle$ & & & & & $V_3$ & & $V_2$ & & RMSE\\ \hline
& $0.35^2$& & & & & $0.004811$& & $0.03717$ & & 0.0278 \\\hline
Merton &$\langle \sigma^2\rangle$ &$\langle \zeta \rangle $ & $m$ & $s^2$ & & & &  & & RMSE\\ \hline
$\nu(dz)=(2\pi s^2)^{-1/2}\exp\left(-\frac{(z-m)^2}{2s^2}\right)$&$0.1529^2$&	1.3720&	$-$0.1397&	$0.1141^2$& & & & & & 0.0215\\ \hline
ExtendedMerton &$\langle \sigma^2\rangle$ &$\langle \zeta \rangle $ & $m$ & $s^2$ & & $V_3^\varepsilon$ & $U_3^\varepsilon$& $V_2^\varepsilon$ &$U_2^\varepsilon$ & RMSE\\ \hline
$\nu(dz)=(2\pi s^2)^{-1/2}\exp\left(-\frac{(z-m)^2}{2s^2}\right)$ & $0.2054^2$&0.8207	&$-$0.5608&	$0.4070^2$&&		$-5.61729\cdot 10^{-4}$	&0.3254	&$-$0.1263&$-$0.1549&	0.0072\\ \hline
Gumbel &$\langle \sigma^2\rangle$ &$\langle \zeta \rangle $ & $m$ & $s^2$ & & & &  & & RMSE\\ \hline
$\nu(dz)=\frac{1}{\sigma}\exp\left(\left(\frac{x-m}{\sigma}\right)-\exp\left(\frac{x-m}{\sigma}\right)\right)$&$0.0705^2$&	5.3221&	$-$0.1875&	$0.0756^2$& & & & & & 0.0202\\ \hline
ExtendedGumbel &$\langle \sigma^2\rangle$ &$\langle \zeta \rangle $ & $m$ & $s^2$ & & $V_3^\varepsilon$ & $U_3^\varepsilon$& $V_2^\varepsilon$ &$U_2^\varepsilon$ & RMSE\\ \hline
$\nu(dz)=\frac{1}{\sigma}\exp\left(\left(\frac{x-m}{\sigma}\right)-\exp\left(\frac{x-m}{\sigma}\right)\right)$ & $0.0717^2$&6.2521&$-$0.1875&	$0.0856^2$&&		$-1.2232\cdot 10^{-5}$	&$-$4.1826$\cdot 10^{-6}$&6.6619$\cdot 10^{-5}$&1.3969$\cdot 10^{-6}$&	0.0116\\ \hline
Dirac &$\langle \sigma^2\rangle$ &$\langle \zeta \rangle $ & $a$ & & & & &  & & RMSE\\ \hline
$\nu(dz)=\delta_a(dz)$&$0.1418^2$&	1.5924&	$-$0.1810&	& & & & & & 0.0212\\ \hline
ExtendedDirac &$\langle \sigma^2\rangle$ &$\langle \zeta \rangle $ & $a$ & & & $V_3^\varepsilon$ & $U_3^\varepsilon$& $V_2^\varepsilon$ &$U_2^\varepsilon$ & RMSE\\ \hline
$\nu(dz)=\delta_a(dz)$ & $0.1398^2$&1.3523&$-$0.1768&	&&		$-2.7651\cdot 10^{-5}$	&1.2326$\cdot 10^{-5}$&4.5501$\cdot 10^{-5}$&1.8360$\cdot 10^{-5}$&	0.0123\\ \hline
Variance Gamma &$\langle \sigma^2\rangle$ &$\langle \zeta \rangle $ & $a$ &$b$ &$B$ & & &  & & RMSE\\ \hline
$\nu(dz)=\frac{e^{-az}}{z}\mathbf{1}\{z>0\}+B\frac{e^{bz}}{-z}\mathbf{1}\{z<0\}$&$0.0510^2$&	0.6783&	35.3325&11.4922	&13.6786 & & & & & 0.0215\\ \hline
Extended Variance Gamma &$\langle \sigma^2\rangle$ &$\langle \zeta \rangle $ & $a$ &$b$ &$B$ & $V_3^\varepsilon$ & $U_3^\varepsilon$& $V_2^\varepsilon$ &$U_2^\varepsilon$ & RMSE\\ \hline
$\nu(dz)=\frac{e^{-az}}{z}\mathbf{1}\{z>0\}+B\frac{e^{bz}}{-z}\mathbf{1}\{z<0\}$ & $0.0922^2$&0.1196&267.4499&14.8657	&69.8231&	$3\cdot 10^{-3}$	&$-$0.0066 & 0.0011&$-7\cdot 10^{-4}$&	0.0147\\ \hline
Uniform &$\langle \sigma^2\rangle$ &$\langle \zeta \rangle $ & $a$ &$b$ & & & &  & & RMSE\\ \hline
$\nu(dz)=\frac{1}{b-a}\mathbf{1}\{a\leq x\leq b\}$&$0.0922^2$&3.9644&	$-$0.2086&0.0588	& & & & & & 0.0214\\ \hline
ExtendedUniform &$\langle \sigma^2\rangle$ &$\langle \zeta \rangle $ & $a$ & & & $V_3^\varepsilon$ & $U_3^\varepsilon$& $V_2^\varepsilon$ &$U_2^\varepsilon$ & RMSE\\ \hline
$\nu(dz)=\frac{1}{b-a}\mathbf{1}\{a\leq x\leq b\}$ & $0.1405^2$&4.0001&$-$0.1009&0.0997	&&		$-4.74943\cdot 10^{-4}$	&$1.1575\cdot 10^{-5}$&0.0078&$-5.3749\cdot 10^{-4}$&0.0142\\ \hline
\end{tabular}
}
\end{landscape}
\end{document}